\documentclass[]{kkluwer}
\usepackage[footnotesize]{caption}
\usepackage{afterpage}
\usepackage{float}
\usepackage{subfig}
\usepackage{lscape}
\usepackage{graphics,graphicx,epsfig,pifont,url,times,sgame,amsfonts,amsmath,color}
\usepackage[dvips,all,color,line]{xy}
\usepackage{enumerate}
\usepackage[vcentering,hcentering]{geometry}
\usepackage{soul}
\newcommand{\note}[1]{}
\providecommand{\note}[1]{~\\\frame{\begin{minipage}[c]{\textwidth}\vspace{2pt}\center{#1}\vspace{2pt}\end{minipage}}\vspace{3pt}\\}

\date{November 15, 2008}

\newfloat{zfigure}{tbph}{lof}
\restylefloat*{zfigure}
\floatstyle{plain}
\floatname{zfigure}{Figure}
\newcommand{\nothing}{}
\makeatletter
\def\kfigure{%
   \let\@makecaption\nothing%
   \global\@floatcaption={}%
   \let\@oldfigcaption\@klu@figcaption%
   \let\@klu@figcaption\nothing%
   \@ifnextchar[{\f@gurewithoptions}{\f@gurewithoptions[tbp]}}%
\def\endkfigure{
       \vspace{-22pt}
       \egroup
       \@getindent \gdef\cap@type{0}%
       \hfuzz=\floatindent
       \if@kaprotate \@setrotatedfigbox \pagebreak
       \else \@setnotrotatedfigbox \fi
       \let\centerline\saved@centerline
       \if@fixedfloat \vskip\intextsep \@fixedfloatfalse
       \else \end@float \fi
       \hfuzz=0.1pt
       \let\@klu@figcaption\@oldfigcaption
       }%
\makeatother

\newcommand{\E}[1]{\mathbb{E}\left[#1\right]}

\newcommand{\rv}{RV$_{\sigma(t)}$}

\newcommand{\agent}[1]{\texttt{#1}}
\newcommand{\zgame}[1]{\textit{#1}}

\newcommand{\awesome}{\agent{AWE\-SOME}}
\newcommand{\Awesome}{\agent{AWE\-SOME}}
\newcommand{\determined}{\agent{det\-er\-min\-ed}}
\newcommand{\Determined}{\agent{Det\-er\-min\-ed}}
\newcommand{\fp}{\agent{fict\-it\-ious play}}
\newcommand{\FP}{\agent{Fict\-it\-ious play}}
\newcommand{\giga}{\agent{GIGA-\-WoLF}}
\newcommand{\gsa}{\agent{GSA}}
\newcommand{\meta}{\agent{meta}}
\newcommand{\Meta}{\agent{Meta}}
\newcommand{\minimax}{\agent{mini\-max-Q}}
\newcommand{\Minimax}{\agent{Mini\-max-Q}}
\newcommand{\minimaxidr}{\agent{mini\-max-Q-IDR}}
\newcommand{\Minimaxidr}{\agent{Mini\-max-Q-IDR}}
\newcommand{\q}{\agent{Q-learn\-ing}}
\newcommand{\Q}{\agent{Q-Learn\-ing}}
\newcommand{\random}{\agent{ran\-dom}}
\newcommand{\Random}{\agent{Ran\-dom}}
\newcommand{\rvs}{\agent{\rv}}

\newtheorem{ob}{Observation}

\newtheorem{theorem}{Theorem}[section]
\newtheorem{lem}[theorem]{Lemma}

\newenvironment{proof}[1][Proof]{\begin{trivlist}
\item[\hskip \labelsep {\bfseries #1}]}{\end{trivlist}}

\newcommand{\var}[1]{Var\left[#1\right]}
\newcommand{\cov}[1]{Cov\left[#1\right]}

\newenvironment{packed_list}{
\begin{itemize}
  \setlength{\itemsep}{1pt}
  \setlength{\parskip}{0pt}
  \setlength{\parsep}{0pt}
}{\end{itemize}}

\begin{document}
\begin{opening}
\title{Empirically Evaluating Multiagent Learning Algorithms}

\author{Erik \surname{Zawadzki}}
\author{Asher \surname{Lipson}}
\author{Kevin \surname{Leyton-Brown}}
\institute{Department of Computer Science, University of British Columbia, Vancouver, Canada\\\{epz, alipson, kevinlb\}@cs.ubc.ca}


\begin{abstract}
There exist many algorithms for learning how to play repeated bimatrix games. Most of these algorithms are justified in terms of some sort of theoretical guarantee. On the other hand, little is known about the empirical performance of these algorithms. Most such claims in the literature are based on small experiments, which has hampered understanding as well as the development of new multiagent learning (MAL) algorithms.
We have developed a new suite of tools for running multiagent experiments: the MultiAgent Learning Testbed (MALT). These tools are designed to facilitate larger and more comprehensive experiments by removing the need to build one-off experimental code. MALT also provides baseline implementations of many MAL algorithms, hopefully eliminating or reducing differences between algorithm implementations and increasing the reproducibility of results.
Using this test suite, we ran an experiment unprecedented in size. We analyzed the results according to a variety of performance metrics including reward, maxmin distance, regret, and several notions of equilibrium convergence. We confirmed several pieces of conventional wisdom, but also discovered some surprising results. For example, we found that single-agent $Q$-learning outperformed many more complicated and more modern MAL algorithms.

\end{abstract}
\end{opening}

\section{Introduction}
Urban road networks, hospital systems and commodity markets are all examples of complicated multiagent systems that are essential to everyday life. Indeed, any social interaction can be seen as a multiagent problem. As a result of the prominence of multiagent systems, a lot of attention has been paid to designing and analyzing learning algorithms for multiagent environments. A multitude of different algorithms exist for a variety of different settings. Some prominent examples include algorithms by \inlinecite{littman94}, \inlinecite{singh00}, \inlinecite{hu03}, \inlinecite{greenwald03}, \inlinecite{bowling04}, \inlinecite{powers04}, \inlinecite{banerjee06}, and \inlinecite{conitzer07}.

We take the position that the best multiagent learning (MAL) algorithm is the one that achieves the highest possible average reward.\footnote{For alternatives, see \inlinecite{shoham06}---who called the approach that we espouse the ``prescriptive, non-cooperative agenda''---or \inlinecite{sandholm2007}.} Under this view, the problem faced by the designer of a MAL algorithm is qualitatively the same as the problem faced by the designer of a single-agent reinforcement learning algorithm. However, there is a fundamental difference between the two settings. In the stationary environment faced by classical reinforcement learners, the concept of an optimal policy is well defined, and hence learning algorithms can attempt to identify this policy. In a multiagent environment, the best policy to follow depends on the actions taken by the opponent, and thus on the ways in which the opponent's future behavior will be affected by the learner's present actions. The best policy depends on the opponent's strategy, and so there can be no global ``optimum.''

It is this added conceptual complexity that makes MAL problems interesting; however, it has also made them harder to analyze. Theoretical claims about MAL algorithms generally do not speak directly about average reward. Instead, they tend to describe alternative aspects of the algorithm's performance that are intended to `stand in' for reward. Some work has insisted that algorithms should converge to stage-game Nash equilibria, or should do so at least in the case of ``self play.'' Others have insisted on other sorts of convergence properties or on regret bounds. Still others have offered different guarantees for performance against different classes of opponents.

Because many MAL algorithms are incomparable on the basis of their theoretical properties, and further because the extent to which these various properties correlate with an algorithm's ability to achieve high average reward in practice is unclear, it is generally argued that MAL algorithms should be compared empirically. Many such experimental comparisons have been performed in the literature (see, e.g., \cite{nudelman04,powers04}). 
However, for the most part these experiments have been designed to advocate for a newly-designed algorithm rather than to survey the whole landscape. As a consequence, most of these experiments have been small in terms of the number of game instances and opposing algorithms considered. Furthermore, different experiments have in many cases measured performance in different ways, making it difficult to compare their results and draw an overall conclusion. There is therefore considerable opportunity to expand our understanding of how existing MAL techniques compare in practice.

Part of the reason for the relative paucity of large-scale empirical work is that neither a centralized algorithm repository nor a standardized test setup exists. This is unfortunate, not only because considerable work has to be invested in designing one-off testbeds and reimplementing algorithms, but also because centralized and public repositories increase reproducibility and decrease the danger that different experiments will achieve different results because of differences in implementations. Publicly available and scrutinized implementations offer the promise of experiments that are easier to run, reproduce, and compare.

In this article we make two main contributions. First, we describe the design and implementation of a platform for running MAL experiments (\S\ref{sec:platform}). This platform offers several advantages over one-off setups. We hope that it will facilitate new and larger-scale empirical work.

Our second main contribution is the analysis of such an empirical study. This experiment is, to our knowledge, unprecedented in terms of scale. We make suggestions about how empirical MAL performance data should be analyzed (\S\ref{sec:experiment}), and offer a detailed discussion of different algorithms' average reward in practice (\S\ref{sec:result}). Furthermore, we draw connections between different performance metrics that have been explored in theoretical work (\S\ref{sec:other-metrics}), and show that some of the least sophisticated algorithms achieve extremely competitive performance.

\section{Algorithms and Past Experimental Work}

MAL algorithms have been studied for over half a century. This rich investigation has produced not only a profusion of competing algorithms but also various distinct problem formulations. Does an algorithm know the game's reward functions before the game starts, or do reward functions need to be learned? How many opponents can an algorithm face? What signals about the opponent's actions can an algorithm observe? Can an algorithm rely on being able to determine stage-game Nash equilibria or other computationally-expensive game properties? Each of these assumptions changes the learning problem.

In this section we describe the algorithms we study in this paper, and also survey past experimental evaluations of MAL algorithms. The creators of the algorithms that we describe answered the above questions in different ways, reflecting the community's broader disagreement about precisely what problem MAL algorithms should aim to solve. In order to permit the study of a broad range of algorithms, we have answered the above questions permissively: we allow algorithms access to the reward functions, to signals about the opponent's actions, and to computationally-costly game properties. Thus, we are able to compare algorithms that require this information to others that are capable of learning it. 


The other important experimental choice we faced was the class of games upon which to evaluate algorithms. We chose to restrict ourselves to 2-player repeated games. (Note, however, that we do not restrict the number of actions in the repeated game.) We chose this setting instead of $n$-player repeated games or either $2$- or $n$-player stochastic games for two reasons. First, the case of two-player repeated games has received the most past study (though see e.g., \cite{vu06}). Second, considerably more work has been done to identify experimentally-interesting test data for this case. We restricted our attention to algorithms that can play two-player games of any size and with any payoff structure. We thus did not make use of work that insists (e.g.) on two-action games \cite{singh00} or constant-sum games \cite{littman94}. We also mention as an aside that MAL experiments have been conducted in settings that are neither generalizations nor restrictions of our setting, such as the population-based work by \inlinecite{axelrod87} and \inlinecite{airiau2007}.


\subsection{Fictitious Play}

\FP\ \cite{brown51} is probably the earliest example of a learning algorithm for two-player games repeated games. Essentially, \fp\ assumes that the opponent is playing an unknown and potentially mixed stationary strategy, and tries to estimate this strategy from the opponent's empirical distribution of actions---the frequency counts for each of its actions normalized to be probabilities. Clearly, in order to collect the frequency counts \fp\ must be able to observe the opponent's actions. The algorithm then, at each iteration, best responds to this estimated strategy. Because \fp\ needs to calculate a best response, it also assumes complete knowledge of its own payoffs.

Fictitious play is guaranteed to converge to a Nash equilibrium in self play for a restricted set of games. These games are said to have the \textit{fictitious play property} (see, for instance \inlinecite{monderer96b}; for an example of a simple $2 \times 2$ game without this property see \inlinecite{monderer96}). \FP\ will also eventually best respond to any stationary strategy. This algorithm's general structure has been extended in a number of ways, including \textit{smooth fictitious play} \cite{fudenberg93}, and we will see later that \fp\ provides the foundation for several more modern algorithms. 

\FP\ is known to be subject to miscoordination problems, particularly in self play, and particularly in games that reward asymmetric coordination (e.g., dispersion games).
There are some clever measures that can be taken to avoid some of these kinds of problems (e.g., best response tie-breaking rules and randomization), but miscoordination remains a general problem for the \fp\ approach.



\subsection{Determined}

\Determined\, or `bully' (see, for example, \inlinecite{powers04}) is an algorithm that solves the multiagent learning problem by ignoring it. MAL algorithms typically change their behavior by adapting to signals about the game. However, \determined, as its name suggests, simply relies on other algorithms to adapt their strategies to it.

\Determined\ enumerates the stage-game Nash equilibria and selects the one that maximizes its personal reward in equilibrium; then, it plays its corresponding action forever.\footnote{We can imagine variations on the \determined\ idea that do not play an action from a Nash equilibrium. For example, a variant could instead choose the action whose best response yields the algorithm the highest payoff. Note that this differs from a stage-game Nash equilibrium because this \determined-like algorithm need not itself play a best response. Such an outcome amounts to an equilibrium of the Stackelberg version of the stage game. That is, we can change the game so that instead of the two players moving simultaneously, the \determined-like agent moves first.} Certainly, \determined\ can lead to some obvious problems. For instance, in self play two \determined\ agents can stubbornly play actions from different equilibria, leading to sub-equilibrium average reward. Additionally, enumerating all the Nash equilibria not only requires complete knowledge of every agents' reward function, but is also computationally costly, limiting the use of this strategy to relatively small stage games. All the same, \determined\ serves as a useful baseline for comparison. Also, slight variations of this algorithm are, like \fp, at the heart of some more modern algorithms.

%

\subsection{Targeted Algorithms}
\label{sec:targeted}

We next focus on two so-called \textit{targeted} algorithms, which focus on playing against particular classes of opponents.
Both these algorithms are based around identifying what the opponent is doing (with particular attention paid to stationarity and Nash equilibrium), and then updating their behavior based on this assessment.

\awesome\ \cite{conitzer03, conitzer07} tracks the opponent's behavior in different periods of play and tries to maintain hypotheses about its play. For example, \awesome\ attempts to determine whether the other algorithm is playing a particular stage-game Nash equilibrium. If it is, \awesome\ responds with its own component of that special equilibrium. This special equilibrium is known in advance by all implementations of \awesome\ to avoid equilibrium selection issues in self play. There are other situations where it acts in a similar fashion to \fp, and there are still other discrete modes of play that it engages in depending on its beliefs.

\Meta\ \cite{powers04} switches between three simpler strategies: a strategy similar to \fp, a \determined-style algorithm that stubbornly plays a Nash equilibrium, and the maxmin strategy. Strategy selection depends on the recorded history of average reward and the empirical distribution of the opponent's actions across different periods of play. \Meta\ was shown both theoretically and empirically to be nearly optimal against itself, close to the best response against stationary agents, and to approach (or exceed) the security level of the game in all cases.

Because both of these algorithms switch between simpler strategies depending on the situation, they can be viewed as portfolio algorithms. Note that both manage similar portfolios that include a \determined-style algorithm and a \fp\ algorithm.

\subsection{Q-learning Algorithms}

A broad family of MAL algorithms are based on \q{} \cite{watkins92}, which is a algorithm for finding the optimal policy in Markov Decision Processes (MDPs). This family of MAL algorithms does not explicitly model the opponent's strategy choices. They instead settle for learning the expected discounted reward for taking an action and then following a stationary policy encoded in the $Q$-function. In order to learn the $Q$-function, algorithms typically take random exploratory steps with a small (possibly decaying) probability.

Each algorithm in this family has a different way of selecting its strategy based on this $Q$-function. For instance, one could try a straightforward adaptation of single-agent \q\ to the multiagent setting by ignoring the impact that the opponent's action makes on the protagonist's payoffs. The algorithm simply updates its reward function whenever a new reward observation is made, where the new estimate is a convex combination of the old estimate and the new information: \begin{equation}
Q(a_i) = (1-\alpha_t)Q(a_i) + \alpha_t\left[r + \gamma \max_a{Q(a)}\right].\label{eqn:q}
\end{equation} This algorithm essentially considers the opponent's behavior to be an unremarkable part of a noisy and non-stationary environment. The non-stationarity of the environment makes learning difficult but this idea is not entirely without merit: \q\ has been shown to work in other non-stationary environments (see, for instance, \inlinecite{sutton99}).

\Minimax\ \cite{littman94} is one of the first explicitly multiagent applications of the $Q$-learning idea. The $Q$-function that it learns is based on the action profile and not just the protagonist's action: it learns $Q(a_i, a_{-i})$. Minimax-Q uses the mixed maxmin strategy calculated from the $Q$-function as its strategy:
\begin{equation}
Q(a_i, a_{-i})	= {}   (1-\alpha_t)Q(a_i, a_{-i}) + \alpha_t\left[r + \gamma \max_{\sigma_i \in \prod{(A_i)}}\left[\min_{a_{-i} \in A_{-i}}\sum_{a_i}\sigma_i(a_i)Q(a_i, a_{-i})\right]\right].\label{eqn:minmax_q}
\end{equation}
Such a strategy is sensible to the extent that the protagonist believes that the opponent aims to minimize his payoff, or that the protagonist cares about worst-case guarantees. It should be noted that since its maxmin strategies are calculated from learned $Q$-values, they may not be the game's actual maxmin strategies and thus fail to reflect the security value. Like \q, \minimax\ also takes the occasional exploration step.

There are further modifications to this general scheme. \agent{Nash-Q}\ \cite{hu03} learns different $Q$-functions for itself and its opponents and plays a stage-game Nash equilibrium strategy for the game in\-duced by these $Q$\--values. \agent{Co\-rre\-lated-Q}\ \cite{greenwald03} does something similar except that it chooses from the set of correlated equilibria using a variety of different selection methods. Both of these algorithms assume that they are able to observe not only the opponents' actions but also their rewards, and additionally that they have the computational wherewithal to compute the necessary solution concept.

\subsection{Gradient Algorithms}\label{gradient-algs}

Gradient ascent algorithms, such as \giga\ \cite{bowling04} and \rvs\ \cite{banerjee06}, maintain a mixed strategy that is updated in the direction of the payoff gradient. The specific details of this updating process depend on the individual algorithms, but the common feature is that they increase the probability of actions with high reward and decrease the probability of unpromising actions. This family of algorithms is similar to \q\ because they do not explicitly model their opponent's strategies and instead treat them as part of a non-stationarity environment.

\giga\ is the latest algorithm in a line of gradient learners that started with \agent{IGA} \cite{singh00}. \giga\ uses an adaptive step length that makes it more or less aggressive about changing its strategy. It compares its strategy to a baseline strategy and makes the update larger if it is performing worse than the baseline. \giga\ guarantees non-positive regret in the limit (regret is discussed in greater detail in \S\ref{sec:regret}) and strategic convergence to a Nash equilibrium when playing against \agent{GIGA} \cite{zinkevich03} in two-player two-action games.

There are two versions of \giga. The first version assumes prior knowledge of personal reward and the ability to observe the opponent's action---this is the version used in the proofs for \giga's no-regret and convergence guarantees. There is also a second version---on which all the experiments were based---that makes limited assumptions about payoff knowledge and computational power. Instead, like \q, it merely assumes that it is able to observe its own reward.

\rvs\ \cite{banerjee06} belongs to a second line of gradient algorithms that started with \agent{ReD\-VaL\-eR} \cite{banerjee04}. This algorithm also uses an adaptive step size when following the payoff gradient, like \giga, but does so on an action-by-action basis. This means that, unlike \giga, \rvs\ can be aggressive in updating some actions while being cautious about updating others. These updates are performed by comparing current reward to the reward at a Nash equilibrium. Therefore, \rvs\ requires complete information about the game and sufficient computational power to discover at least one stage-game Nash equilibrium. \rvs\ also guarantees no-regret in the limit and additionally provides some convergence results for self play in a restricted class of games.

\giga\ and \rvs\ differ in the way that they ensure that their updated strategies remain valid probability distributions. \giga\ \emph{retracts}: it maps an unconstrained vector to the vector on the probability simplex that is closest in $\ell_2$ distance. This approach has a tenancy to map vectors to extreme points of the simplex, reducing some action probabilities to zero. \rvs\ \emph{normalizes}, which is less prone to removing actions from its support. 

\subsection{Previous Experimental Results}

\begin{table}
			\footnotesize{
		\begin{tabular}{lllllll}
		\tiny{\textbf{Paper}}			& \tiny{\textbf{Algorithms}} 	& \tiny{\textbf{Distributions}} & \tiny{\textbf{Instances}} & \tiny{\textbf{Runs}} & \tiny{\textbf{Iterations}}\\
		\hline

		\inlinecite{littman94}		& 6		&	1		&	1					&	3 						 & $1 \times 10^5$ 					\\
		\inlinecite{claus97}			& 2		& 3		&	1 - 100	  & 1 - 100							& 50-2500						\\
		\inlinecite{greenwald03} 	& 7 	& 5		& 1					& 2500 - 3333		& $1 \times 10^5$		\\
		\inlinecite{bowling04b}		& 2		&	6		&	1					& ?							 & $1 \times 10^6$		\\
		\inlinecite{nudelman04}		&	3		& 13	&	100				&	10						 &	$1 \times 10^5$		\\
		\inlinecite{powers04}			& 11	& 21	& ?					&	?							 & $2 \times 10^5$		\\
		\inlinecite{banerjee06}		& 2		& 1		&	1					&	1							 &	16000							\\
		\inlinecite{conitzer07} 	& 3 	& 2		& 1					& 1							& 2500							\\
		\hline
		\end{tabular}}
		\caption[Previous experiments]{This table shows a summary of the experimental setup for a selection of papers. The summary includes the number of algorithms, the number of game distributions, the number of game instances drawn from these distributions, the number of runs or trials for each instance, and the number of iterations that the simulations were run for. In some cases, the setup was unclear, indicated  with a `?'. In many cases, fewer than $\left[Algorithms \times Distributions \times Instances \times Runs\right]$ runs were simulated, due to some sparsity in the experimental structures.}	
		\label{tab:setup}		
\end{table}

As discussed in the introduction, surprisingly little past work has aimed primarily to use large-scale experiments to compare the performance of MAL algorithms. Nevertheless, a considerable number of papers from the literature describe experimental comparisons, often in the context of arguing for a particular MAL algorithm or approach. We briefly survey that literature here.

Setting up a general-sum repeated two-player game experiment requires a number of design choices. What set of algorithms should be considered? On what set of games should these algorithms be run? If one is dealing with randomized algorithms (which includes any algorithm that is able to submit a mixed strategy), how many different runs should be simulated? For a particular game, for how many iterations should a simulation be run? As can be seen in Table~\ref{tab:setup}, experiments from the literature varied in all of these dimensions. Additionally, some papers do not describe all experimental parameters, making it difficult to compare results.

Overall, most of the tests performed in these papers considered few algorithms. In most cases, a newly proposed algorithm was evaluated by playing against one or two opponents. Some papers superficially appear to have used many algorithms, but in fact considered algorithms that varied only in small details. For example, in \inlinecite{littman94} two versions of \minimax\ and two versions of \q\ were tested, with each version differing only in  its training regime. In \inlinecite{greenwald03}, four versions of \agent{Correlated-Q} were tested against \q\ and \agent{Friend-Q} and \agent{Foe-Q} \cite{littman01}. \agent{Foe-Q} is the same as \minimax.

To our knowledge, the experiment that considered the greatest variety of algorithms was \inlinecite{powers04}. While four of the eleven algorithms tested in this study were simple stationary-strategy baselines, the remaining seven were MAL algorithms including \agent{Hyper-Q} \cite{tesauro04}, \agent{WoLF-PHC} \cite{bowling02b}, and a joint action learner \cite{claus97}.

Previous experiments have tended to investigate only small numbers of game instances, and these instances have tended to come from an even smaller number of game distributions. For example, \inlinecite{banerjee06} used only a single $3 \times 3$ action ``simple coordination game'' and \inlinecite{littman94} probed algorithm behavior with a single grid-world version of soccer. Initially, this limitation was partly due to the difficulty of creating a large number of diverse game instances. However with the creation of GAMUT \cite{nudelman04}, a suite of game generators, generating large game sets is now easy. Indeed, \inlinecite{nudelman04} also performed one of the largest previous MAL experiments, using three MAL algorithms (\minimax, \agent{WoLF} \cite{bowling01:1}, and \q) on $100$ game instances from each of thirteen distributions. Some recent papers have also leveraged GAMUT, such as \inlinecite{powers04}.

Finally, previous experiments have differed substantially in the number of iterations considered, ranging from 50 \cite{claus97} to $1 \times 10^6$ \cite{bowling04b}. Iterations in a repeated game are typically divided into ``settling in'' (also called a ``burn-in'' period) and ``recording'' phases, allowing the algorithms time to converge to stable behavior before results are recorded. \inlinecite{powers04} recorded the final $20~000$ of $200~000$ iterations and \inlinecite{nudelman04} used the final $10~000$ of $100~000$ iterations.

\label{sec:related_work}

\section{Platform}
\label{sec:platform}
The empirical experiments just described were generally conducted using one-off code tailored to the investigation of a particular feature of a given algorithm. This experimental design has a number of negative consequences. First, it decreases the reproducibility of experiments by, for instance, obscuring the details of algorithm implementation. Even when source code for the original experiment is available, its special-purpose nature can make it difficult to repurpose \note{Asher:Is repurpose the right word here? Do you mean generalize?} for follow-on studies or new experiments. Finally, rewriting similar code again and again wastes time that could be spent running more comprehensive experiments.

In this section, we describe our solution to this problem: an open and reusable platform called MALT (MultiAgent Learning Testbed) 2.0. It is available for free download at {\small{\url{http://www.cs.ubc.ca/~kevinlb/malt}}}. \note{If this is the right URL, I need new stuff to put on this page. If not, we should get the new URL. Either way, we should point the appropriate URL to a page that is in my web space (so it will never go away) but that you can edit.} \note{EPZ: Let's push this conversation to email}
This platform is designed for running two-player, general-sum, repeated-game MAL experiments. Basic visualization and analysis features are also included, as is support for running experiments using a computer cluster. Version 1.0 of MALT was introduced by \inlinecite{lipson05}; the version described here is a complete reimplementation of that work in a faster programming language (Java vs. Matlab), offering a wide variety of new features, bug fixes, and efficiency gains. Overall, we hope that other researchers will see MALT not as a finished product, but as a growing repository of tools, algorithms and experimental settings, and that they will use it as a base upon which to build (e.g., for the study of $N$-player repeated games or stochastic games). We have worked hard to make MALT easily extensible. For example, adding a new algorithm to the MALT GUI is as simple as providing a text file with a list of parameters, and adding an algorithm to the engine requires very little coding beyond the implementation of the algorithm itself.

\subsection{Definitions}

We now define some terms. An ordered pair of two algorithms is a \textit{pairing}. This pair is ordered because many two-player games are asymmetric: the payoff structure for the row player is different than the payoff structure for the column player. The case where an algorithm is paired with a copy of itself (but with different internal states and independent random seeds) is called \textit{self play}.

We concentrate on drawing games from distributions called \textit{game generators}. A particular sample from a game generator is a \textit{game instance}. \zgame{Prisoner's Dilemma} is a game generator and an example game instance is a particular set of payoffs that obey the \zgame{Prisoner's Dilemma} preference ordering. Other game generators are more heterogeneous; for example, one that we will discuss later samples from the space of all strategically distinct $2 \times 2$ games.

A pairing and a game instance, taken together, are called a \textit{match}. A  match with one of the algorithms in the pairing left unspecified is a \textit{partially specified match} (PSM). If two algorithms play the same PSM, we conclude that any differences between their performances are due to the algorithms themselves (including any internal randomization) because all else is held constant.

A particular simulation of a match is called a \textit{run} or \textit{trial}. For pairs of deterministic algorithms, a single run is sufficient to characterize a match; for randomized algorithms (including any algorithm that plays a mixed strategy) multiple runs may each yield different behavior. In such cases, the match must be characterized by a solution quality distribution (SQD)---the empirical distribution of a performance metric.\footnote{We use the term SQD because it is standard in the empirical study of algorithms. We note nevertheless that in MAL there is no clear notion of a game having a `solution', and that these distributions might be more meaningfully called `metric distributions'.} Each run consists of a number of \textit{iterations}. In each iteration, the algorithms select strategies and then receive some feedback: e.g., their reward; the action choice of their opponent. Algorithms are allowed to select mixed strategies; in this case, a single action is sampled from the mixing distribution by the game. The iterations are separated into \textit{settling-in iterations} and \textit{recorded iterations}.

\subsection{Platform Structure}

\begin{zfigure}[t]
	\centering
		\includegraphics[width=0.6\textwidth]{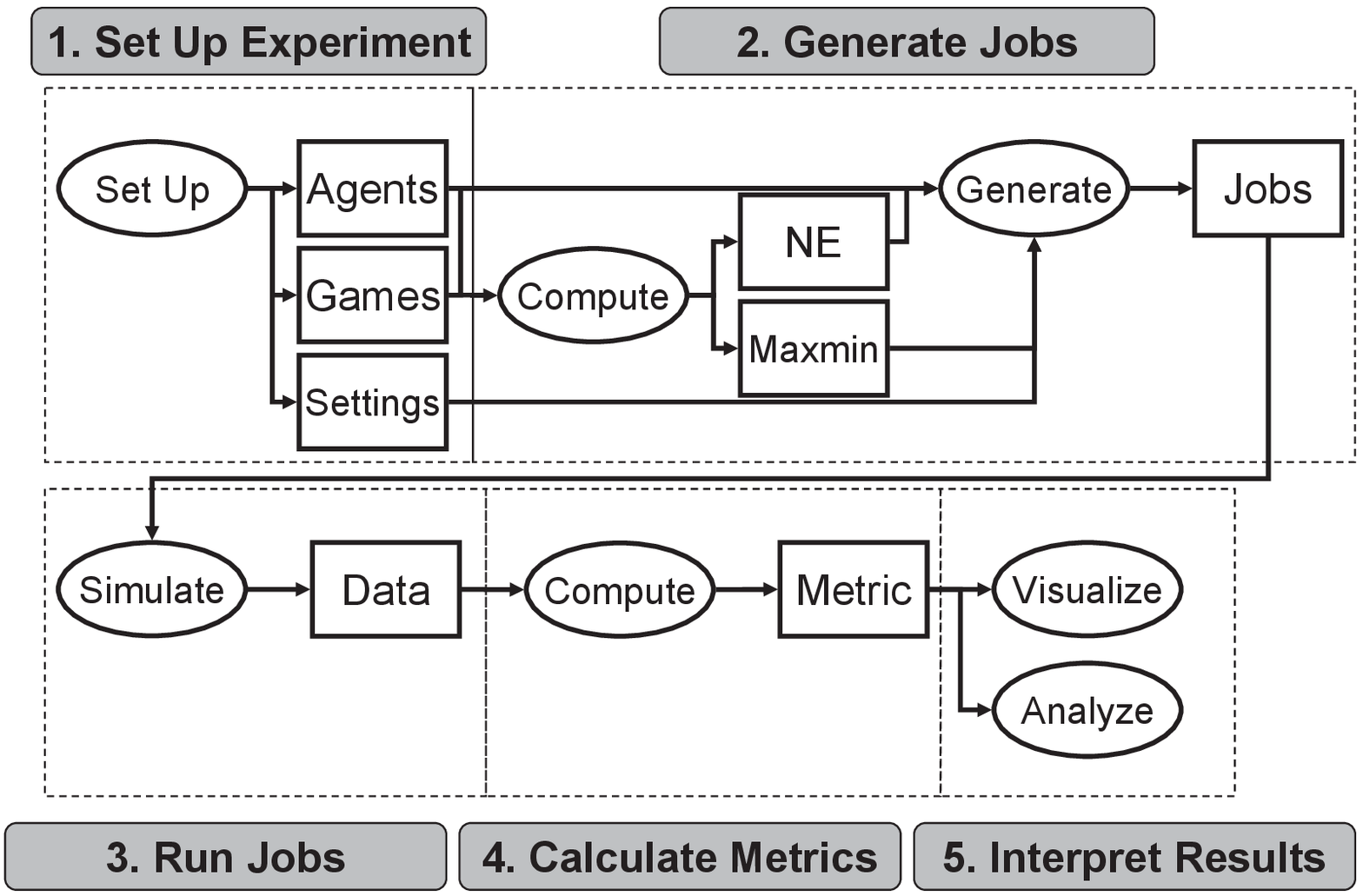}
	\caption{The five steps for running and analyzing an experiment using MALT.}
	\label{fig:platform}
\end{zfigure}

In this section we give an overview of the structure of the platform. The five steps for running an experiment with the platform are summarized in Figure~\ref{fig:platform}. There are three major components to this platform: the configuration GUI, the experiment engine (the piece that simulates the repeated games) and the visualization GUI. We describe each in turn.

The first step is to set up the experiment. First, a group of algorithms must be picked and algorithm parameters set. Second, a set of GAMUT game distributions must be selected and parameters for these games chosen. Third, general experimental parameters must be established, such as the number of iterations for each simulation. These decisions are encoded in human-readable text files, and can either be generated using a provided GUI or using batch scripts.

The second step is to generate a job file for each desired match. Each job file references the agent, game, equilibrium, and maxmin-strategy files. These files are referenced, \note{what does that mean?} \note{EPZ: It means that if an algorithm uses a parameter $\alpha$ then the value is stored in a single referenced file, rather than coping $\alpha$'s value into each of the $m$ jobfiles that use that particular algorithm. This means that if we want to adjust $\alpha$ after creating the jobfiles than it's a single change, rather than $m$ changes.} making altering the job files simple even after they have been generated.

The third step is to run the jobs. This primarily involves running the MALT ``engine''; however, \note{I've moved this to step 3. Please move it as necessary.}\note{EPZ: Are you asking me to change the engine to make this true?}
MALT calls GAMBIT's \cite{mckelvey04} implementation of Lemke-Howson \cite{lemke64} when an algorithm needs to find the set of Nash equilibria for a game instance, and CPLEX when an algorithm needs a maxmin strategy. Jobs may be run in several ways. The most basic is to run them in a batch job on a single machine. However, for large experiments this can be prohibitively expensive. Because each job is independent, it is straightforward to use a compute cluster. To facilitate such parallelization, each job creates an individual data file upon completion that records the history of play. For each recorded iteration and for each agent in the pair, the strategy, sampled action, reward received, and beliefs about the opponents are recorded.

Step four is to compute performance metrics based on these data files. A plain-text file specifies the metrics to be calculated, based on an extensible library of available metrics. As above, metrics can be computed in a batch or can be distributed across a cluster.

Finally, step five is to analyze and visualize these results. To make this task easier, MALT includes some basic analysis tools and a visualization GUI.

\subsection{Algorithm Implementations}
\label{sec:algorithms}

To carry out this study, we selected and implemented eleven MAL algorithms, most of which we discussed previously in \S\ref{sec:related_work}. In cases where reference code was available, we performed extensive validation experiments to ensure that our implementation was correct.

\subsubsection{\FP}

Parameters for \textbf{\fp} are given in Table~\ref{tab:fictitious_parameters}. We note that the initial action frequencies were set to one for each action, which is a uniform and easily overwhelmed prior. Actions were selected from non-singleton best-response sets by favoring an action that was played in the previous iteration if present, and selecting uniformly at random otherwise.

\begin{kfigure}[t]
\begin{minipage}[t]{0.45\linewidth}
\begin{table}[H]
\begin{tabular}{ll}
		\hline
		\textbf{Design Decision~~~~~~~} 						&	 \textbf{Setting~~~~~~~~~~~~~~}\\
		\hline
			BR Tie-Breaking 		&	Previous action if still BR\\
													&	Uniform otherwise\\
			Initial Beliefs			&	Unit virtual action count\\
			\hline
		\end{tabular}
		\caption{Design decisions for \agent{fictitious play}}
    \label{tab:fictitious_parameters}
\end{table}
\end{minipage}
\hfill
\begin{minipage}[t]{0.45\linewidth}
\begin{table}[H]
		\begin{tabular}{ll}
				\hline
		\textbf{Design Decision~~~~~~~} 						&	 \textbf{Setting~~~~~~~~~~~~~~}\\
		\hline
			NE Tie-Breaking 		&	Highest opponent utility\\
			\hline
		\end{tabular}
		\caption{Design decisions for \agent{determined}}
    \label{tab:determined_parameters}
\end{table}
\end{minipage}
\end{kfigure}

\subsubsection{\Determined}

Our implementation of \textbf{\determined} (see Table~\ref{tab:determined_parameters}) repeatedly plays the Nash equilibrium that obtains the highest personal reward, but if there are multiple equilibria with the same protagonist reward, then the equilibrium with the highest opponent reward is selected. If there are any equilibria that are still tied we use the one found first by GAMBIT's implementation of Lemke-Howson.

\subsubsection{\awesome}

\textbf{\awesome} is implemented according to the pseudo-code in \inlinecite{conitzer07}, and largely uses the parameter settings given there; see Table~\ref{tab:awesome_parameters}. Previous work on \awesome\ assumed that the games that it encountered had a single Nash equilibrium. However, the games that we examine in this paper may have multiple equilibrium, so we needed to decide how our implementation of \awesome\ would single out a special equilibrium for convergence in self-play. In our implementation the special equilibrium is the first equilibrium found by GAMBIT's implementation of Lemke-Howson. It would be interesting to compare our implementa\-tion of \awesome\ to one that used the more computationally-expensive approach of picking, say, a socially optimal equilibrium.\footnote{In our validation experiments we observed a small but statistically significant difference between the behavior of our implementation of \agent{AWESOME} and the original implementation from \inlinecite{conitzer07}. (The original implementation was in C and MALT 2.0 is written in Java, so the original implementation could not be used directly.) Specifically, a test involving ten different game instances and 100 runs against the random agent showed a significant difference between solution quality distributions on three instances. We used a two-sample Kolmogorov-Smirnov independence test (see \S\ref{sec:ks_test}) with $\alpha = 0.05$ to check for significance. For these three game instances, our implementation probabilistically dominated (see \S\ref{sec:prob_dom}) the original implementation in terms of reward (i.e., every reward quantile was higher for our implementation). We were not able to track down the source of this behavior difference; however, we spent a considerable amount of time verifying our implementation against the pseudo-code in the paper and were unable to find any difference.}

\begin{kfigure}[t]
\begin{minipage}[t]{0.45\linewidth}
\begin{table}[H]
		\begin{tabular}{ll}
				\hline
		\textbf{Design Decision~~~~~~~} 						&	 \textbf{Setting~~~~~~~~~~~~~~}\\
		\hline
						Special Equilibrium($\pi^{\ast}_p$)					& First found\\
						Epoch period ($N(t)$)												& $\left\lceil \frac{\left|A\right|_{\Sigma}}{\left(1 - \frac{1}{2^{t^{-2}}}\right)\left(\epsilon_{e}^{t}\right)^2} \right\rceil$		 \\
						Equilibrium threshold ($\epsilon_e(t)$)			& $\frac{1}{t + 2}$\\
						Stationarity threshold ($\epsilon_s(t)$) 		& $\frac{1}{t + 1}$\\
					\hline
		\end{tabular}
		\caption{Design decisions for \awesome}
    \label{tab:awesome_parameters}
\end{table}
\end{minipage}
\hfill
\begin{minipage}[t]{0.45\linewidth}
\begin{table}[H]
\begin{tabular}{ll}
		\hline
		\textbf{Design Decision~~~~~~~} 						&	 \textbf{Setting~~~~~~~~~~~~~~}\\
		\hline
						Security threshold ($\epsilon_0$)						&	$0.01$	\\
						Bully threshold ($\epsilon_1$)							&	$0.01$	\\
						``Generous'' BR parameter  ($\epsilon_2$)		&	$0.005$	\\
						Stationarity threshold ($\epsilon_3$)				&	$0.025$	\\
						Coordination/exploration period ($\tau_0$)	&	$90~000$	\\
						Initial period ($\tau_1$)										&	 $10~000$	 \\
						Secondary period ($\tau_2$)									&	$80~000$	 \\
						Security check period ($\tau_3$)						&	$1~000$	\\
						Switching probability ($p$)									&	$0.00005$	 \\
						Window ($H$)																 &	 $1~000$	\\
						$\left\|\cdot\right\|$											& $\ell_2$ \\
			\hline
		\end{tabular}
		\caption{Design decisions for \meta}
    \label{tab:meta_parameters}
\end{table}
\end{minipage}
\end{kfigure}

\subsubsection{\Meta}

\textbf{\Meta} is implemented according to the pseudo-code in \inlinecite{powers04}. The \inlinecite{powers04} implementation of \meta\ used a distance measure based on the Hoeffding Inequality, even though the pseudo-code called for using an $\ell_2$ norm. We follow the pseudo-code and use the $\ell_2$ norm. We do not adjust the default threshold level ($\epsilon_3$) for distance, leaving it at the original value. All parameters for \meta\ are summarized in Table~\ref{tab:meta_parameters}.

\subsubsection{Gradient Algorithms}

Our implementation of \textbf{\giga} follows the original pseudo-code and uses the learning rate and step size schedules from the original experiments by \inlinecite{bowling04} as defaults; see Table~\ref{tab:giga_parameters}. We note, however, that these step sizes were set for drawing smooth trajectories and may not necessarily yield strong performance, and furthermore that the original experiments for \giga\ involved more iterations than we simulated ($10^6$ as compared to $10^5$).
For \giga's retraction map operation (the function that maps an arbitrary vector in $\Re^n$ to the closest probability vector in terms of $\ell_2$ distance) we used an algorithm based on the method described in \inlinecite{govindan03}. \giga\ has two variants: in one it assumes that it can counterfactually determine the reward for playing an arbitrary action in the previous iteration, and in the other it only knows the reward for the the action that it played and has to approximate the rewards for the other actions. We implemented the latter approach, as all of \giga's experimental results are produced by this version. The formula for the approximation is given by
\begin{equation}
	\forall \dot{a} \in A_i \; \hat{r}^{(t + 1)}_{\dot{a}} = (1 - \alpha)r^{(t)}\mathbb{I}_{\dot{a} = a^{(t)}} + \alpha(\hat{r}^{(t)}_{\dot{a}}).
	\label{eq:r_update}
\end{equation}
In this equation, $r^{(t)}$ is the reward that the algorithm experienced while playing action $a^{(t)}$ in iteration $t$. The vector $\hat{r}^{(t)}$ is an $|A_i|$-dimensional vector that reflects the algorithm's beliefs about rewards.

\begin{kfigure}[t]
\begin{minipage}[t]{0.45\linewidth}
\begin{table}[H]		
\begin{tabular}{lll}
		\hline
		\textbf{Design Decision~~~~~~~~~~~~~~~~~} 		&	\textbf{Setting~~~~~~~~~~~~~~}\\
		\hline
		Learning rate ($\alpha(t)$) 			&	$\frac{1}{\sqrt{\frac{t}{10} + 100}}$ \\
		Step size ($\eta(t)$)							& $\frac{1}{\sqrt{10^4t + 10^8}}$ \\
			\hline
		\end{tabular}
		\caption{Design decisions for \giga.}
    \label{tab:giga_parameters}
\end{table}
\end{minipage}
\hfill
\begin{minipage}[t]{0.45\linewidth}
\begin{table}[H]		
\begin{tabular}{lll}
		\hline
		\textbf{Design Decision~~~~~~~} 		&	\textbf{Setting~~~~~~~~~~~~~~}\\
		\hline
		Learning rate ($\alpha(t)$) 			&	$\frac{1}{\sqrt{\frac{t}{10} + 100}}$ \\
		Step size ($\eta(t)$)							& $\frac{1}{\sqrt{10^4t + 10^8}}$ \\
		Noise Weight ($\lambda(t)$)				& $\frac{1}{\sqrt{10^5t + 10^8}}$ \\							 
			\hline
		\end{tabular}
		\caption{Design decisions for \gsa.}
    \label{tab:gsa_parameters}
\end{table}
\end{minipage}
\end{kfigure}

We also tested the Global Stochastic Approximation algorithm, \textbf{\gsa}, of \inlinecite{spall03}; see Table~\ref{tab:gsa_parameters}. To our knowledge we were the first to suggest its use in a MAL setting \cite{lipson05}. This algorithm is a stochastic optimization method that resembles \agent{GIGA}, but takes a noisy, rather than deterministic, step. The \agent{GSA} strategy is updated as
\begin{equation}
	x^{(t+1)} = P(x^{(t)} + \eta^{(t)}r^{(t)} + \lambda^{(t)}\zeta^{(t)}),
	\label{eq:gsa}
\end{equation}
where $x_t$ is the previous mixed strategy, $r_t$ is the reward vector, $\zeta_t$ is a vector where each component is sampled from the standard normal distribution (with variance controlled by the parameter $\lambda^{(t)}$), and $P(\cdot)$ is the same retraction function used for \giga.

\textbf{\rvs} is a implementation of the algorithm given in \inlinecite{banerjee06}. Some initial experiments showed that the settings of the algorithm used in the paper performed very poorly, and so we used some hand-picked parameter settings that were more aggressive and seemed to perform better. These are given in Table~\ref{tab:rv_parameters}.

\begin{kfigure}[t]
\begin{minipage}[t]{0.45\linewidth}
  \begin{table}[H]		
\begin{tabular}{lll}
		\hline
		\textbf{Design Decision~~~~~~~~~~~~~~~~~} 		&	\textbf{Setting~~~~~~~~~~~~~~}\\
		\hline						
		$\sigma$-schedule ($\sigma(t)$) 	&	$\frac{1}{1 + \frac{1}{25}\sqrt{t}}$ \\
		Step size ($\eta(t)$)							& $\frac{1}{ \sqrt{1000t + 10^5}}$ \\
			\hline
		\end{tabular}
		\caption{Design decisions for \rvs.}
    \label{tab:rv_parameters}
\end{table}
\end{minipage}
\hfill
\begin{minipage}[t]{0.45\linewidth}
  \begin{table}[H]		
\begin{tabular}{lll}
		\hline
		\textbf{Design Decision~~~~~~~} 						&	 \textbf{Setting~~~~~~~~~~~~~~}\\
		\hline
		Learning rate	($\alpha(t)$) 					&	$\left(1 - \frac{1}{2000} \right)^t$ \\
		Exploration rate ($\epsilon(t)$)			& $\frac{1}{5}\left(1 - \frac{1}{500} \right)^t$ \\
		Future discount factor ($\gamma$)			& $0.9$\\
		\hline
		\end{tabular}
		\caption{Design decisions for \q.}
    \label{tab:q_parameters}
\end{table}
\end{minipage}
\end{kfigure}

\subsubsection{\Q}

Our implementation  of \textbf{\q} is very basic; see Table~\ref{tab:q_parameters}. Since in a repeated game there is only one `state', \q\ essentially keeps track of $Q$-values for each of its actions. We use an $\epsilon$-greedy exploration policy (perform a random action with probability $\epsilon$) with a decaying $\epsilon$. $400$ exploration steps are expected for this $\epsilon$-schedule, and $\epsilon$ drops below a probability of $0.05$ at approximately iteration $2800$. It is negligible at the end of the settling-in period (less than $3E{-9}$). The learning rate ($\alpha$) decays to $0.01$ at the end of the settling in period. The discount factor of $\gamma = 0.9$ was set rather arbitrarily. There is no need to trade off current reward with future reward: all actions take the algorithm back to the same state.

\subsubsection{\Minimax\ and \Minimaxidr}

For \textbf{\minimax}, we solved a linear program to find the mixed maxmin strategy based on the $Q$-values. This program was
\[\begin{array}{lll}
\textrm{\textbf{Maximize}}		&   U_1 																													 & \\
\textrm{\textbf{subject to}}	&  	\sum_{j \in A_1}{u_1(a_1^j, a_2^k) \cdot \sigma_1^j} \geq U_1 & \forall k \in A_2\\
															&		\sum{\sigma_1^j} = 1 																					 & \\
															&		\sigma_1^j \geq 0 																						 & \forall j \in A_1
\end{array}\]
(see, for example, \inlinecite{MASfoundations08}). We also considered a variant of \minimax\ in which iterative domination removal (IDR) is used as a preprocessing step. To our knowledge, we were the first to propose this algorithm in \inlinecite{lipson05}; we dubbed it \minimaxidr. In each step of the iterative IDR algorithm mixed-strategy domination is checked using a linear program (see, for example, \inlinecite{MASfoundations08}). Both LPs are solved with CPLEX 10.1.1. For both \minimax\ and \minimaxidr, the learning rate, exploration rate, and future discount factor were set as in \q; see Table~\ref{tab:q_parameters}.

\subsubsection{\Random}

The final algorithm, \textbf{\random}, is an simple baseline that uniformly mixes over the available actions. Specifically, it submits a mixed strategy $\sigma$ where $\forall a \in A, \; \sigma(a) = \frac{1}{\left|A\right|}$.

\section{Experimental Setup and Statistical Methods}
\label{sec:experiment}
As described in the preamble, this paper makes two main contributions. The first is the MALT platform, which we have now explained. The second is a demonstration of what MALT can do. Specifically, we conducted an large-scale experiment with the goal of investigating the empirical relationship between average reward and other performance metrics (e.g., equilibrium convergence; regret) that have been considered in the literature. In this section we describe the setup of this experiment and some of the statistical tools we used in our analysis.

We studied all eleven of the algorithms described in \S\ref{sec:algorithms}, and set their parameters as described there. We note in passing that this choice was important, as some algorithms are very sensitive to parameter settings. Nevertheless, we considered the issue of parameter optimization to be beyond the scope of our study, and took parameter settings largely from the literature as given. 

We selected thirteen game generators from the GAMUT game collection; these are summarized in Table~\ref{tab:generators}. Details of each generator are available in GAMUT's online documentation; see \url{gamut.stanford.edu}. We normalized the rewards of all game instances to the $[0,1]$ interval in order to make the results more interpretable and comparable.
We generated a total of $600$ different game instances. Specifically, we generated games of five different sizes: $2 \times 2$, $4 \times 4$, $6 \times 6$, $8 \times 8$ and $10 \times 10$. For each size, we generated 100 game instances, drawing uniformly from the first twelve generators. We drew an additional 100 instances from the last distribution, D13, which spans all strategically distinct $2 \times 2$ games \cite{rapoport22}. We call the distribution induced by mixing over all 13 GAMUT generators the \textit{grand distribution}.

\begin{table}[t]
    \begin{tabular}{ll}        \hline
     D1 & A Game With Normal Covariant Random  Payoffs  \\
     D2 & Bertrand Oligopoly      \\
     D3 & Cournot Duopoly        \\
     D4 & Dispersion Game       \\
     D5 & Grab the Dollar             \\
     D6 & Guess Two Thirds of the Average   \\
     D7 & Majority Voting               \\
     D8 & Minimum Effort Game          \\
     D9 & Random  Symmetric Action Graph Game   \\
    D10 & Travelers Dilemma      \\
    D11 & Two Player Arms Race Game    \\
    D12 & War of Attrition     \\
    D13 & Two By Two Games     \\
    		\hline
    \end{tabular}
    \caption{The number and name of each game generator.}
    \label{tab:generators}
\end{table}

With eleven algorithms and $600$ game instances there were $11 \times 11 \times 600 = 72~600$ matches. We ran each match once\footnote{We note that each match could have been run multiple times instead of just once, and indeed that doing so would have been essential if we wanted to understand the behavior of randomized algorithms in individual matches. However, holding CPU time constant, conducting more runs per match would have meant either experimenting with fewer games or with fewer algorithms.
Indeed, we show in Appendix~A that not stratifying (holding one experimental variable fixed while varying another; as opposed to varying both) on game instances reduces variance for sample estimates of summary statistics like mean and median. Thus, we ran each match only once, and therefore use the terms `run' and `PSM' interchangeably in what follows.} for $100~000$ iterations, recording the last $10~000$ iterations. This generated $143\,GB$ of data and took about a third of a CPU-year to run. 
In order to interpret the results we relied upon a variety of different empirical methods. We briefly describe some of them below.

\subsection{Bootstrapping}
\label{sec:boots}

If we conduct an experiment where two algorithms are run on a number of PSMs then a natural way to compare their performance is to compare the sample means of some measure of their performance (average reward, for example). However, if we have the conclusion that `the sample mean of algorithm $A$ is higher than the sample mean of algorithm $B$', how robust is this claim? If we ran this experiment again are we confident  that it would support the same conclusion?

A good way to check the results of an experiment is to run it multiple times. For example, imagine that we ran an experiment $100$ times and found that $95$ of the experiments had a sample mean for algorithm $A$ of between $[\underline{a},\, \overline{a}]$, and that $95$ of the experiments had a sample mean for algorithm $B$ of between $[\underline{b},\, \overline{b}]$. If $\underline{a} > \overline{b}$ (the lower bound of $A$'s interval was greater than the upper bound of $B$'s) then we can be confident that $A$ is significantly better in terms of mean. (Specifically, these intervals are the $95\%$ percent confidence intervals of the sample mean distribution, and the fact that they do not overlap serves as sufficient evidence that there is a significant relationship between the means.)

While such repeated experimentation can be used to ensure that results are significant, it is also expensive. To verify the summary statistics from one experiment, we had to run many more. This is not always possible (e.g., our experiments took $7$ days on a large computer cluster, so to rerun them a hundred more times would have taken the better part of two years). Bootstrapping is a technique that allows us to use the data from a \emph{single} experiment to construct confidence intervals of summary statistics. Given an experiment with $m$ data points, we can `virtually' rerun the the experiment by subsampling from the empirical distribution defined by those $m$ points. For example, if we have a sample with $100$ data points, we could subsample $50$ data points (with replacement) from these $100$ and look at the statistic for this subsample. We can cheaply repeat this procedure as many times as we like, creating a distribution for each estimated statistic. From these bootstrapped estimator distributions we can form bootstrapped confidence intervals and check for overlap. 

There are two parameters that control the bootstrapped distribution: we form the distribution by subsampling $l$ points from the original $m$, and we repeat this process $k$ times. For our analysis we chose $l$ to be $\left\lfloor  m/2 \right\rceil$ and $k$ to be around $2~500$. These particular parameters were chosen to ensure that there would be diversity among the subsamples (this explains the moderate size of $l$) and that the empirical distributions would be relatively smooth (this explains the large $k$).

\subsection{Kolmogorov-Smirnov Test}
\label{sec:ks_test}

While bootstrapping is useful for seeing if summary statistics are significantly different, we will also want to check if two distributions are themselves significantly different. A beta distribution and a Gaussian distribution might coincidentally have the same mean, but are nevertheless different distributions. We use the Kol\-mo\-gorov\-Smir\-nov (KS) 
test for determining whether two distributions are different. This test is non\-parametric, meaning that it does not assume that the underlying data is drawn from a known (e.g., normal) probability distribution. The KS test works by examining the maximum vertical distance between two CDFs \note{Asher: Does this need to be spelled out?}. Two distributions are considered significantly different if this maximum vertical distance exceeds a given significance level, $\alpha$. In our analysis we use the standard $\alpha = 0.05$ unless otherwise noted.

\subsection{Spearman's Rank Correlation Test}
\label{sec:spearman}

Spearman's rank correlation test is a way to establish whether or not there is a significant monotonic relationship between two paired variables. For example, we might want to show that there is some significant monotonic relationship between the size of a game's action set size and an agent's average reward. Like the KS test, the Spearman's rank correlation test is non-parametric. The relationship between the two variables can be positive (high values of one variable are correlated with high values of the other variable) or negative (high values of one variable are correlated with low values of the other). \note{What does this last sentence mean? Shouldn't it describe what it means to pass or fail the test?} \note{EPZ: Correlation and significance are different. When we run a test, we get a $p$-value and a $\rho$-value. If the $p < 0.05$, then the results is significant. Separately, we have the coefficient of correlation $\rho$ which tells use how strong the correlation is---this is related to the $p$-value but not identical to it---and whether it is positive or negative correlation. For instance, if $x_1 = x_2 + \mathcal{N}(0,\sigma)$ then they are positively correlated and if $x_1 = -x_2 + \mathcal{N}(0,\sigma)$ then they are negatively correlated. $\vert\rho\vert$ and $p$ depends on both $\sigma$ and the number of data points.}

\subsection{Assessing Convergence}
\label{sec:threshold}

We are interested in studying the convergence behavior of MAL algorithms. One issue in doing so based on empirical data is dealing with runs that appear ``not quite'' to have converged because of random fluctuations in the empirical action frequency. A natural solution to this problem is to perform a statistical test to determine whether one part of the run exhibits the same action distribution as a later part. For example, we might check whether a later empirical action distribution was drawn from the same distribution as an earlier sample (establishing that  empirical mixed strategies were stationary) or that an empirical action distribution profile was drawn from a given mixed-strategy profile (establishing convergence to a Nash equilibrium).

Two obvious candidates for such a test are the Fisher exact test (FET) and Pearson's $\chi^2$-test, which can be used for checking whether two multinomial samples are drawn from a distribution. However, each test was unfortunately inappropriate for our problem. The $\chi^2$ test does not handle situations where some of the actions are rare or not present. The FET is very computationally expensive, and the implementation of it that we used \cite{r06} failed on some of the larger and more balanced action vectors (typically in the $10 \times 10$ case).

Instead, we used the incomplete set of FET results to calibrate a threshold based on vector distance, where we considered any two action distribution vectors \note{Asher: Is this vector specifically the distribution over actions?} that were closer than the threshold $\theta$ to be the same. We calibrated $\theta$ using a receiver operating characteristic (ROC) curve. We used the incomplete FET results as ground truth, and plotted the change in true positive rate and false positive rate as we varied $\theta$. We picked the threshold that led to an equal number of false positives and false negatives. Based on this ROC analysis, we picked a $\theta$ of $0.02$.

\subsection{Probabilistic Domination}
\label{sec:prob_dom}

%
%
%
The concept of probabilistic domination can be used to argue that one distribution should be preferred to another in terms of a given performance metric. Specifically, a solution quality distribution (SQD) $A$ dominates another SQD $B$ if $\forall q \in [0,1]$, the $q$-quantile of $A$ is higher than the $q$-quantile of $B$. If there are two algorithms, $A$ and $B$, that are trying to maximize reward, and $A$'s SQD probabilistically dominates $B$ then regardless of the reward value $r$, there are more runs of $A$ than of $B$ that attain a reward of at least $r$.
Probabilistic domination is stronger than a claim about the mean of the distributions: domination implies higher means.

Checking for probabilistic domination between two samples can be performed visually. If one of the CDF curves is below the other curve everywhere, than the former dominates the latter. Intuitively, this is because the better SQD has less probability mass on low solution qualities, and more mass on higher solution qualities; better distributions are right-shifted.

\section{Empirical Evaluation of MAL Algorithms: Average Reward}
\label{sec:result}
As we discussed at the beginning of this paper, we consider average reward to be the most fundamental metric for assessing the performance of a MAL algorithm. We take the average with respect to the sampled actions rather than the submitted mixed strategy. Formally, where the iterations $1$ to $T$ refer to the $10~000$ iterations we recorded, we define the average reward an algorithm $i$ obtains in a single match as $\bar{r}_i^{(T)} = \frac{1}{T}\sum_{t=1}^{T}r_i^{(t)}.$

In this section, we investigate the average reward metric in detail. We begin in \S\ref{sec:reward} by comparing algorithms according to their ``raw'' average reward, averaging not only over iterations but also across both generators and opponents \note{Asher: Also averaged over number of actions?} \note{EPZ: Don't understand: are you asking if we are dividing by the size of the action set, or the number of iterations?}. Next, we investigate each of these dimensions separately. In \S\ref{sec:block-generators} we explore algorithm performance across different generators, and also examine the effect of game size. In \S\ref{sec:block-opponent} we explore algorithm performance across different opponents, and also analyze the equilibria of the ``algorithm game'', in which available actions are different choices of MAL algorithms. \S \ref{sec:reward-probdom} investigates probabilistic domination relationships between different algorithms and \S\ref{sec:selfplay} considers each algorithm's performance in self play. Finally, \S \ref{sec:alg-similarity} explores similarities between different algorithms.

\subsection{``Raw'' Average Reward}
\label{sec:reward}

First we consider each algorithm's ``raw'' average performance, averaged across the number of iterations, games and opponents. \note{Asher: Ditto for previous note on averaging over the number of actions as well?}

\begin{ob}
\Q\ and \rvs\ attained the highest rewards on the grand distribution.
\end{ob}

\begin{zfigure}[t]
\begin{minipage}[t]{0.45\linewidth}
	\includegraphics[width=1\textwidth]{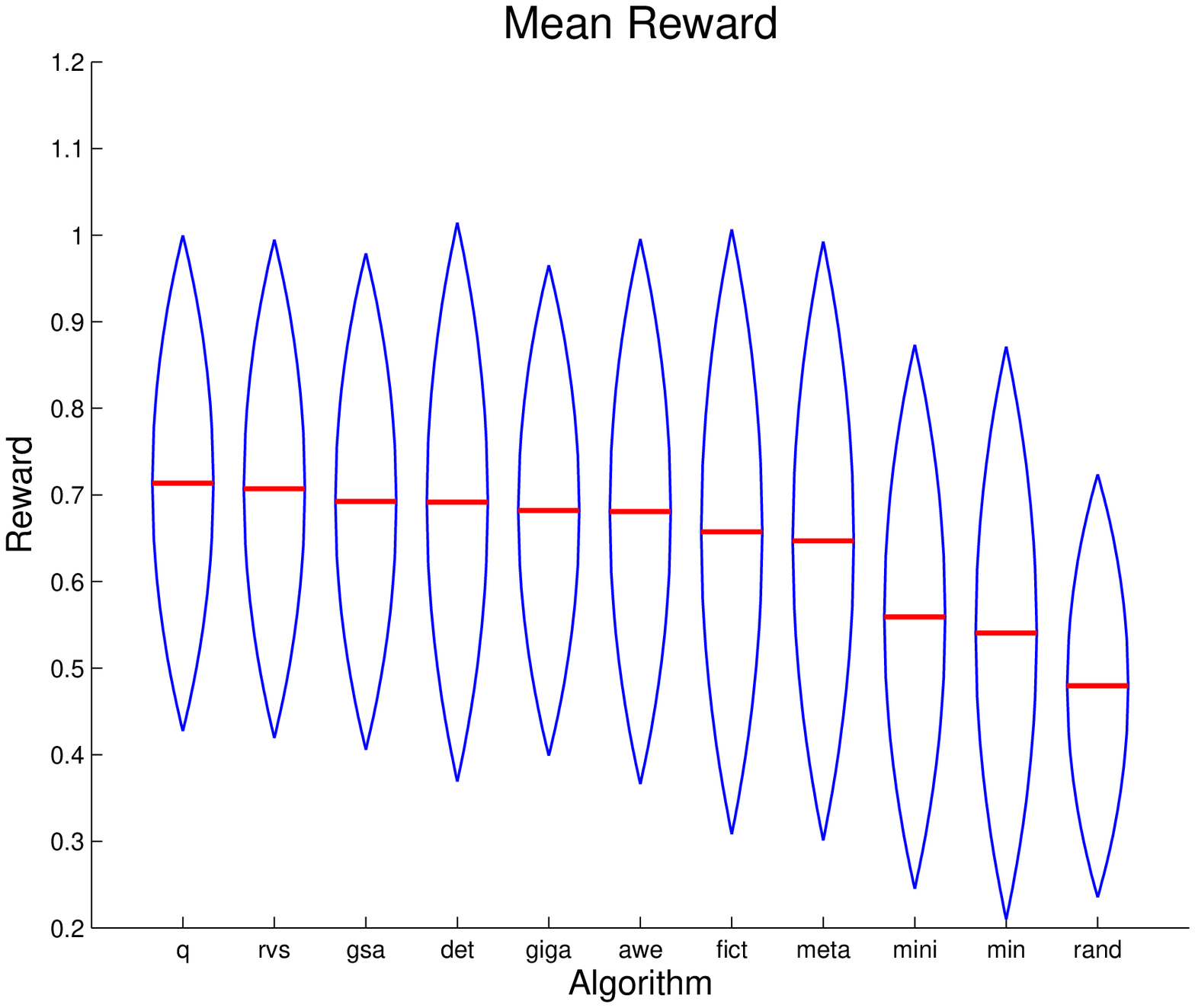}
	\caption[Lens plot for reward]{A plot that shows the mean reward (bar) for each algorithm and one standard deviation in either direction (the size of the lens).}
	\label{fig:reward_mvplot}
\end{minipage}
\hfill
\begin{minipage}[t]{0.45\linewidth}
	\includegraphics[width=1\textwidth]{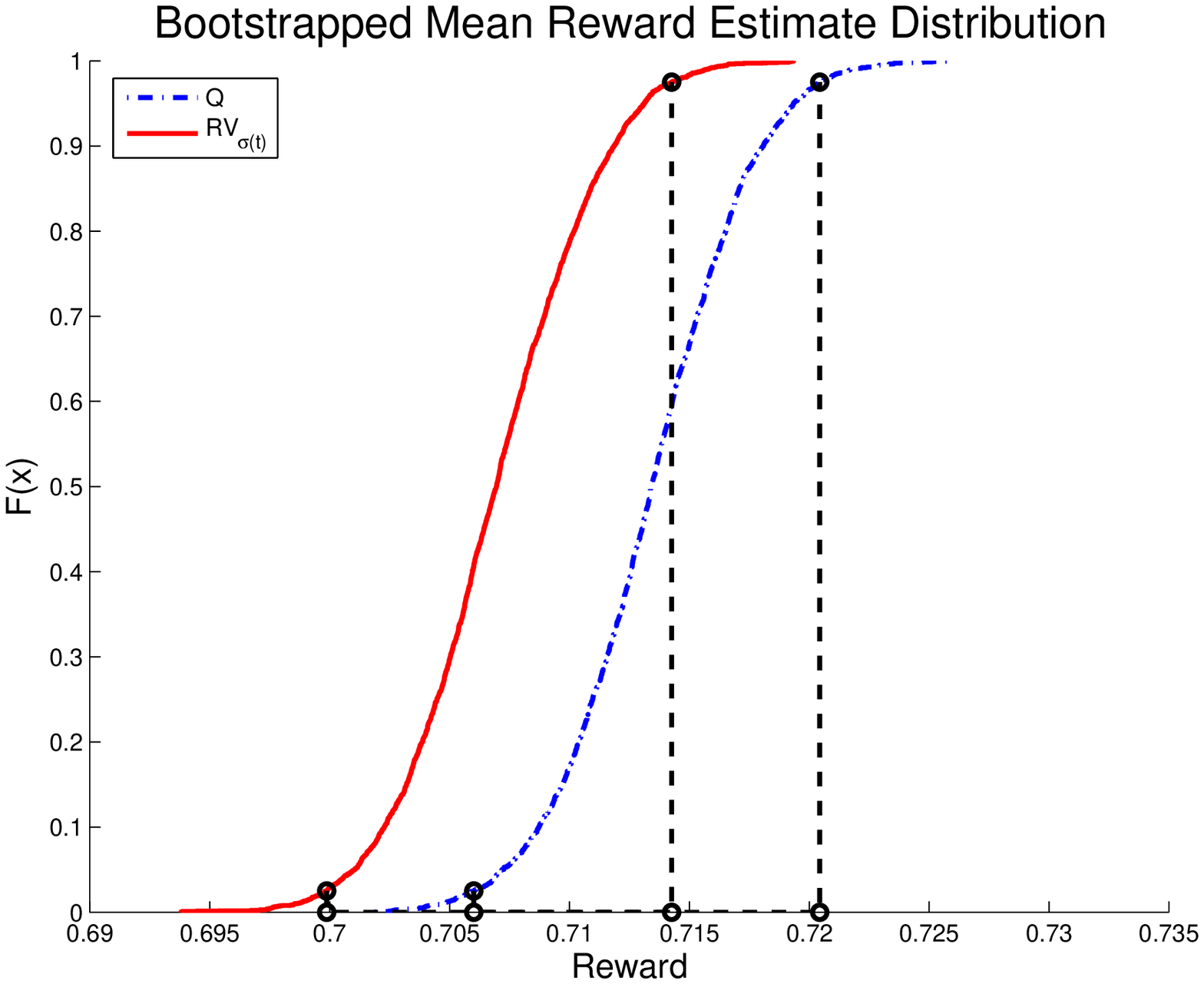}
	\caption[Bootstrapped mean estimator distributions for \q\ and \rvs]{The distribution of mean reward estimates for \q\ and \rvs, constructed by bootstrapping. The $95\%$ confidence intervals are indicated by the dark circles and dashed-lines.}
		\label{fig:bootstrap_mean_reward}
\end{minipage}
\end{zfigure}

\q\ had the highest mean reward at $0.714$, although \rvs\ was close with an average of $0.710$ (see Figure~\ref{fig:reward_mvplot}). We noticed considerable variation within the reward data, and all of the other algorithms' sample means still were within one standard deviation of \q, including \random\ (which obtained a sample mean of $0.480$).

These rankings were not all significant. The slight difference in means between \q\ and \rvs\ does not in fact indicate that \q\ was a better algorithm (in terms of means) on the grand distribution of games and opponents. However, these two algorithms attained significantly higher reward than any other algorithm. We determined this by examining the $95\%$ percentile intervals on bootstrapped mean estimator distributions (see \S\ref{sec:boots}) and seeing which intervals overlapped (see Figure~\ref{fig:bootstrap_mean_reward}). We obtained the distributions by subsampling $2~500$ times, where each subsample consisted of $6~600$ runs (half as many as the $13~200$ runs that each algorithm participated in).

The distribution of reward was not symmetric, and specifically tended to exhibit negative skewness, indicating that the proportion of runs that attained high reward was larger than the proportion of runs that attained low reward. (\Random\ was the only exception.) \q's distribution had the highest skewness, $-0.720$, so was the best algorithm in this respect. \note{Asher: Should this have an extra sentence or two saying whether this is interesting, or is this put in as an example, without comment?} \note{EPZ: added a fragment to nail down that having a high negative skew is ``good'', and that \q\ was the most skewed.}

\subsection{Per-Generator Average Reward and the Effect of Game Size}
\label{sec:block-generators}

\begin{zfigure}[t]
\begin{minipage}[t]{0.45\linewidth}
			\includegraphics[width=1\textwidth]{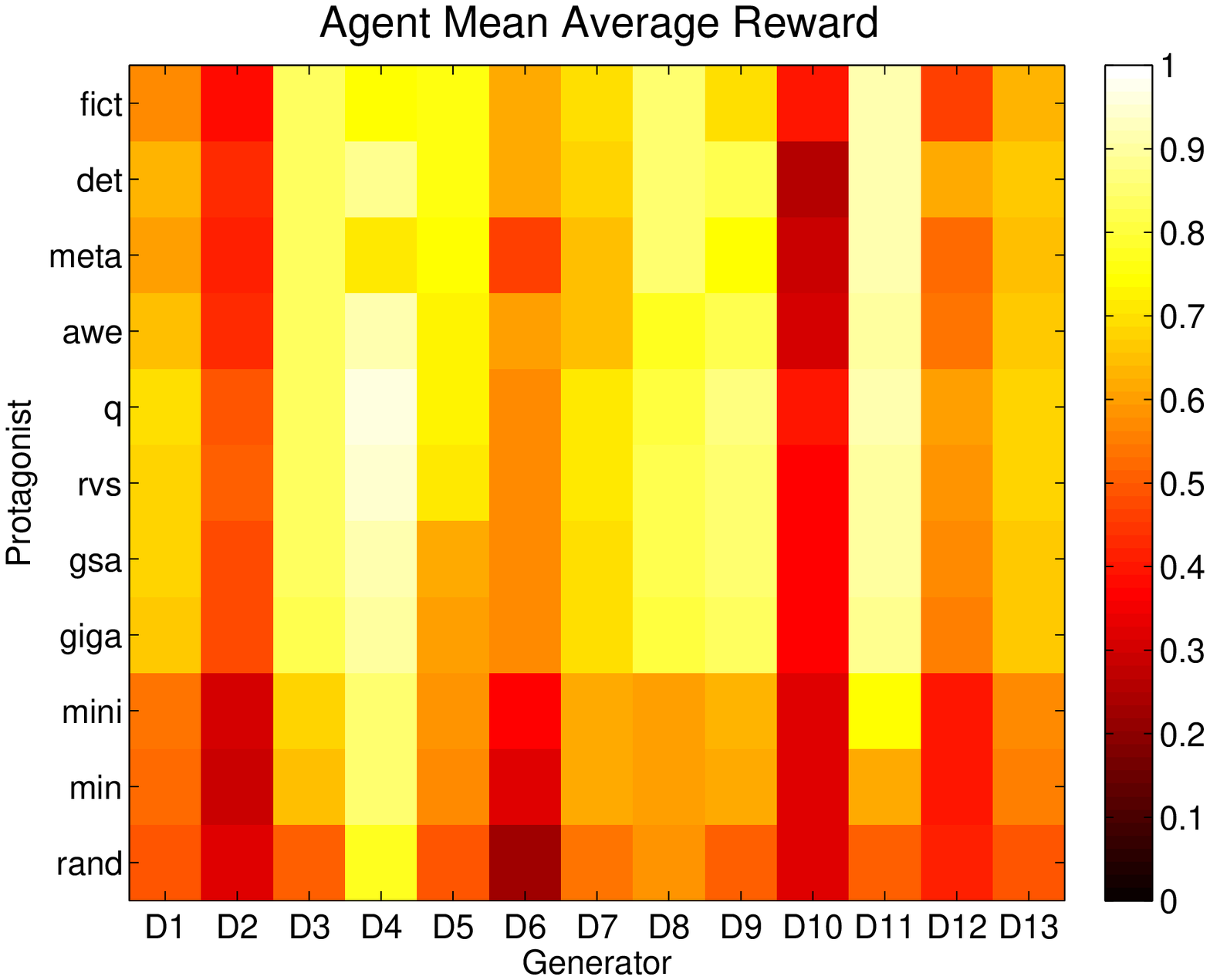}
	\caption[Algorithm-generator heatmap for reward.]{A heatmap showing the reward for the protagonist algorithm playing PSMs from a particular generator, averaged over both iterations and PSMs.}
		\label{fig:ag_heatmap_utility_mean_grand}
\end{minipage}
\hfill
\begin{minipage}[t]{0.45\linewidth}
		\includegraphics[width=1\textwidth]{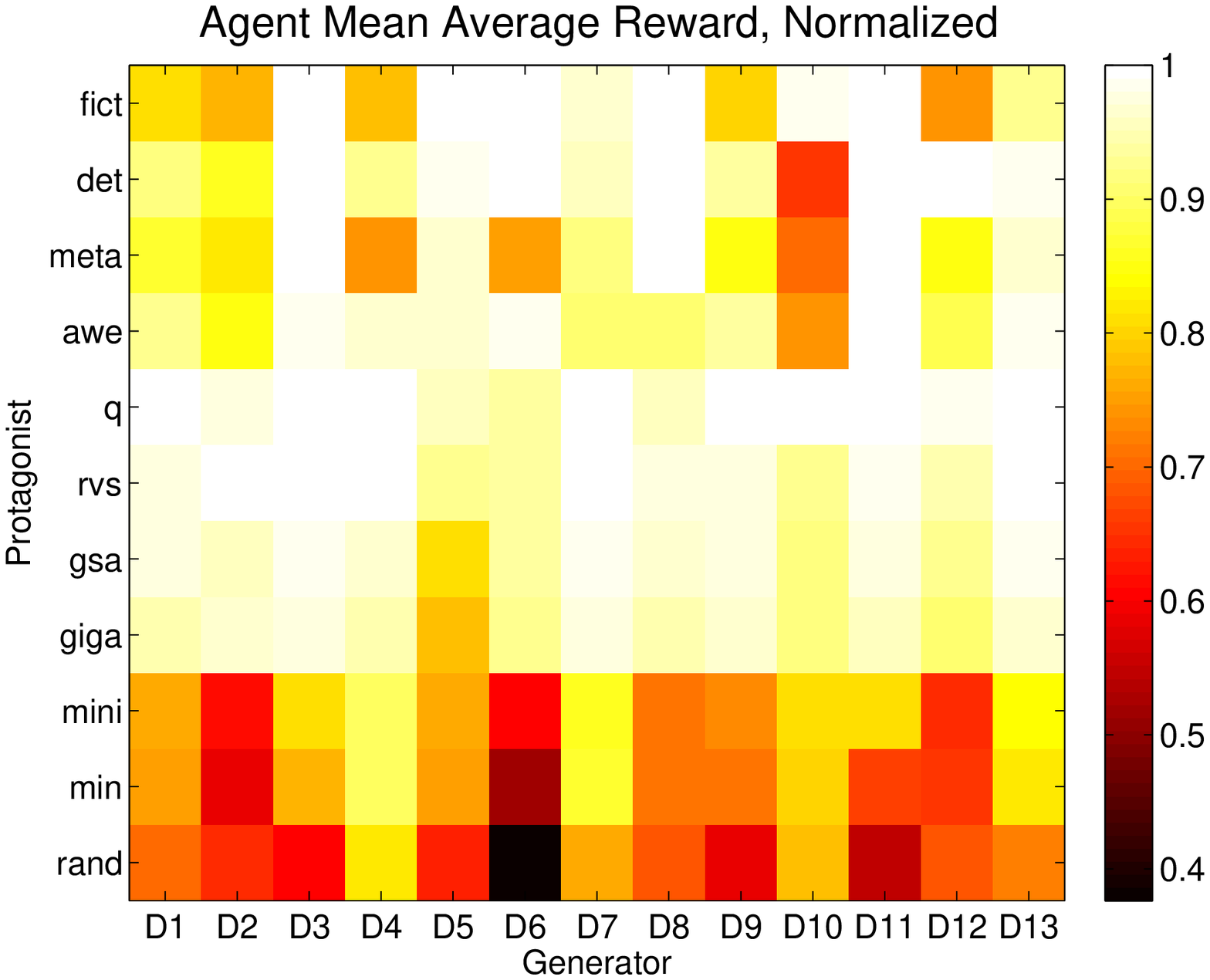}
	\caption[Normalized algorithm-generator heatmap for reward.]{A heatmap showing the mean reward for the protagonist algorithm, playing against the opposing algorithm. These cells have been normalized. Each column has been divided by the maximum average reward attained by any algorithm on that particular generator.}
		\label{fig:ag_heatmap_utility_mean_grand_normalized}
\end{minipage}
\end{zfigure}

\note{Asher: I like these heatmaps a lot - a very nice way of visualising everything} Now we go beyond performance on the grand distribution. First we consider each algorithm's performance across individual game distributions. 
%
As can be seen in Figure~\ref{fig:ag_heatmap_utility_mean_grand}, every algorithm's performance varies considerably across the different game generators. However, this figure makes it difficult to determine the best algorithm for generators that all algorithms found challenging. Thus, we also present a normalized version of these per-generator reward results, obtained by dividing the results for each algorithm on a particular generator by the maximum reward attained by any algorithm (Figure~\ref{fig:ag_heatmap_utility_mean_grand_normalized}). We can see that \minimax, \minimaxidr\ and \random\ were all worse than the other algorithms across a broad range of generators, and \q\ and \rvs\ tended to do well.

\begin{ob}
\Q\ was the best or one of the best algorithms to use for most generators.
\end{ob}

%

We define the set of best algorithms for a generator as the set of algorithms whose bootstrapped mean estimator $95\%$ percentile intervals overlapped with the algorithm with the best sample mean.
\Q\ was the unique best algorithm or was one of the best algorithms for $10$ of our $13$ generators (see Table~\ref{tab:reward_generator_best}). It was the only algorithm that was the unique best choice for any generator, taking this role for generators D1, D4, and D9. Furthermore, \q\ also belonged to the set of best algorithms for generators D2, D3, D7, D10, D11, D12 and D13. While \q\ most frequently was a member of a generator's best algorithm set, \fp\ and \determined\ were also frequently in these sets ($6$ and $7$ generators respectively).

\begin{zfigure}[t]
\begin{minipage}{0.46\linewidth}
\begin{table}[H]
		\begin{tabular}{lp{.92\linewidth}}
		\textbf{Gen} & \textbf{Set of Best Algorithms}\\\hline
		D1	&	\q\\
		D2	&	\q, \rvs\\
		D3	&	\awesome, \determined, \fp, \gsa, \meta, \q, \rvs\\
		D4	&	\q\\
		D5	&	\determined, \fp\\
		D6	&	\awesome, \determined, \fp\\
		D7	&	\gsa, \q, \rvs\\
		D8	&	\determined, \fp, \meta\\
		D9	&	\q\\
		D10	&	\fp, \q\\
		D11	&	\determined, \fp, \meta, \q\\
		D12	&	\determined, \q\\
		D13	&	\awesome, \determined, \gsa, \q, \rvs\\			
		\end{tabular}
		\caption{The set of best algorithms for each generator.}
				\label{tab:reward_generator_best}
\end{table}
\end{minipage}
\hfill
\begin{minipage}{0.45\linewidth}
		\includegraphics[width=1\textwidth]{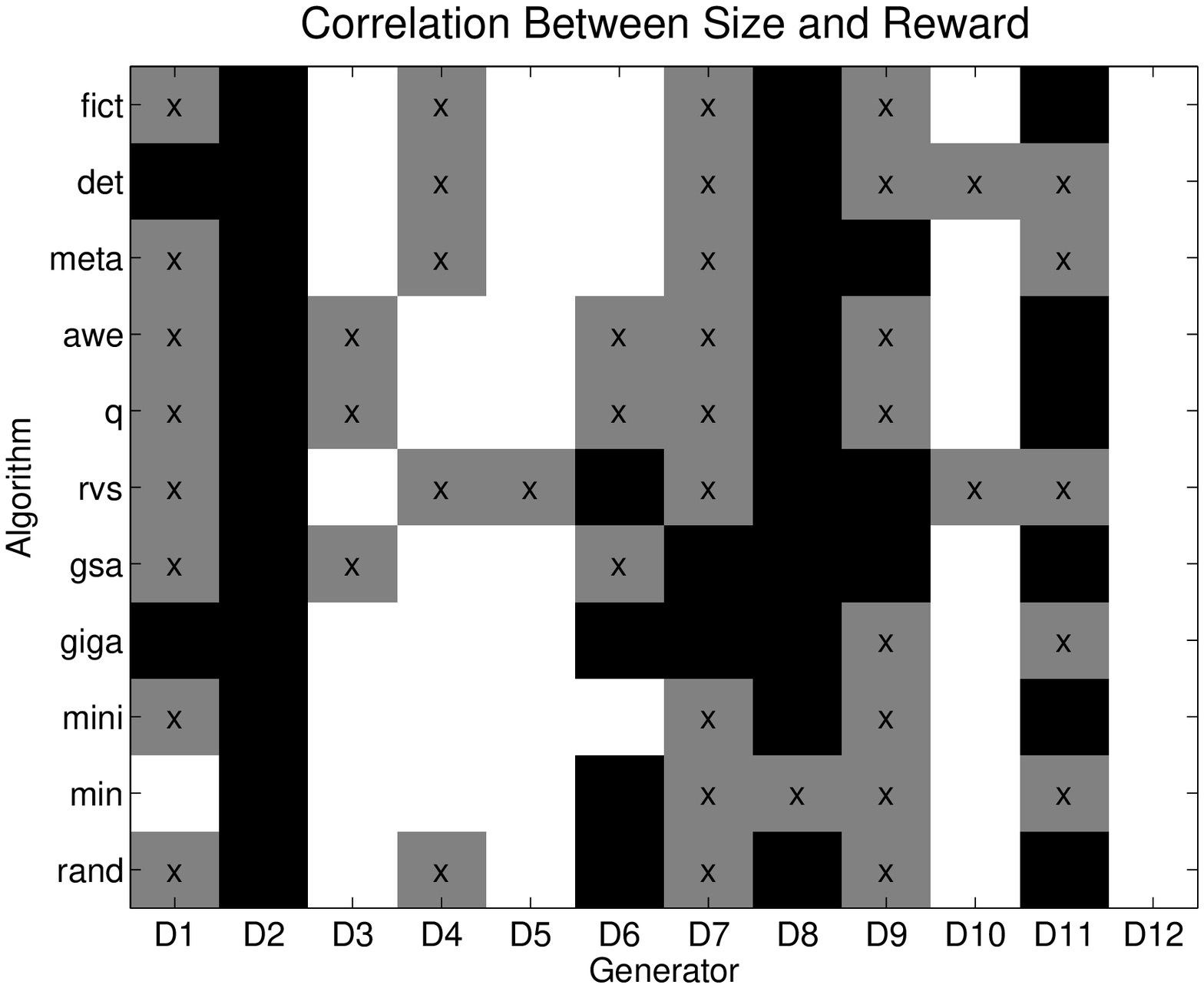}
	\caption[Correlation between size and reward]{A heatmap summarizing the correlations between size and reward for different agents on different generators. A white cell indicates positive correlation, a black cell indicates negative correlation, and a gray cell with an `x' indicated an insignificant result.}
		\label{fig:ag_heatmap_spearman_utility_size}
\end{minipage}
\end{zfigure}

The gradient algorithms were especially strong on D7; indeed, this was the only generator for which all three gradient algorithms were in the best algorithm set. 
D5, D6, and D8 were interesting distributions for \awesome\ and \meta. In D5, neither \awesome\ nor \meta\ managed to be one of the best algorithms despite the fact that both \fp\ and \determined---two of the algorithms that they manage---were. In D6, \awesome\ joined \fp\ and \determined\ but \meta\ did not, and in D8 the reverse happened: \meta, \fp\ and \determined\ were the three best algorithms. These three generators illustrate situations where portfolio algorithms failed to capitalize on one of their managed algorithms. It would be interesting to run further experiments to determine why this occurred and if it could be remedied. \note{Asher: It is interesting that Giga-WoLF, as a gradient descent algorithm similarish to  RV and GSA, does not appear to do as well, especially as per Table XI tab: reward generator best}

%



For all but one of our generators (D13: $2 \times 2$ games) we generated games of varying sizes. Now we consider how the size of a game's action set affected performance. Our hypothesis was that larger action spaces raise the possibility of more complicated game dynamics, and that such complex dynamics can slow learning. Thus, we expected to see average reward decreasing as the size of the game grew. 

\begin{ob}
There was no general relationship between game size and reward: for some generators, algorithms achieved higher rewards on larger games, and for other generators algorithms achieved higher rewards on smaller games.
\end{ob}

Our experiment showed that this intuition did not always hold. First, for many algorithms on many generators we could not reject the null hypothesis of the Spearman rank correlation test---that there was no significant correlation between size and performance---at a significance level of $\alpha = 0.05$. For instance, in D7 only \gsa\ and \giga\ had significant trends (both exhibited negative correlation; reward was lower in larger games).
Second, even when a significant correlation did exist, it was not always negative. We did observe that for most distributions, significant correlations were either entirely negative or entirely positive. For D2, D7, D8, D9, and D11 the correlation was negative; for D3, D4, D5, D10, and D12 it was positive. D1 and D6 exhibited both kinds of correlation for different algorithms. \note{Asher: Would it not be better to mention with the majority for D1, D7, D9 as being insignificant correlation?}
\note{EPZ: Not sure that I follow: I already called out that many of the distributions did not have many significant relationships...}

Overall, the relationship between game size and reward appears to depend strongly on the choice of generator. It could be the case that when the action spaces increase in size, important game features tied with high reward become more common, or it could be that larger action spaces make it easier for MAL algorithms to miscoordinate, which is desirable for some games. Indeed, D4---\zgame{Dispersion Games}--- show positive correlation between the number of actions and reward, and this is a game where agents need to miscoordinate to do well.

As Figure~\ref{fig:ag_heatmap_spearman_utility_size} shows, D2 and D12 were the only two distributions on which we could reject the null hypothesis for all algorithms, and they supported opposite conclusions. On instances from D2, correlation was completely and strongly negative: the larger the game, the worse every algorithm performed. The least correlated algorithm was \random\ with a Spearman's coefficient of correlation $\rho = -0.329$. Correlation was entirely positive for D11, but some of the coefficients were smaller. \FP\ was the least sensitive to size ($\rho = 0.07$), but it was anomalous. The algorithm with the next smallest coefficient was \giga, with $\rho = 0.267$.

\subsection{Per-Opponent Average Reward and the Algorithm Game}
\label{sec:block-opponent}

We now consider each algorithm's average reward on a per-opponent basis. 

\begin{ob}
Al\-gor\-ithm per\-for\-mance de\-pended sub\-stant\-ially on which op\-po\-nent was played.
\end{ob}

\begin{zfigure}[t]
\begin{minipage}[t]{0.45\linewidth}
		\includegraphics[width=1\textwidth]{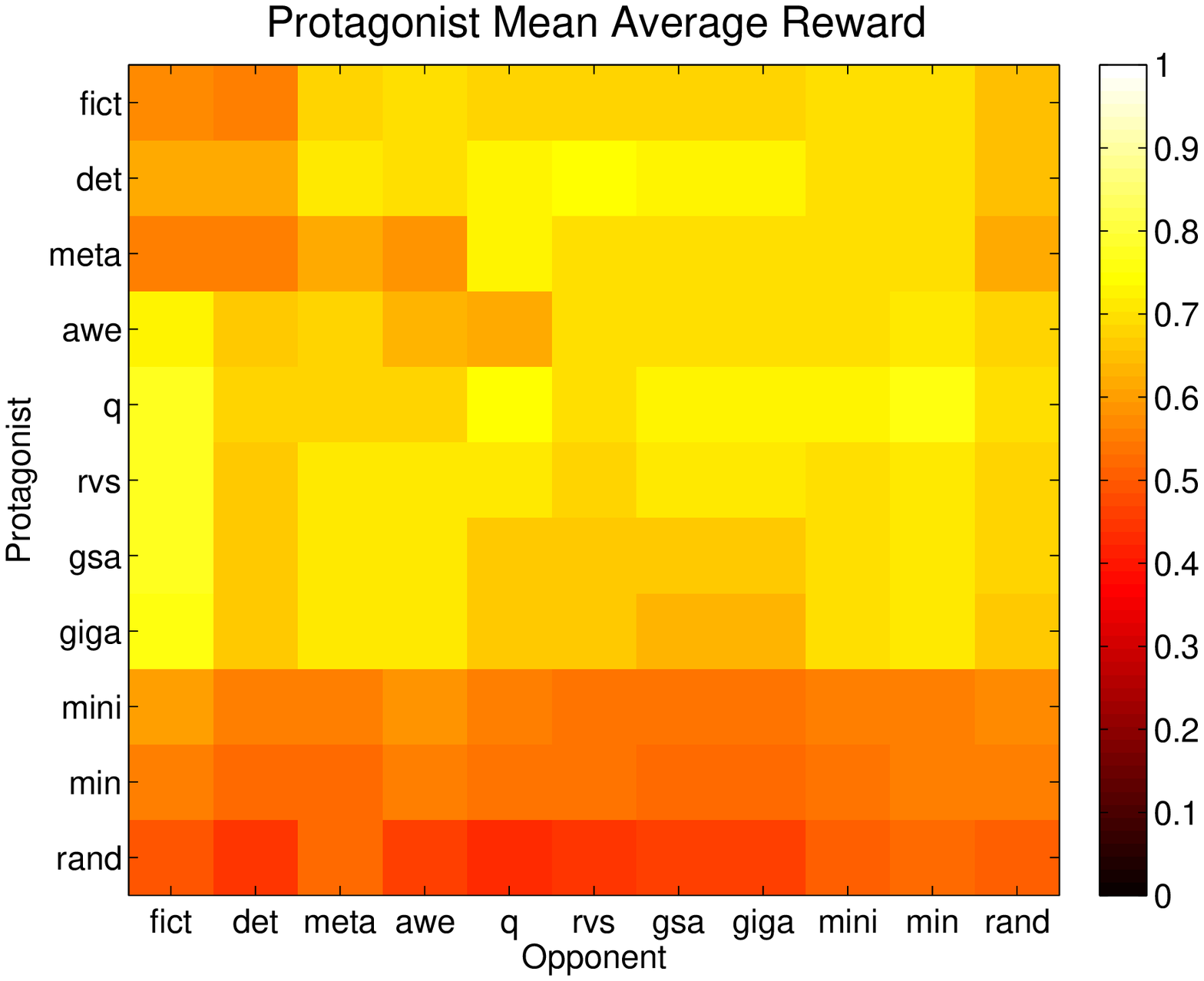}
	\caption[A heatmap showing the mean reward for each algorithm]{A heatmap showing the mean reward for each protagonist algorithm (ordinate) playing against each opposing algorithm (abscissa).}
		\label{fig:aa_heatmap_utility_grand}
\end{minipage}
\hfill
\begin{minipage}[t]{0.45\linewidth}
	\includegraphics[width=1\textwidth]{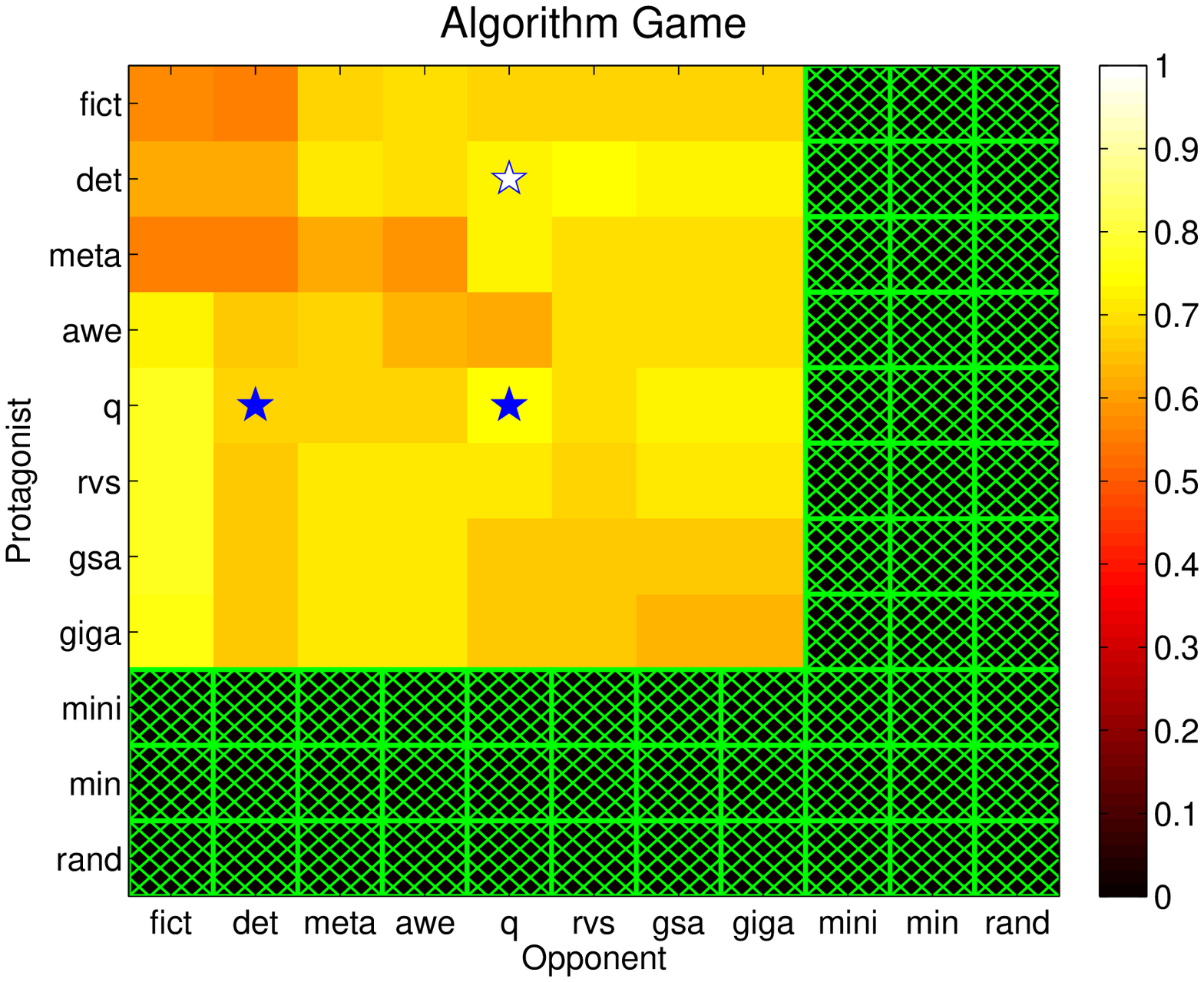}
	\caption[The mean reward algorithm game]{Interpreting the mean reward results as a one-shot game. The cells that are cross-hatched are dominated and the `$\star$'s indicate pure-strategy Nash equilibria. Because the \determined{} vs. \q\ equilibrium is asymmetric, it appears twice. To indicate this, we make one of the corresponding stars hollow.}
		\label{fig:reward_algorithm_game}
\end{minipage}
\end{zfigure}

Figure~\ref{fig:aa_heatmap_utility_grand} shows the mean reward achieved by each algorithm against every possible opponent. One striking fea\-ture of this fig\-ure is that \minimax, \minimaxidr\ and \random\ were all relatively weak against a broad range of opponents. We also observe that \fp\ and \determined\ tended to get lower reward in self play and against each other than against other opponents. \Meta---an algorithm that manages a profile of algorithms including \fp\ and \determined---also appear to have inherited these performances issues, while \awesome---the other portfolio algorithm---substantially avoided them.

If we know what algorithm the opponent is using, which algorithm should we use? We constructed ``best-response sets'' for each possible opponent using bootstrapped percentile intervals. We call the algorithm with the highest mean against a particular opponent a best response, but also assign any algorithm with an overlapping bootstrapped $95\%$ percentile interval to the set---we cannot claim that these algorithms do significantly worse than the apparent best algorithm. These best response sets are summarized in Table~\ref{tab:reward_best_response}. \q\ and \rvs\ were most frequently best responses, while \fp, \meta, \minimax, \minimax\ and \random\ were never best responses.

\begin{kfigure}[t]
\begin{minipage}[t]{0.55\linewidth}
\begin{table}[H]
		\begin{tabular}{p{.3\linewidth}p{.6\linewidth}}
			\textbf{Opponent} & \textbf{Best-Response Set}\\\hline
			\awesome & \giga, \gsa\, \rvs\\
			\Determined & \awesome, \giga, \gsa,\\
			 & \q\, \rvs\\
			\FP & \gsa, \q\, \rvs\\
			\giga & \determined, \q\, \rvs\\
			\gsa & \determined, \q\, \rvs\\
			\Meta & \determined, \giga, \gsa\, \rvs\\
			\Minimax & \q\\
			\Minimaxidr & \q\\
			\Q & \determined, \q\, \rvs\\
			\Random & \determined, \q\, \rvs\\
			\rvs & \determined\\
		\end{tabular}
		\caption{The different algorithms and their best-response sets}
				\label{tab:reward_best_response}
\end{table}
\end{minipage}
\hfill
\begin{minipage}[t]{.4\linewidth}
\begin{table}[H]
		\begin{tabular}{lll}
			\textbf{Algorithm} & \textbf{SD} & \textbf{WD}\\\hline
			\awesome 		& $10.8\%$ 	& $11.7\%$\\
			\Determined & $0.0\%$	 	& $0.0\%$\\
			\FP 				& $35.9\%$ 	& $36.4\%$\\
			\giga 			& $54.1\%$ 	& $55.1\%$\\
			\gsa 				& $0.4\%$		& $0.4\%$\\
			\Meta 			& $28.8\%$ 	& $28.2\%$\\
			\Minimax 		& $100.0\%$ & $100.0\%$\\
			\Minimaxidr & $100.0\%$ & $100.0\%$\\
			\Q 					& $0.0\%$ 	& $0.0\%$\\
			\Random 		& $100.0\%$ & $100.0\%$\\
			\rvs 				& $0.0\%$ 	& $0.0\%$\\
		\end{tabular}
		\caption{The proportion of subsampled algorithm games in which each algorithm was strictly dominated (SD) or weakly dominated (WD).}
				\label{tab:domination}
\end{table}
\end{minipage}
\end{kfigure}

One interesting way to interpret these best response results is to consider the one-shot ``algorithm game'': a single-shot normal-form game in which the actions correspond to our 11 algorithms and the payoff for using algorithm $A$ against algorithm $B$ is the mean reward that algorithm $A$ attained against $B$. There were three algorithms that were strictly dominated in this grand distribution algorithm game: \minimax, \minimaxidr\ and \random. Strict domination of algorithm $A'$ by $A$ means that regardless of what algorithm the opponent selects, $A$ is always a better choice than $A'$.
As with best responses, we required domination to be significant: we wanted to be confident that if the experiment were repeated, we would get a similar result. We used bootstrapping to check this, subsampling $6~600$ PSMs $10~000$ times and from these forming $10~000$ `subsampled' games. We checked for strict domination in each game, and considered an algorithm dominated if it was dominated in at least $95\%$ of the subsampled games. The proportion of subsampled algorithm games in which each algorithm was dominated is shown in Table~\ref{tab:domination}; we also distinguish strict domination from weak domination.

\begin{ob}
\Determined\ and \q\ were the only algorithms to participate in pure-strategy Nash equilibria of the algorithm game.
\end{ob}

Only two pure-strategy Nash equilibria ever occurred in the subsampled games for the grand distribution: \q\ in self play, and \q\ against \determined.  The \q--\q{} equilibrium is particularly convincing because it is symmetric and so does not require that the players coordinate to playing different strategies, and furthermore because it occurred in $90.2\%$ of the subsampled games. The other equilibrium occurred in the remaining $9.8\%$ of games. (Because both equilibria involved \q, we did not observe them together in the same subsampled games.)

We looked more deeply into the algorithm games by restricting attention to individual generators. The generators varied substantially in their pure-strategy Nash equilibria. Overall, \Determined\ in self play constituted the most common symmetric pure-strategy Nash equilibrium. It was a significant Nash equilibrium for seven of the generators. (That is, \determined\ in self play was a pure-strategy Nash equilibrium in more than $95\%$ of the subsampled games for each of these generators.) \Q\ in self play was the second most common symmetric pure-strategy Nash equilibrium, arising in the algorithm games for four generators.

Generators also differed substantially in their \emph{number} of pure-strategy Nash equilibria. For instance D1 (\zgame{A Game with Normal Covariant Payoffs}) had no significant pure-strategy Nash equilibrium.  D4 (\zgame{Dispersion Game}), at the other extreme, had $22$ pure-strategy Nash equilibria (see Figure~\ref{fig:reward_algorithm_game_d4}). Part of the reason for the large number of equilibria in $D4$ was that a majority of runs for many of the algorithms yielded a reward of $1$ (e.g., $84.6\%$ of \awesome's runs yielded a reward of $1$). This meant that in many of the subsampled games, the majority of payoffs were exactly $1$ and so there were many weak Nash equilibria. For example, both \rvs\ and \q\ attained a reward of $1$ against \fp, and \fp\ itself attained a reward of $1$ against \rvs\ and \fp. Therefore both \rvs--\fp\ and \q--\fp\ were pure Nash equilibria.

\begin{zfigure}[t]
\begin{minipage}[t]{0.45\linewidth}
			\includegraphics[width=1\textwidth]{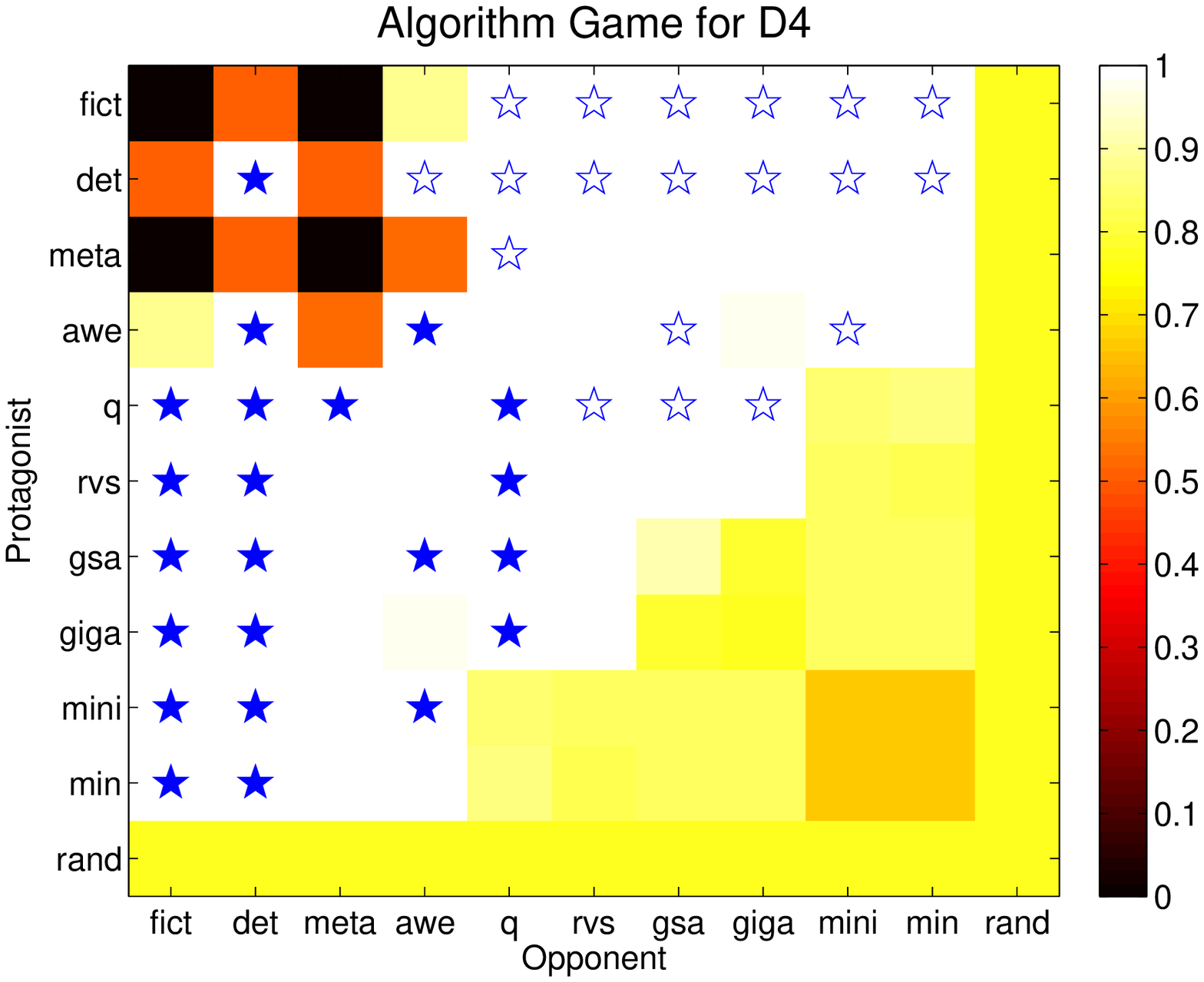}
	\caption[The mean reward algorithm game for D4]{Interpreting the mean reward results for D4 (\zgame{Dispersion Game}) as a one-shot game. No cells were dominated; the `$\star$'s indicate pure-strategy Nash equilibria. Asymmetric equilibria appear twice; to indicate this we make one of the corresponding stars hollow.}
		\label{fig:reward_algorithm_game_d4}
\end{minipage}
\hfill
\begin{minipage}[t]{0.45\linewidth}
		\includegraphics[width=1\textwidth]{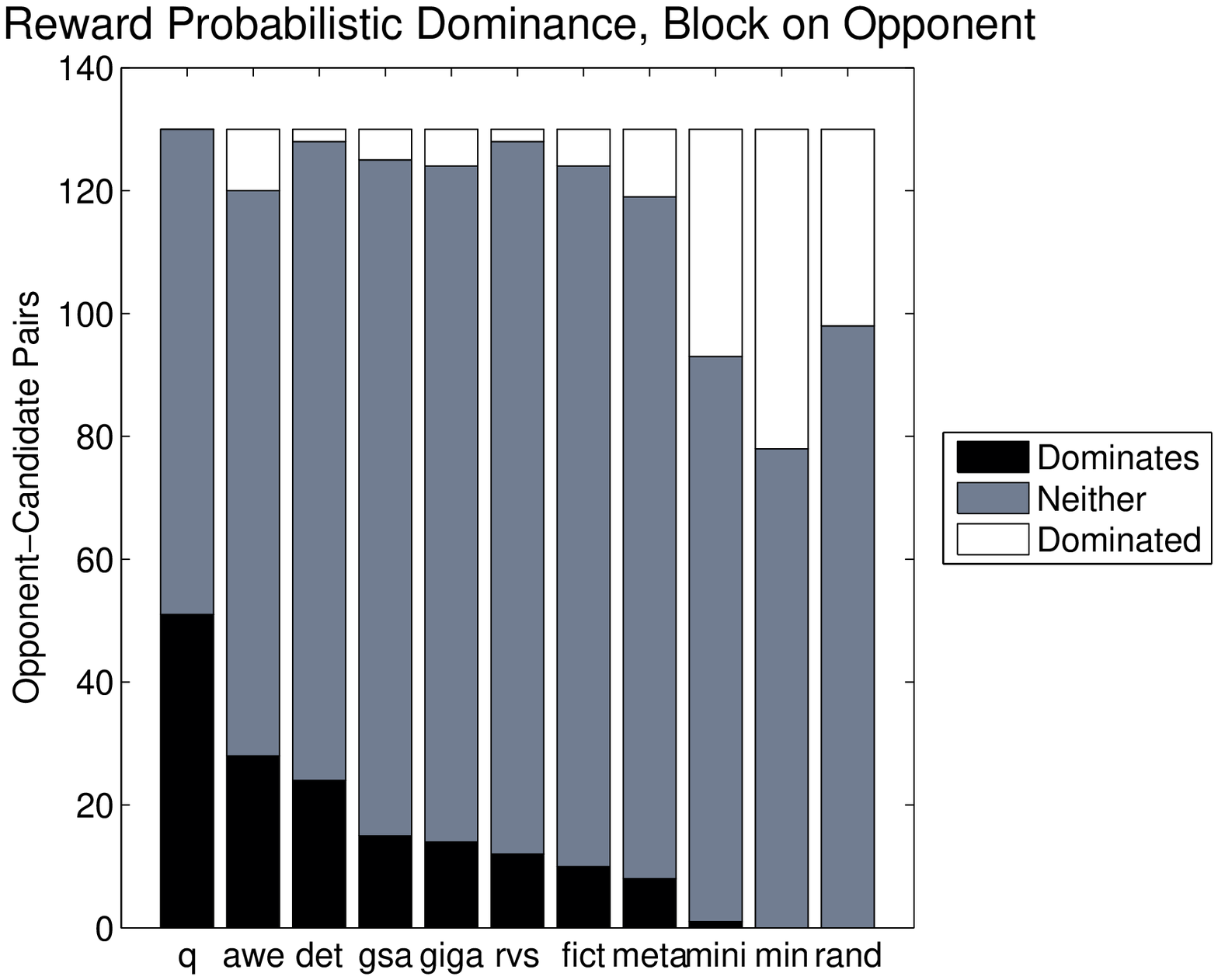}
	\caption[Reward distribution probabilistic domination between the algorithms]{For each algorithm, the number of opponents and candidate algorithms the algorithm dominated, was dominated by, or neither.}
		\label{fig:bar_utility_dominance_opponent}
\end{minipage}
\end{zfigure}

\subsection{Probabilistic Domination of One Algorithm by Another}
\label{sec:reward-probdom}

Now we consider the following question: given a fixed opponent, is a given algorithm probabilistically dominated by any alternative algorithm in terms of average reward?

\begin{ob}
\Q\ was the only algorithm that was never probabilistically dominated by any other algorithm when playing any opponent.
\end{ob}

\Q{} had the best performance in terms of probabilistic domination. \Determined\ and \rvs\ were the next-least-dominated algorithms: \determined\ was only probabilistically dominated by \awesome\ against a \fp\ opponent, which was in turn dominated by \q; \rvs\ was dominated by \q\ when playing against the \minimax\ variants, and also by \determined\ when playing against \rvs. On the whole, domination by another algorithm in self play was a common trend; only \awesome, \determined\ and \q\ avoided being dominated by another algorithm when playing themselves. It is interesting that \determined\ was not dominated: we see this as a property of the specific game distributions that we studied.

Overall, while we observed some strong domination relationships, these were the exceptions while ambiguity was the rule. For most algorithm pairs against most opponents, no probabilistic domination relationship existed (see Figure~\ref{fig:bar_utility_dominance_opponent}). Furthermore, there was no opponent for which one algorithm probabilistically dominated all others.

\subsection{Self Play}
\label{sec:selfplay}

We have already seen evidence that self-play was challenging for many algorithms (e.g., see the tendency towards `cool' cells on the main diagonal of Figure~\ref{fig:aa_heatmap_utility_grand}). A closer analysis shows that for most algorithms there was indeed a significant relationship between self play and low reward.

\begin{ob}
Most algorithms attained lower average reward in self play.
\end{ob}

The distributions of reward in self-play runs for \awesome, \determined, \fp\ and \meta\ were probabilistically dominated by the distribution of reward in non-self-play runs.
While the same was not true for the gradient algorithms (they achieved fewer low-reward runs in self play), their self-play means were nevertheless significantly lower than their non-self-play means. We verified this by looking at the $95\%$ bootstrapped percentile intervals. There was no significant relationship for \minimax\ and \minimaxidr, and this self-play trend was reversed for \q: its self-play runs probabilistically dominated its non-self-play runs. Furthermore, \q\ achieved a higher mean reward in self play than any other algorithm (see Figure~\ref{fig:reward_sp_mvplot}).

\begin{zfigure}[t]
\begin{minipage}[t]{0.45\linewidth}
	\includegraphics[width=\textwidth]{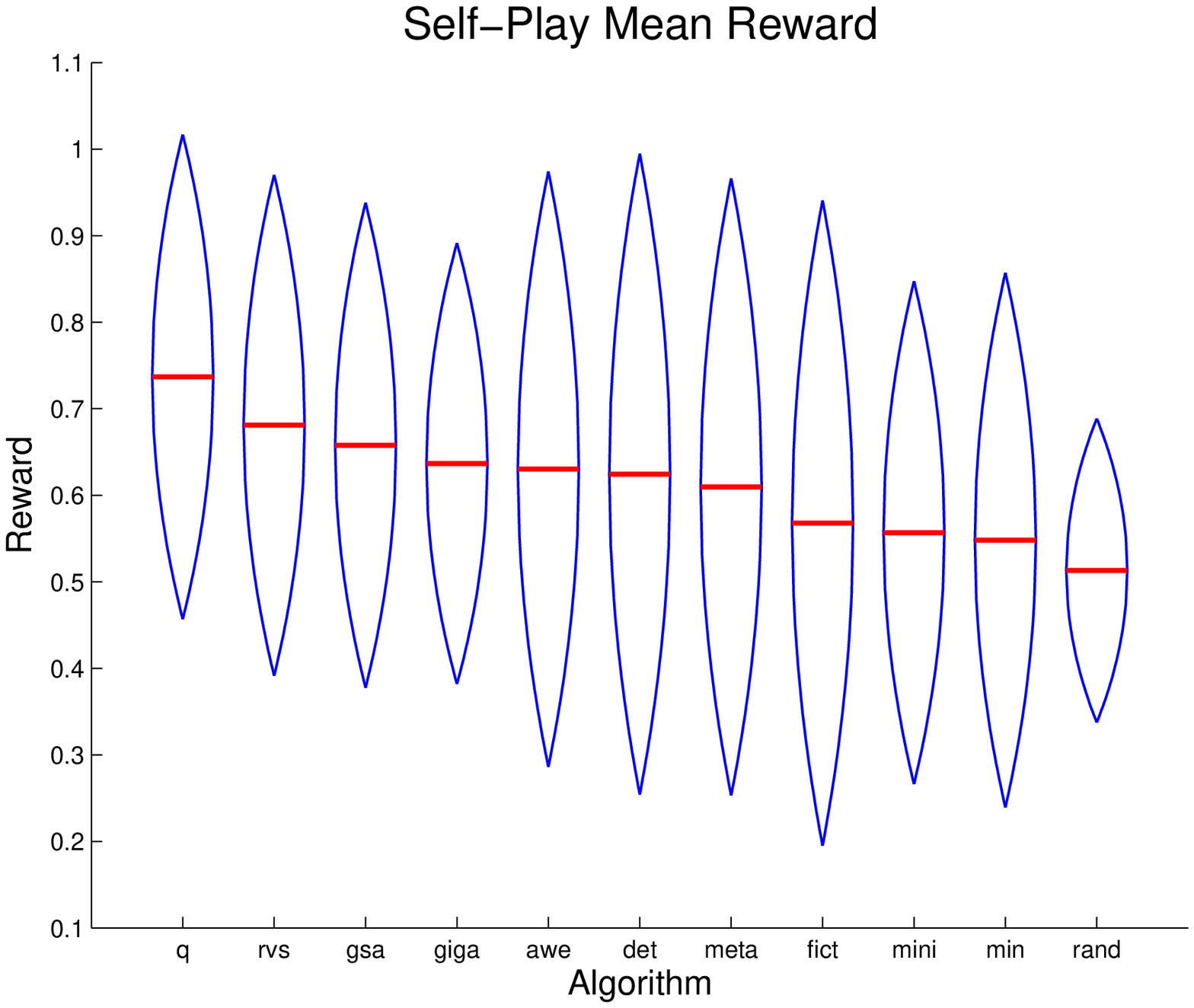}
	\caption[Self-play lens plot for reward]{A plot that shows the mean reward (bar) for each algorithm in self play and one standard deviation in either direction (the size of the lens).}
		\label{fig:reward_sp_mvplot}
\end{minipage}
\hfill
\begin{minipage}[t]{0.45\linewidth}
	\includegraphics[width=\textwidth]{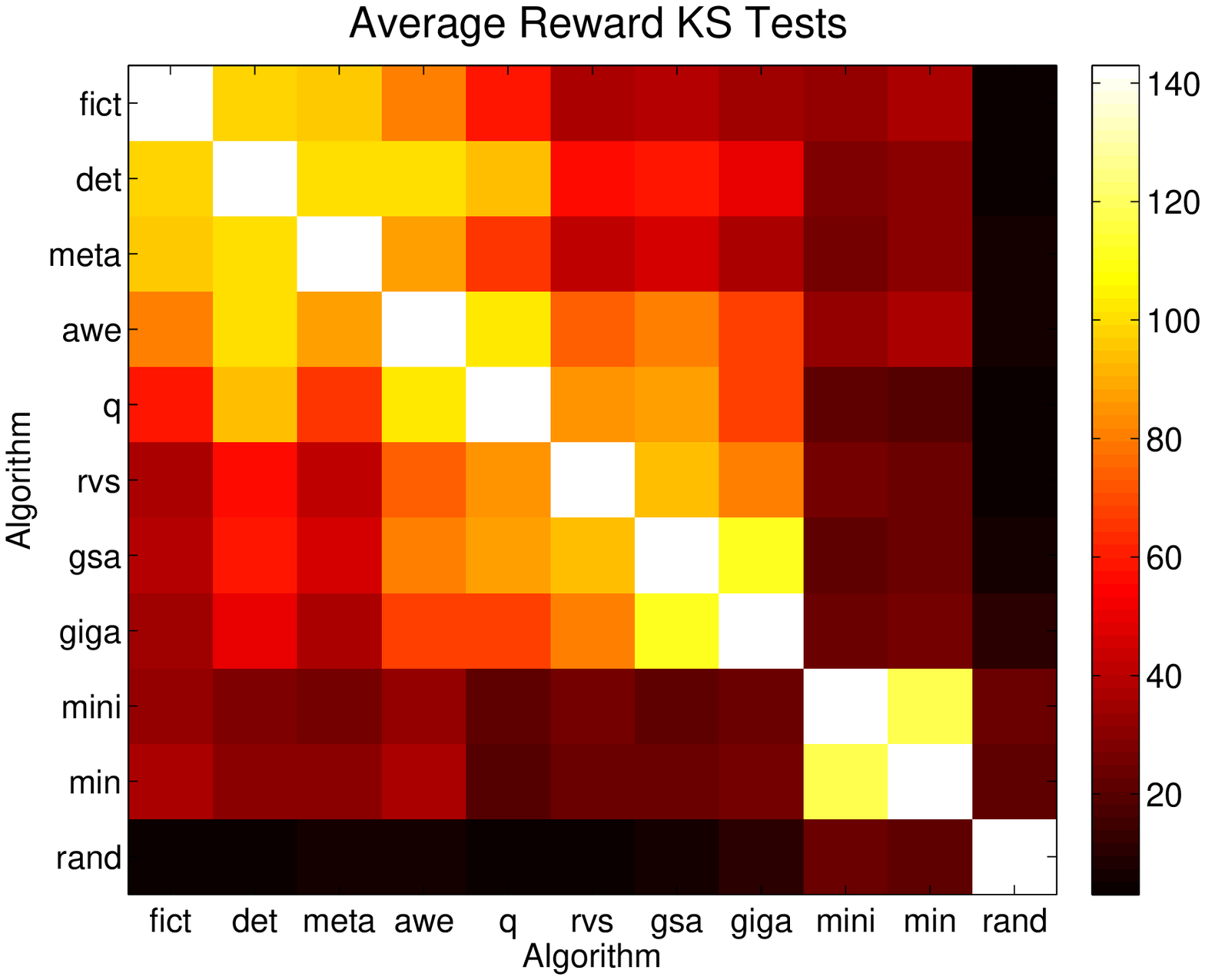}
	\caption[Similarity among different algorithms]{A heatmap that summarizes the number of opponent/generator pairs two algorithms are similar on in terms of reward distribution. This relationship is symmetric, so only the lower half of the plot is presented. The hotter the cell, the more situations the two algorithms are similar in.}
\note{I think this figure would be much easier to see if it showed the whole square, even though it's symmetric. (This would make it much clearer how much of the density lives on the diagonal.) Can you try it and see?}
		\label{fig:utility_heatmap_ks}
\end{minipage}
\end{zfigure}
		
Interestingly, \awesome\ was one of the algorithms with poorer self-play runs, despite its machinery for converging to a special equilibrium in self play. We wondered whether this occurred because \awesome\ did not converge due to an overly-conservative threshold for detecting whether its opponent was playing part of an equilibrium, or because \awesome\ did converge to the special equilibrium, but that equilibrium did not offer high reward. (Note that our implementation of \awesome\ coordinates to the first Nash equilibrium found by GAMBIT's implementation of Lemke-Howson.) At the risk of keeping the reader in suspense, we defer the answer to \S\ref{sec:ne_convergence}, in which we examine equilibrium convergence results.

\subsection{Algorithm Similarity}
\label{sec:alg-similarity}

Finally, we investigate similarities between algorithms' abilities to achieve high reward.
We can assign some of our algorithms to one of three major blocks. First, \awesome\ and \meta\ are similar because they both manage portfolios incorporating \fp\ and \determined; likewise, we expect them to be similar to the \fp\ and \determined\ algorithms themselves. Second, \giga, \gsa\ and \rvs\ are similar because they all follow a reward gradient. Finally, \minimax\ and \minimaxidr\ are similar because the latter is the same as the former except for the addition of an IDR preprocessing step. We call these the portfolio, gradient, and minimax blocks. We also might suspect that \q, an algorithm that does not explicitly model the opponent, might be similar to the gradient algorithms. \note{so why not put it in that block?} Nevertheless, we do not assign \q\ to a block; likewise, we leave \random\ unassigned.

We tested all pairs of algorithms for similarity by comparing their average reward distributions for all generator--opponent pairs. Thus, we tested each algorithm pair $13 \times 11 = 143$ times---every algorithm is of course similar to itself and so we did not check these cases. Failing to reject the null hypothesis of the KS test (that both samples were drawn from the same population) is some evidence for the samples being similar.
\note{Why not use the methods discussed in the convergence section to test for two distributions being the same? It's odd to go to all that care in that section, but here to treat failure to reject a null hypothesis as passing the test, for the same kind of claim. This subsection comes across as the only real technical weakness in the paper.}
This rough-and-ready approach does not establish significant similarity and is merely suggestive of similarity; failing to reject a null hypothesis is not the same as having shown that the null hypothesis is true. However, with this caveat in mind, we observed some interesting trends.

\begin{ob}
Similar algorithms tended to exhibit similar per\-for\-mance.
\end{ob}

All three predicted blocks emerge, as can be seen in Figure~\ref{fig:utility_heatmap_ks}. First, \meta, \awesome, \fp\ and \determined\ were all similar to each other on a large number of opponent--generator pairs. Both \meta\ and \awesome\ were similar in more cases to \determined\ than to \fp. For instance, \awesome\ was similar to \determined\ in $101$ out of $130$ cases while similar to \fp\ in only $81$ cases. \Meta\ and \awesome\ were also quite similar to each other ($88$ cases). \q\ was similar to the algorithms in this block, especially \determined\ and \awesome, which we had not expected. \awesome\ was more similar to \q\ than to any other algorithm: they were similar in $103$ cases, while even \determined\ and \awesome\ were only similar in $101$ cases.

The block of algorithms consisting of \rvs, \giga\ and \gsa\ were all similar in a large number of cases, with a particularly tight relationship evident between \giga\ and \gsa\ (similar in $111$ cases). \Q\ also bore similarities to the gradient-algorithm block. \note{Stronger or weaker than its similarities discussed above?}These algorithms also showed somewhat weaker similarity to \determined\ and \awesome.

The connection between \minimax\ and \minimaxidr\ was particularly strong (similar in $118$ cases). These were also the algorithms most similar to \random---indeed, similar almost twice as often as the next-most-similar algorithm (\awesome: it was similar to \random\ in $11$ cases, as compared to \minimax's $21$ cases). 

\section{Empirical Evaluation of MAL Algorithms: Other Metrics}
\label{sec:other-metrics}
So far, all of our experimental discussion has concerned the average reward metric. However, a wide variety of other metrics have also been proposed and studied in the literature. Here we consider many of the most prominent. This allows us to understand our experimental results in different ways, and furthermore sheds light on the extent to which each metric correlates with high reward in practice. In \S\ref{sec:regret} we investigate regret, specifically considering mean regret, probabilistic domination of one algorithm by another, and the relationship to reward. In \S\ref{sec:stationarity} we assess algorithms' tendencies to converge to stationary strategies. \S\ref{sec:ne_convergence} considers convergence to Nash equilibrium of the stage game, and relates this metric to reward. In \S\ref{sec:maxmin} we consider algorithms' abilities to achieve at least their maxmin payoffs, and consider both per-opponent maxmin performance and the relationship to reward. Finally, in \S\ref{sec:rne_convergence}, we measure algorithms' tendency to converge to payoff profiles consistent with Nash equilibria of the infinitely-repeated stage game.

\subsection{Regret}
\label{sec:regret}
Regret is the difference between the reward that an algorithm could have received by playing the best static pure strategy and the reward that it did receive:  
\begin{equation}
Regret(\vec{\sigma}_i,\vec{a}_{-i}) = \max_{a \in A_{i}}\sum_{t=1}^{T}\left[\,r(a,a_{-i}^{(t)}) - \E{\,r(\sigma_{i}^{(t)}, a_{-i}^{(t)})\,}\right].
\label{eqn:regret}
\end{equation} 
The best static pure strategy is determined after the run, based on the assumption that the opponent's strategy in each round would not change. We use the strategy formulation of regret---as opposed to one that uses the sampled actions that the algorithm played---following \inlinecite{bowling04}. Rather than looking at the total sum of regret over all $10~000$ recorded iterations, we will discuss the mean regret over these iterations. Since player payoffs are restricted to the $\left[0,1\right]$ interval, mean regret can give a better sense of the magnitude of regret with respect to possible reward.
%
%

Regret has been suggested as a measure of how exploitable an algorithm is. If an agent accrues significant regret one possible explanation is that it did the wrong thing. However, in some games (e.g., the Traveler's dilemma) ignoring regret can lead to greater long-term reward.

Some algorithms, including \giga\ and \rvs, are \textit{no-regret} learners: they come with the guarantee that they will always approach zero regret as the number of iterations approaches infinity. However, to our knowledge it has not been shown experimentally how the regret achieved by these algorithms compares to the regret achieved by other algorithms that lack such a guarantee; nor has it been demonstrated whether these algorithms achieve better than zero regret in practice.

\begin{ob}
\Q\ best minimized regret. \giga\ most frequently achieved negative-regret runs.
\end{ob}

\begin{zfigure}[t]
\begin{minipage}[t]{0.45\linewidth}
	\includegraphics[width=1\textwidth]{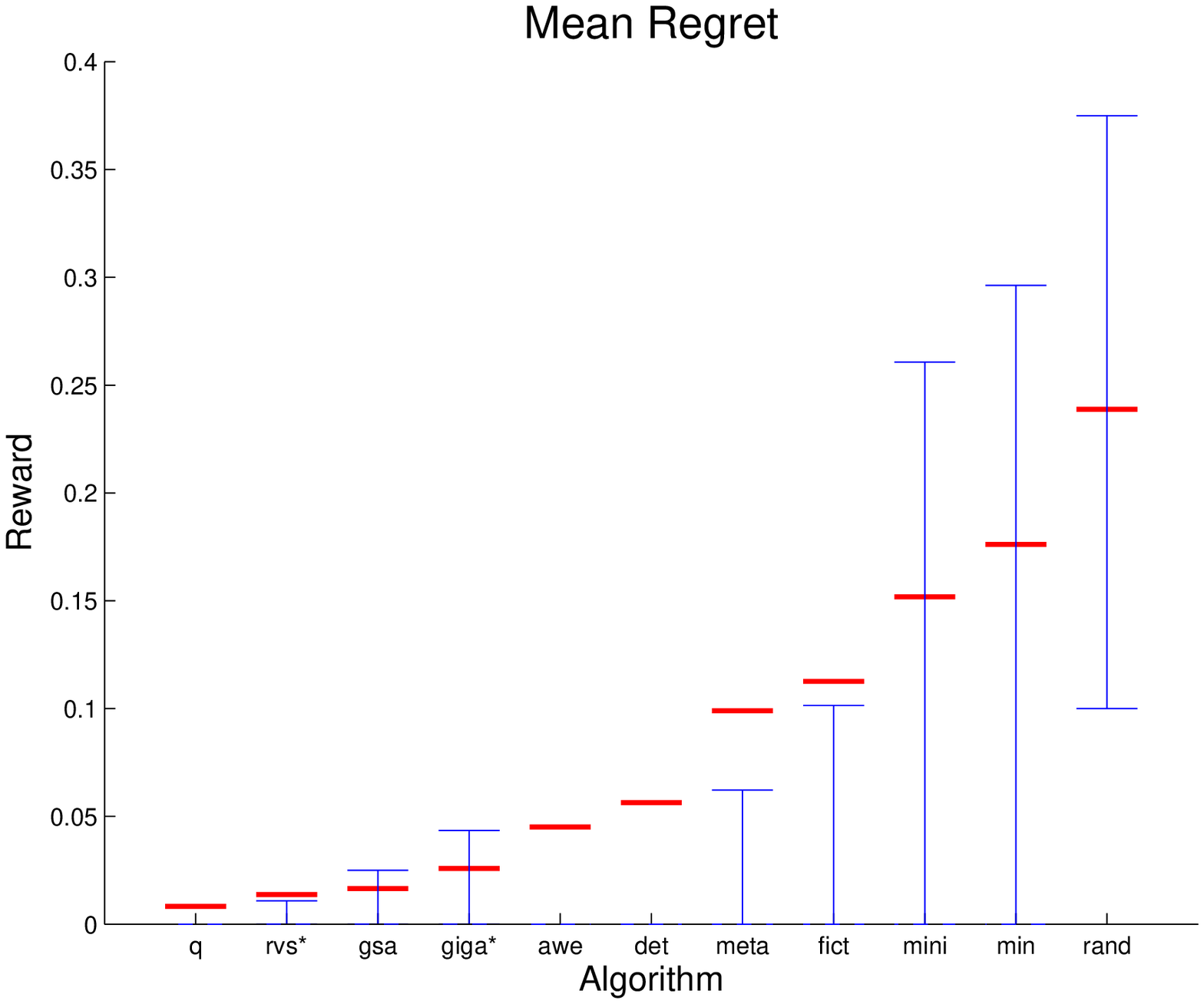}
\note{How reasonable is a lens plot here, as compared to a box plot? It seems likely that variation around the mean was dramatically asymmetric.}
\note{EPZ: While more accurate: I think that this hybrid IQR/mean chart  looks ugly.}
	\caption[Lens plot for regret.]{A plot that shows the mean regret (bar) for each algorithm and the IQR. Note that for highly skewed distributions, the mean need not be contained in the IQR. Algorithms with an asymptotic no-regret guarantee are indicated `$\ast$'.}
	\label{fig:regret_mvplot}
\end{minipage}
\hfill
\begin{minipage}[t]{0.45\linewidth}
	\includegraphics[width=1\textwidth]{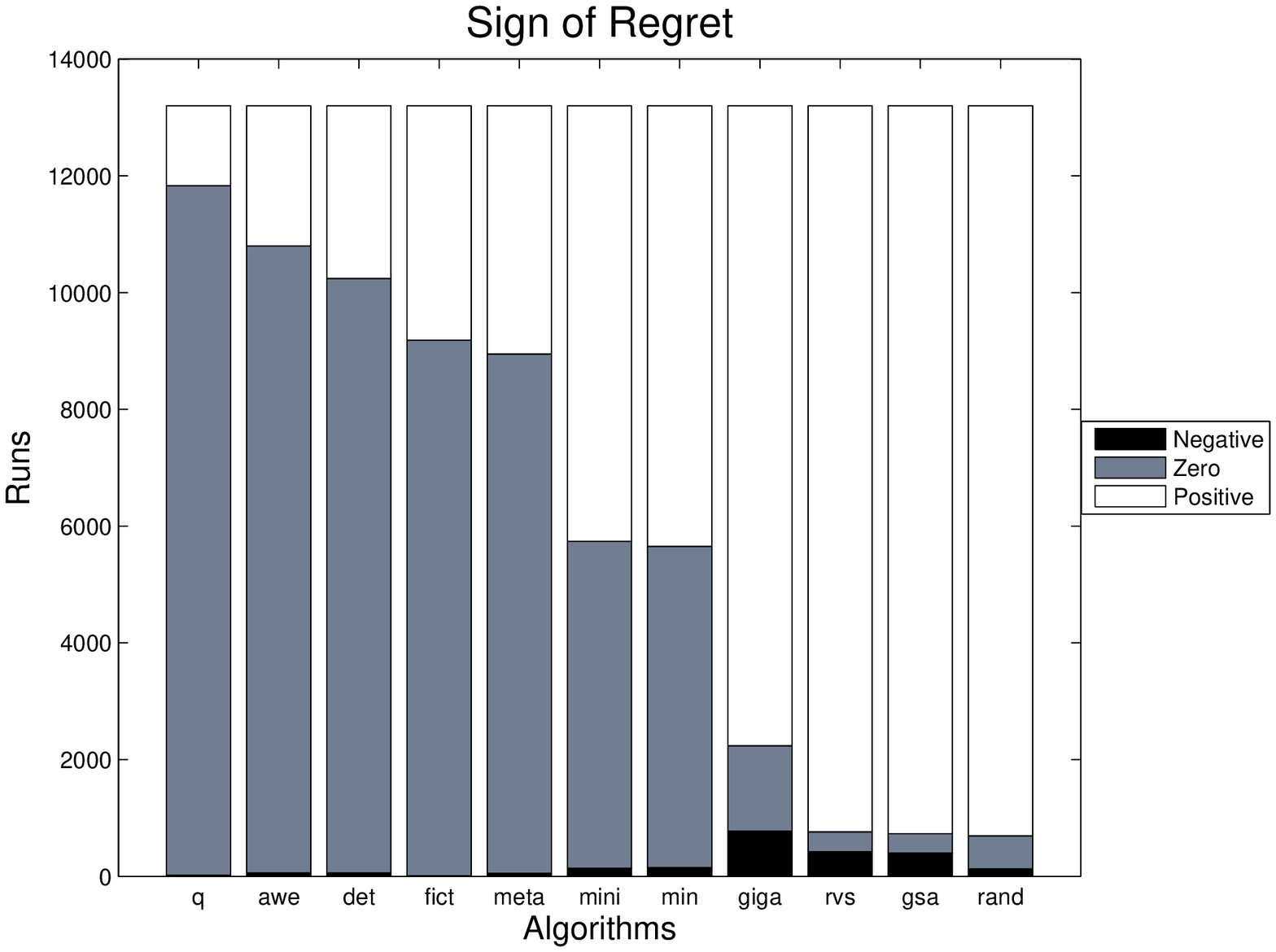}	
	\caption{The number of runs in which each algorithm achieved negative, zero, or positive regret.}
	\label{fig:bar_regret_frequency_grand}
\end{minipage}
\end{zfigure}

In our experiment, all algorithms achieved positive mean regret (Figure~\ref{fig:regret_mvplot}), though they differed substantially in the fraction of their matches in which they achieved positive regret (Figure~\ref{fig:bar_regret_frequency_grand}). All the means were significantly different, based on overlaps in the $95\%$ percentile intervals (there was none). Of these, \q\ had the lowest regret, at $0.008$. The gradient algorithms---\giga, \gsa\ and \rvs---had the next lowest mean regret after \q.  Among the gradient algorithms, \gsa\ achieved the lowest mean regret, followed by \rvs\ and then by \giga. These empirical results are concordant with \giga\ and \rvs's theoretical no-regret guarantees---not only are these algorithms guaranteed zero regret in the limit, but they also achieved low regret in practice. At the same time, it is interesting that the algorithm with the best results, \q, comes with no such guarantee.

Considering only mean regret masks an interesting difference between \q\ and the gradient algorithms: they achieve low mean regret in different ways (see Figures~\ref{fig:bar_regret_frequency_grand} and \ref{fig:ecdf_utility_q_giga_grand}). \Q\ achieved low mean regret by attaining zero regret in most ($89.5\%$) of its runs. It had the fewest positive-regret runs ($10.4\%$; the next lowest was \awesome\ at $18.2\%$), and also had the second-fewest negative-regret runs ($0.1\%$; only \fp\ had (slightly) fewer). On the other hand, the gradient algorithms rarely achieved zero regret (the algorithms with the fewest zero runs were \gsa with $2.49\%$, \rvs with $2.58\%$, \random\ with $4.29\%$ and \giga with $11.10$) but often achieved negative regret (the three algorithms with the most negative regret runs were \giga\ ($5.8\%$), \rvs\ ($3.2\%$) and \gsa\ ($3.0\%$)).

\begin{zfigure}[t]
\begin{minipage}[t]{0.45\linewidth}
	\includegraphics[width=1\textwidth]{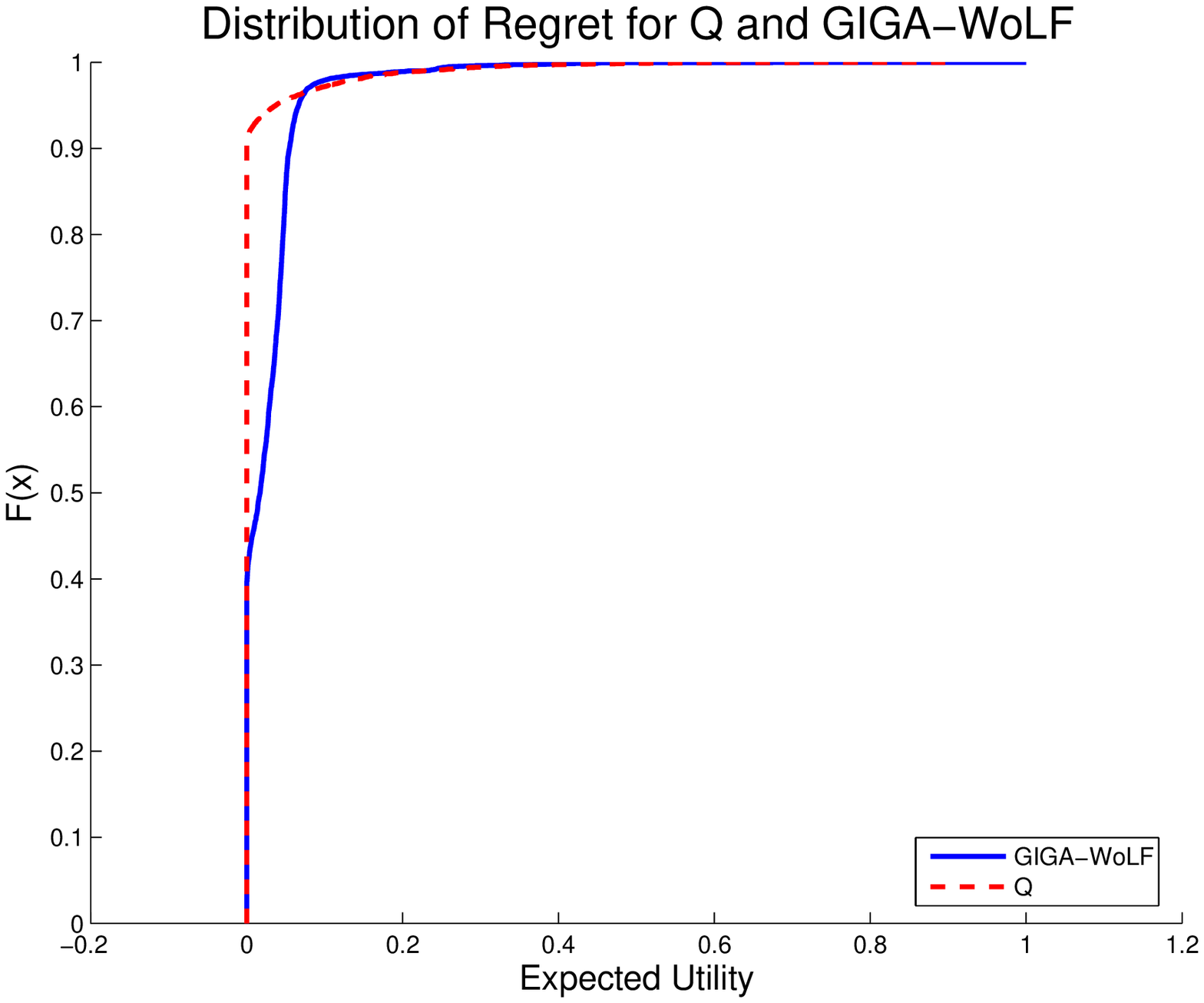}
	\caption{The distribution of regret for \q\ and \giga.}
	\label{fig:ecdf_utility_q_giga_grand}
\end{minipage}
\hfill
\begin{minipage}[t]{0.45\linewidth}
	\includegraphics[width=1\textwidth]{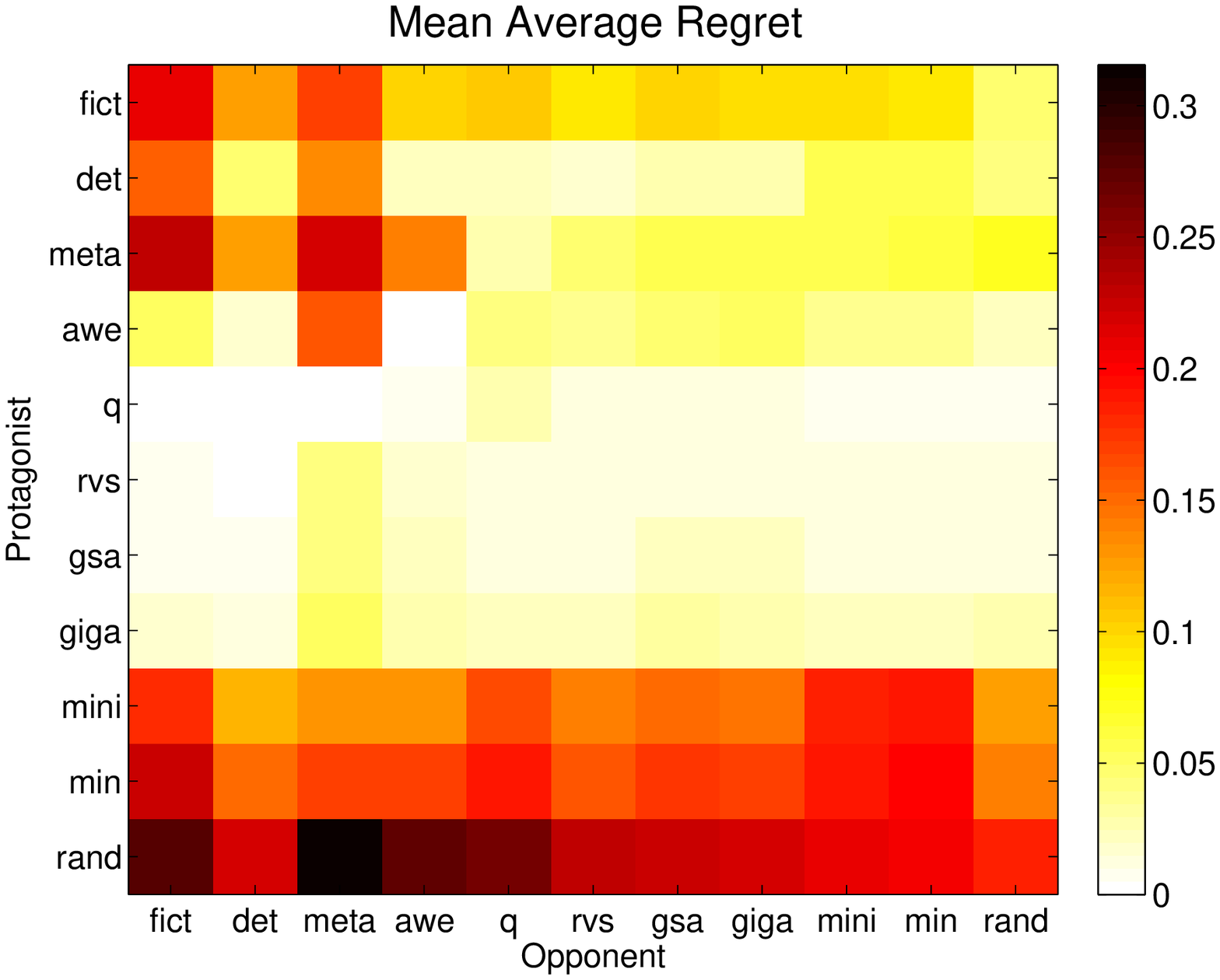}
	\caption{Mean average regret, blocked by opponent.}
	\label{fig:aa_heatmap_regret_mean_grand}
\end{minipage}
\end{zfigure}

Overall, no algorithm achieved less than very slightly negative regret: the run with the smallest regret had an average regret of $-2\times10^{-6}$. \note{So, indeed, a lens plot seems the wrong thing to use here.}The converse was not true for positive regret: in $440$ different runs some algorithm attained average regret of $1$, meaning that it took precisely the wrong action in every round. $48.6\%$ of these runs involved \fp\ or one of the algorithms that wrap around \fp\ (\awesome\ or \meta) in self play, and were on generator D4 (\zgame{Dispersion Games}), which reward miscoordination.
We can conclude that in these cases \fp\ became stuck in pathological cycling between the symmetric outcomes (where both agents play the same action), which yield zero reward. Such cycling is a well-known problem with \fp; based on claims in the literature, a judicious application of noise to the algorithm would have broken the cycle and improved \fp's performance.

Considering regret on a per-generator basis, \q\ achieved the lowest mean regret on every generator except for D13 (strategically distinct $2 \times 2$ games), on which \rvs\ was the best. \q\ was also the best algorithm to use against almost every opponent. There were only two exceptions: \rvs\ was better against \q\ and \awesome\ was better against itself. Another interesting pairing was when \q\ played against \fp: \q\ attained zero regret in every single game. This indicates that \q\ (uniquely among our algorithms) converged to a pure-strategy best response in every game against \fp.

\subsubsection{Probabilistic Domination of One Algorithm by Another}

When we consider regret distributions on a per-opponent basis, some strong probabilistic dominance trends emerge.

\begin{ob}
On a per-opponent basis, \q, \giga, \gsa\ and \rvs\ were rarely probabilistically dominated in terms of regret.
\end{ob}

\begin{zfigure}[t]
\begin{minipage}[t]{0.45\linewidth}
	\includegraphics[width=1\textwidth]{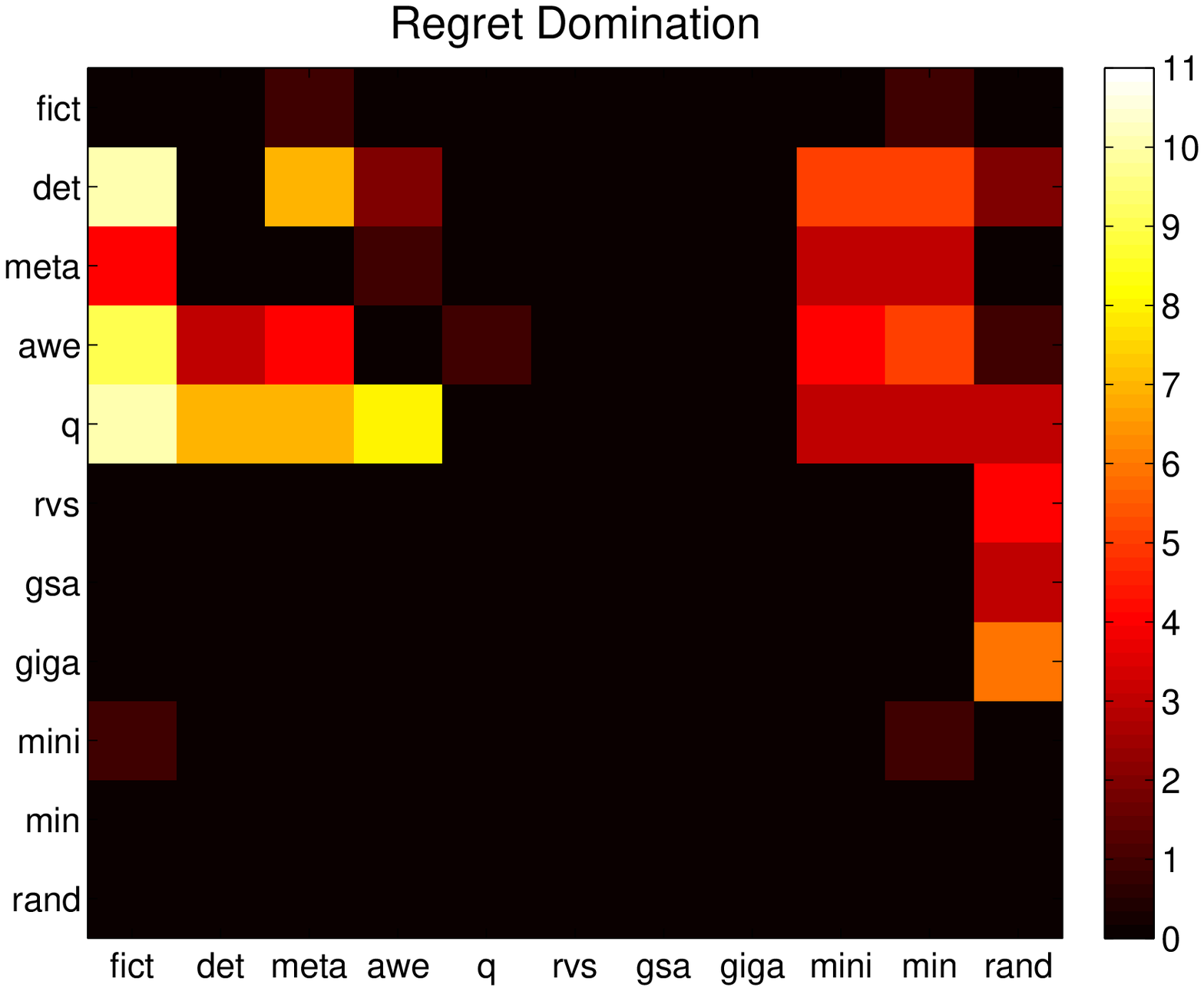}
	\caption{The number of opponents for which the algorithm on the ordinate probabilistically strictly dominates the algorithm on the abscissa. For example, \q\ probabilistically dominates \fp\ on PSMs involving ten out of eleven possible opponents.}
	\label{fig:heatmap_regret_strict_dominance_opponent}
\end{minipage}
\hfill
\begin{minipage}[t]{0.45\linewidth}
	\includegraphics[width=1\textwidth]{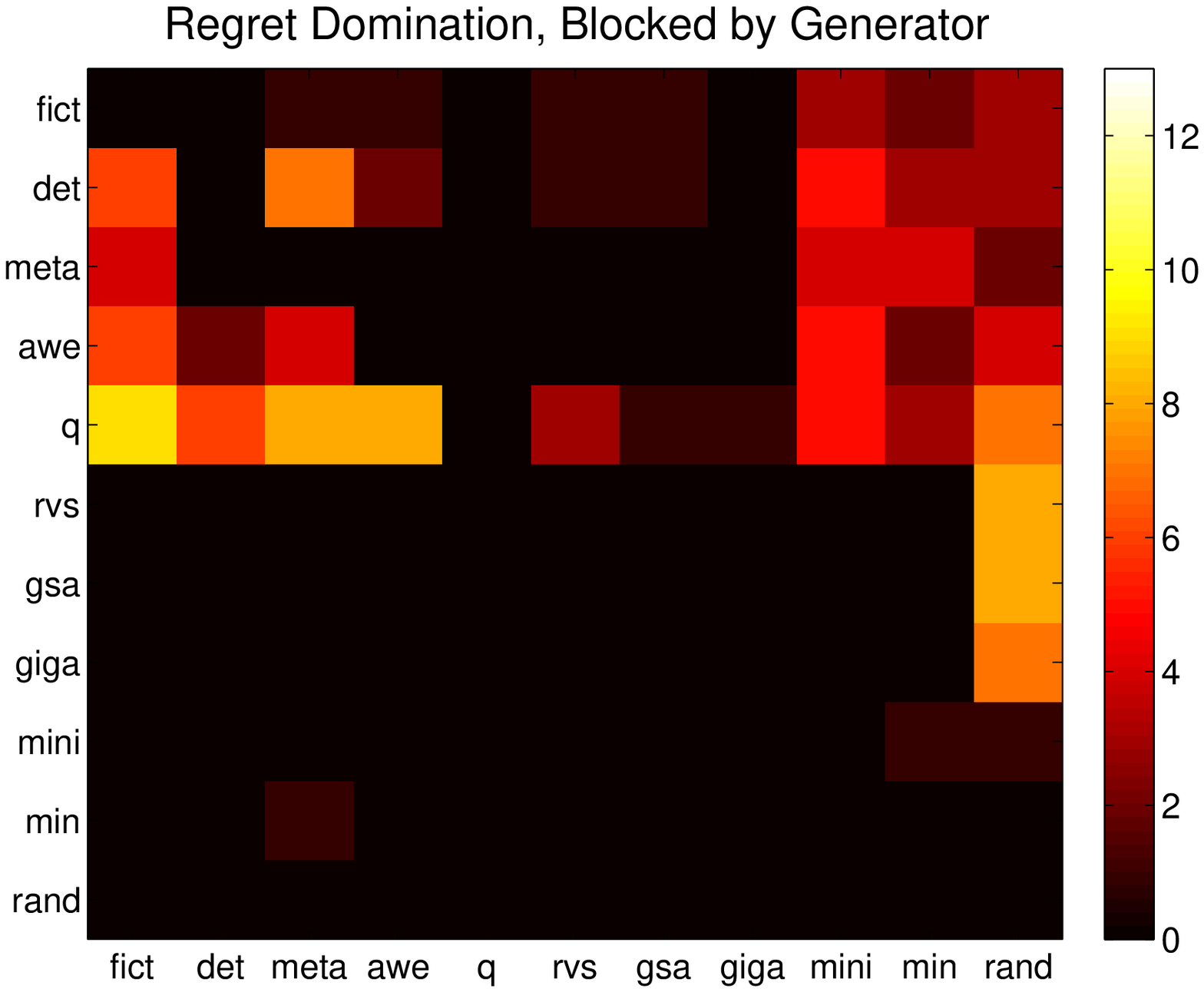}
	\caption{The number of generators for which the algorithm on the ordinate probabilistically strictly dominates the algorithm on the abscissa.}
	\label{fig:heatmap_regret_strict_dominance_generator}
\end{minipage}
\end{zfigure}

First, say that algorithm $A$ dominates $B$ $k$ times if there are $k$ opponents $C$ such that $A$'s regret distribution for matches against $C$ probabilistically dominates $B$'s regret distribution for matches against $C$. Under this notion of domination, we found that the gradient algorithms were never dominated by any other algorithm (Figure~\ref{fig:heatmap_regret_strict_dominance_opponent}). \q\ was only dominated once, by \awesome\ in the case of an \awesome\ opponent. We were not surprised by this, since \awesome\ has special machinery for converging to a stage-game Nash equilibrium in self play. (In a Nash equilibrium, of course, both agents play best responses to each other and hence both accrue zero regret.) On the other hand, \fp\ was frequently dominated, especially by \awesome, \determined, \q\ and to a lesser degree \meta. Both \determined\ and \q\ dominated \fp\ against $10$ opponents (\q\ was the exception for \determined\ and vice versa), and \awesome\ dominated \fp\ on $9$ opponents (\giga\ and \meta\ were the only opponents for which \awesome\ did not dominate \fp).

We can also define probabilistic domination in another way, saying that algorithm $A$ dominates $B$ $k$ times if there are $k$ \emph{generators} $G$ such that $A$'s regret distribution on games from $G$ probabilistically dominates $B$'s regret distribution on games from $G$. Considering domination in this sense, we can draw similar conclusions (Figure~\ref{fig:heatmap_regret_strict_dominance_generator}). \Q\ dominated other algorithms frequently---particularly \fp\ (on $9$ generators), \meta\ ($8$ generators),  and \awesome\ (on $8$ generators)---while avoiding domination by any other algorithm. \FP\ was dominated frequently: by \q\ ($9$ generators), \determined\ ($6$), \awesome\ ($6$) and \meta\ ($4$).

\subsubsection{Links Between Regret and Reward}

What is the connection between regret and reward? We expected that high reward should be correlated with low regret, and vice versa. This intuition was largely supported by our experimental data. Regret and reward were negatively correlated for all algorithms (Spearman's rank correlation test; $\alpha = 0.05$): high reward was linked with low regret. On a per-generator basis, we observed that D10 (\zgame{Traveler's Dilemma}) induced \emph{positive} correlation between regret and reward for all algorithms except \determined\ (Figure~\ref{fig:ag_heatmap_spearman_regret_reward}). This makes sense: in this game, algorithms do better when they do not play best responses, and indeed the unique Nash equilibrium is one of the worst outcomes of the game.

\begin{zfigure}[t]
\begin{minipage}[t]{0.45\linewidth}
			\includegraphics[width=\textwidth]{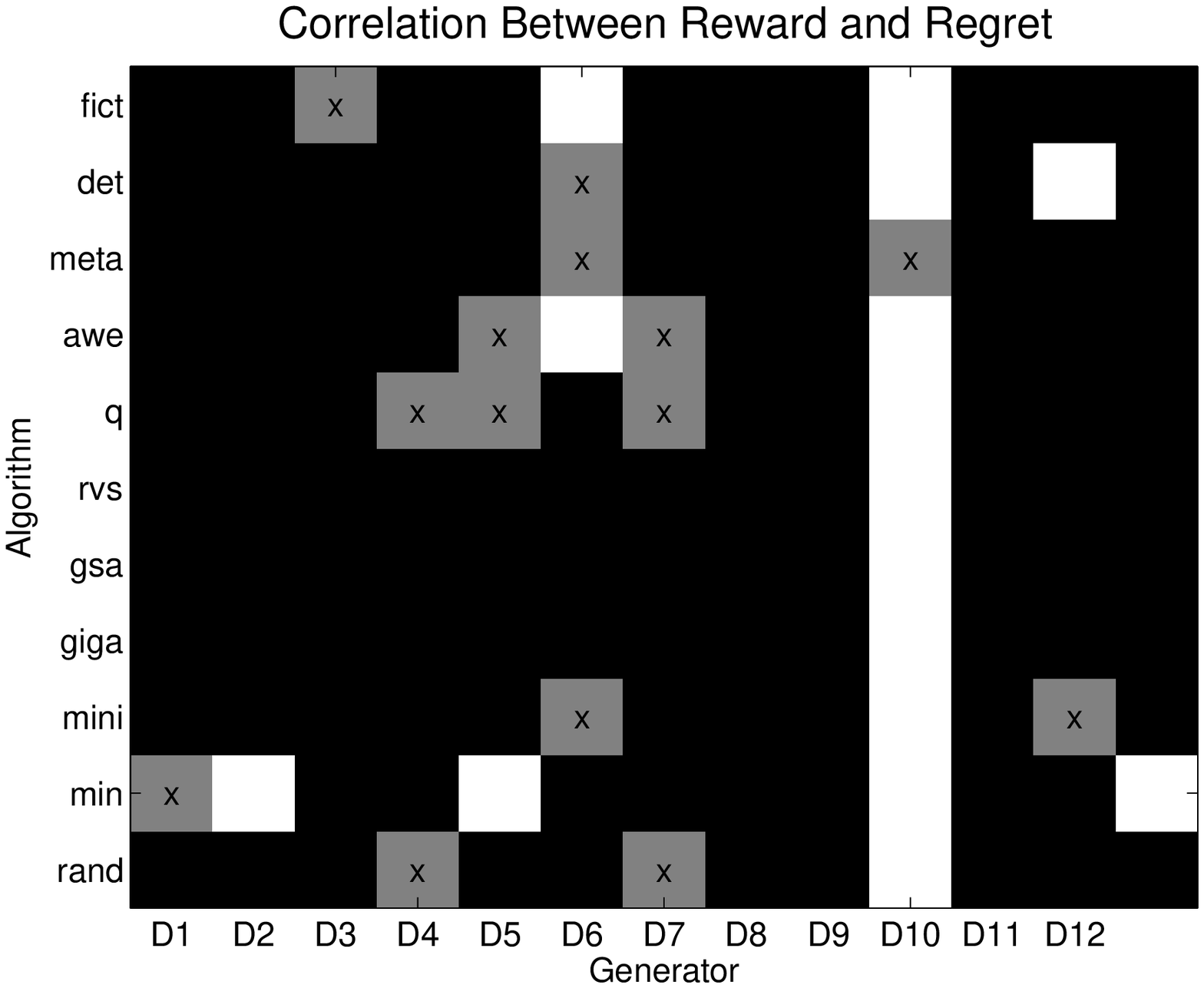}
	\caption[The sign of correlation between reward and regret]{The sign of correlation between reward and regret for each algorithm and each game generator. A white cell indicates positive correlation, a black cell indicates negative correlation, and a grey cell with an `x' indicates insignificant correlation.}
	\label{fig:ag_heatmap_spearman_regret_reward}
\end{minipage}
\hfill
\begin{minipage}[t]{0.45\linewidth}
	\includegraphics[width=\textwidth]{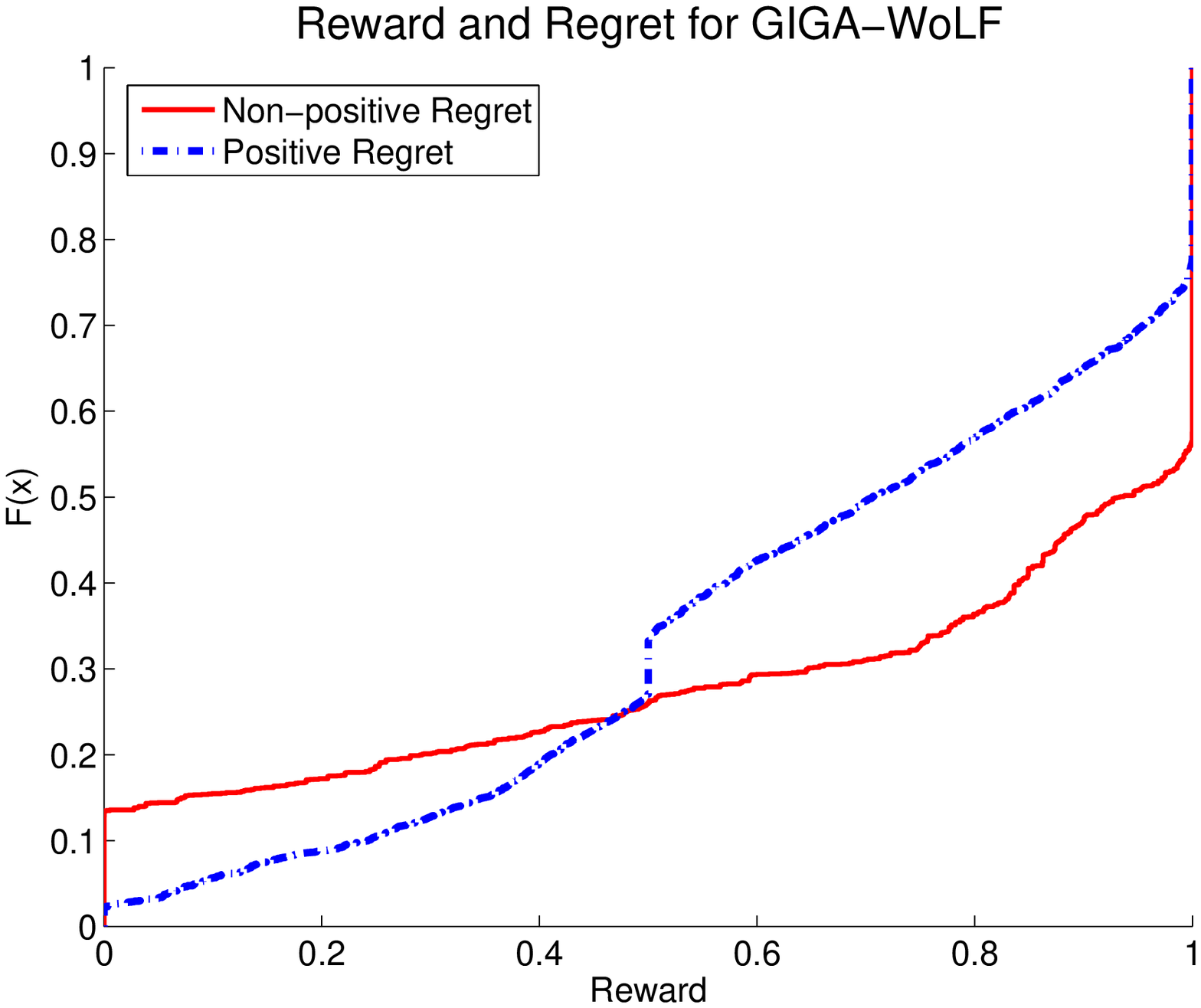}
	\caption[\giga's CDF curves for positive and non-positive reward]{A CDF plot showing \giga's average reward obtained on runs in which it obtained either positive or non-positive reward. Notice that positive-regret runs were less likely to yield zero reward.}
	\label{fig:cdf_reward_and_regret_giga}
\end{minipage}
\end{zfigure}

We compared the average reward each algorithm obtained in positive-regret runs and non-positive-regret runs. For most of the algorithms, the distribution of average reward obtained in non-positive-regret runs probabilistically dominated the distribution of average reward obtained in positive regret runs. There were some exceptions. For example, \q\ exhibited a relatively minor crossover. The same phenomenon occurred with \giga, but in a more pronounced fashion: runs that attained positive regret less often attained zero reward  (Figure~\ref{fig:cdf_reward_and_regret_giga}). Even more dramatically, the positive-regret run distributions probabilistically dominated the non-positive run distributions for \gsa\ and \rvs. These two (gradient) algorithms exhibited behavior different from the other nine: runs with positive regret had better reward characteristics than runs with zero or negative regret. This phenomenon did not seem to arise in the context of a single generator or opponent. However, we did note that the probabilistic domination seemed the weakest when PSMs involving \zgame{Traveler's Dilemma} were omitted.

\subsection{Strategic Stationarity}
\label{sec:stationarity}

All of the metrics we have discussed so far have been based on reward. We now consider several that are based on empirical frequency of action, and that ask whether these frequencies converge. The first---and weakest---notion of convergence that we consider measures whether or not an algorithm converges to a stationary strategy profile. This is interesting in its own right, and is also a necessary condition for stronger forms of convergence.

We consider a run to have been stable if the joint distribution of actions was the same in the first and second halves of the recorded iterations, tested according to the threshold criterion described in \S\ref{sec:threshold} and using $\ell_\infty$-distance. Stability is a property of a run rather than a single algorithm's play in a run, so even algorithms that always play stationary strategies can still participate in unstable runs.

To check how successful our threshold criterion was at detecting stationarity, we began by examining the results for our two algorithms that always play stationary strategies. Our criterion found \determined\ to be stable in $99.5\%$ of self-play matches and \random\ to be stable in $92.0\%$ of self-play matches. When playing each other, they were found to be stable $94.8\%$ of the time. The differences between these cases are likely because \determined\ tends to adopt mixed strategies with smaller supports than \random does, and such a mixed strategy is more likely to yield an empirical action distribution that closely resembles it.\footnote{We note that a false positive rate of between $0.5\%$ and $8\%$ is larger than might be hoped, but nevertheless defer consideration of improved criteria for measuring empirical convergence to future work.}

We found \giga\ and \gsa\ to be the least likely to be stable---particularly in self play, against each other, or against \meta\ (see Figure~\ref{fig:aa_heatmap_stability}). Their striking instability with \meta\ was potentially because they tripped \meta's internal stability test and changed its behavior. However, \awesome\ also has a similar internal check, but the stability of \giga\ and \gsa\ were not noticeably different between matches with \awesome\ and with \q\ (which has no such check). \rvs, the other gradient algorithm, was more stable than \giga\ and \gsa. This might be because \rvs\ had a more aggressive step length: the parameters used in this experiment for \giga\ and \gsa\ were taken from \inlinecite{bowling04}, who indicated that these parameters were intended to produce smooth trajectories rather than fast convergence.

\begin{zfigure}[t]
\begin{minipage}[t]{0.45\linewidth}
	\includegraphics[width=1\textwidth]{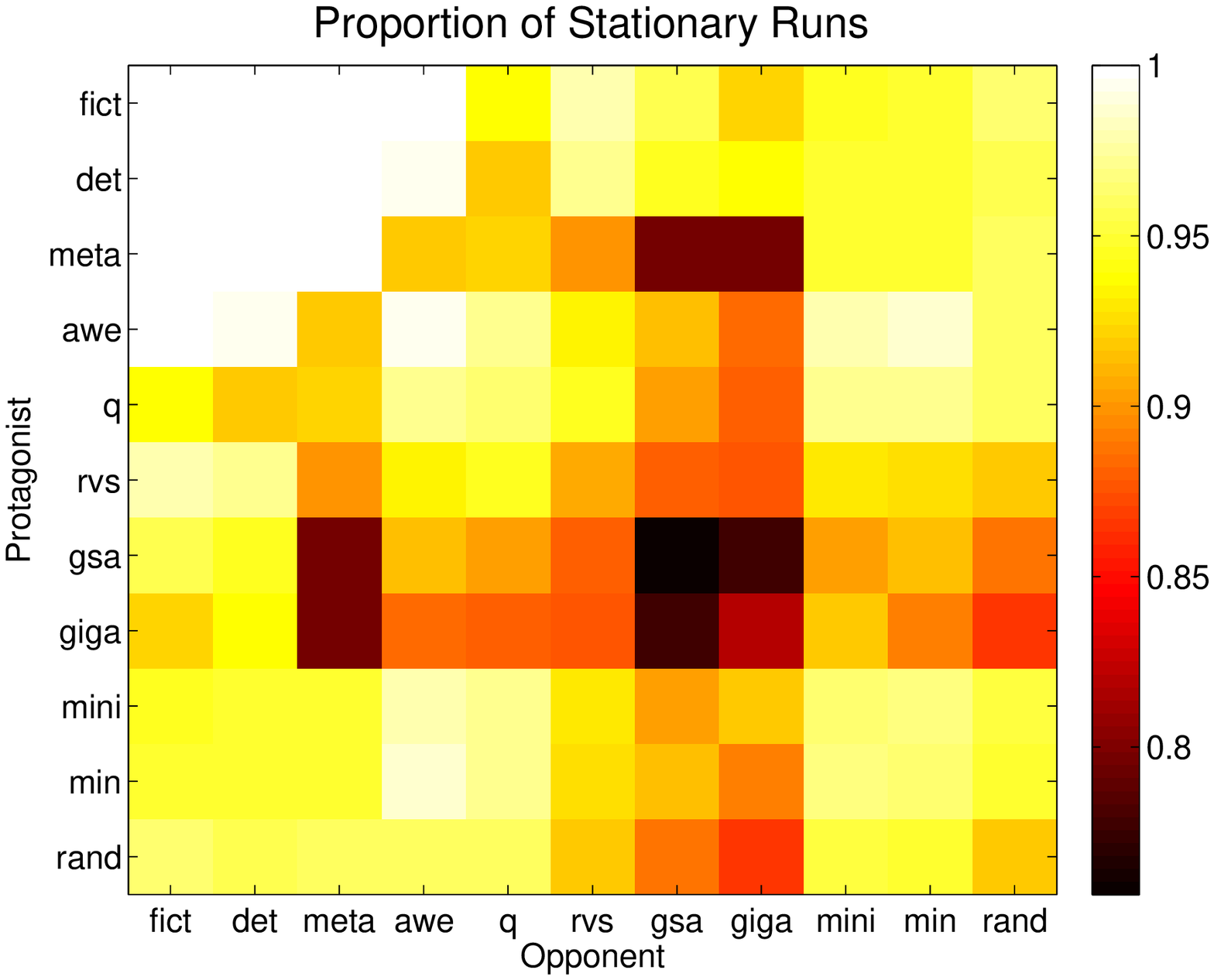}
	\caption[Proportion of stationary runs, blocked on opponent]{Proportion of stationary runs, blocked on opponent. This intensity map is symmetric; we removed redundant entries for clarity.}
\note{I think it would be better to fill in the whole square; symmetry makes it easier to read.}
	\label{fig:aa_heatmap_stability}
\end{minipage}
\hfill
\begin{minipage}[t]{0.45\linewidth}
		\includegraphics[width=1\textwidth]{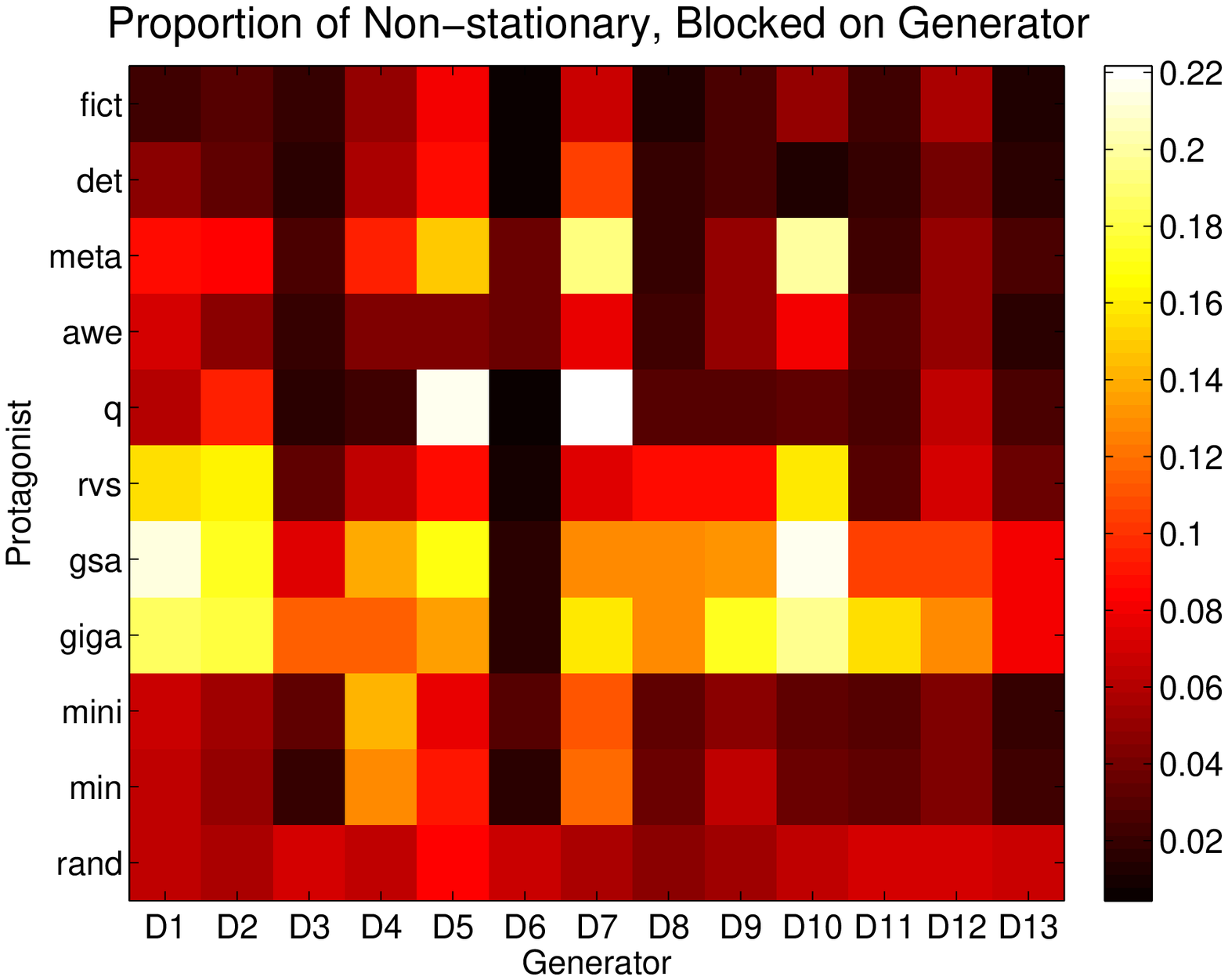}
	\caption{Proportion of non-stationary runs, blocked on generator and protagonist}
		\label{fig:ag_heatmap_stability}
\end{minipage}
\end{zfigure}

\Meta, \determined, \fp\ and \awesome\ were, for the most part, quite good at achieving stationarity. \Meta\ and \fp\ were particularly strong against each other, and always reached a stationary strategy profile. The only exception to the rule of stability in this group was \awesome\ vs. \meta; this pairing was unstable in $10.3\%$ of runs. We are not sure why this occurred, but conjecture that it arose because of the discrete behavioral changes that both algorithms undergo when their internal states are updated.

There were a number of problem generators for the different algorithms (see Figure~\ref{fig:ag_heatmap_stability}). For example: generators D1, D2, and D10 created instances that were particularly difficult for the gradient algorithm in terms of strategic stability; \Q\ was weak on both D5 and D7; and \meta\ tended to be unstable on D5, D7 and D10. However these unstable instances were rare regardless of the algorithm pairing. The vast majority of runs found a stationary strategy profile. Even \giga, which was the algorithm least likely to stabilize, found stationarity in $87.0\%$ of its runs (see Figure~\ref{fig:bar_stability}).

\subsection{Convergence to Stage-Game Nash Equilibrium}
\label{sec:ne_convergence}

Stable runs are those that converged to any strategy; we now consider which of these selected a (possibly mixed-strategy) stage-game Nash equilibrium. For some algorithms, Nash equilibrium convergence was reasonably common. \Awesome\ converged in $54.3\%$ of its runs, and \determined\ converged in $53.1\%$ of its runs. \Determined\ was better than \awesome\ at converging to a Pareto-optimal Nash equilibrium (a Nash equilibrium that was not Pareto-dominated by any other Nash equilibrium). \Awesome\ most frequently converged to a Pareto-dominated equilibrium. This was likely influenced by the way that our implementation of \awesome\ picked its `special' equilibrium:\footnote{The original paper, \inlinecite{conitzer07}, left the method of picking the `special' equilibrium unspecified.} the first equilibrium found by the Lemke-Howson algorithm, without attention to whether it was, e.g., Pareto-dominated. \awesome\ also tended to attain lower reward when it converged to a Pareto-dominated Nash equilibrium than when it did not converge or converged to a non-dominated Nash equilibrium.

Figure~\ref{fig:bar_sp_stability} shows the extent to which each algorithm converged to a stage-game Nash equilibrium in self play. Notice how often determined converged: this indicates that the games we studied often possessed one Nash equilibrium that was the best for both agents. Indeed, we can see that a surprisingly high number of games had a \emph{unique} stage-game Nash equilibrium ($58.5\%$). We expect that convergence results would look qualitatively different with generators that were much less likely to produce games with unique equilibria.

\begin{zfigure}[t]
\begin{minipage}[t]{0.45\linewidth}
					\includegraphics[width=1\textwidth]{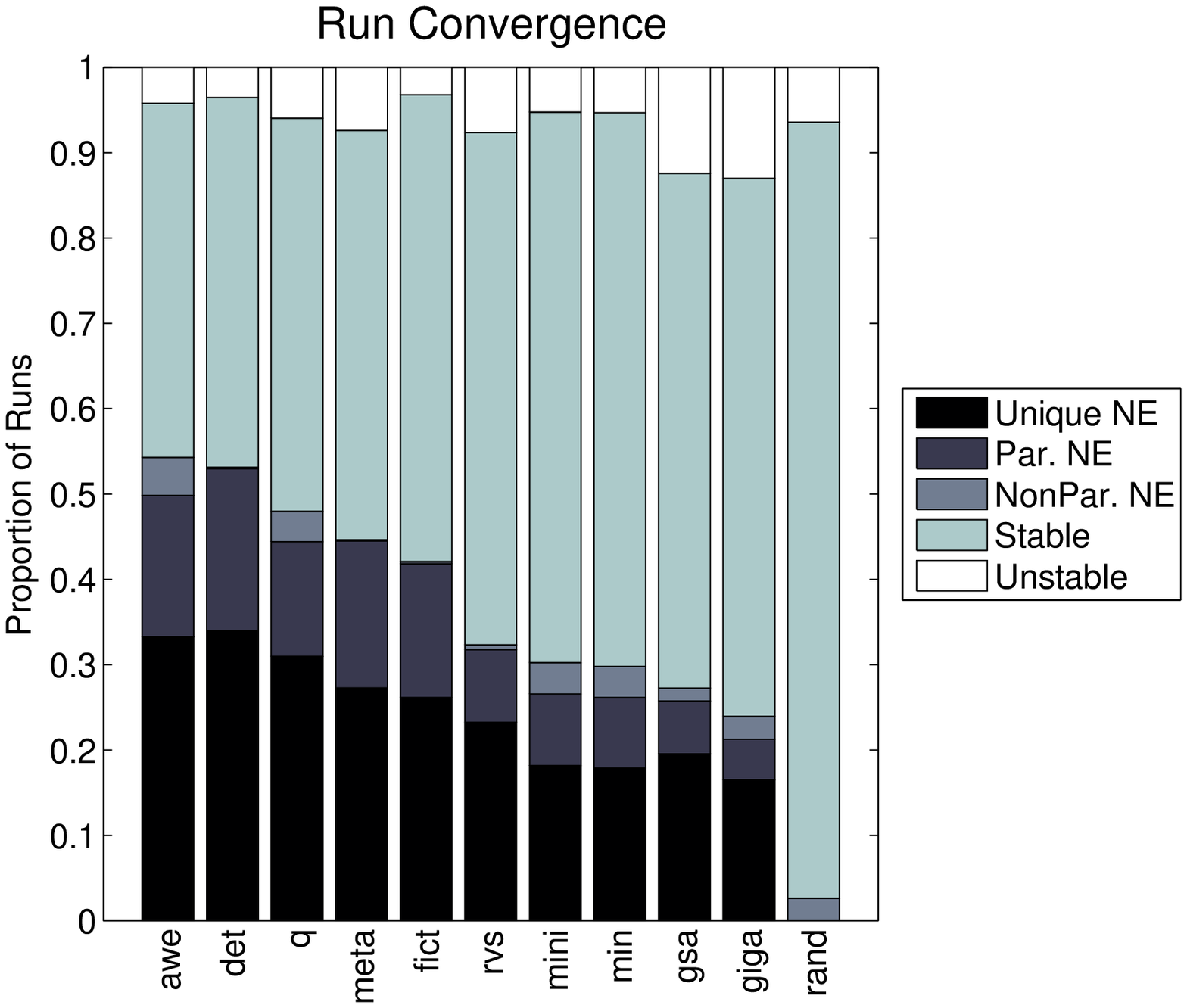}
	\caption[Convergence proportions]{The proportion of runs that were stationary, converged to a non-Pareto-optimal Nash equilibrium, or converged to a Pareto-optimal Nash equilibrium.}
\note{The x axis labels run together in this figure, and in the next.}
		\label{fig:bar_stability}
\end{minipage}
\hfill
\begin{minipage}[t]{0.45\linewidth}
			\includegraphics[width=1\textwidth]{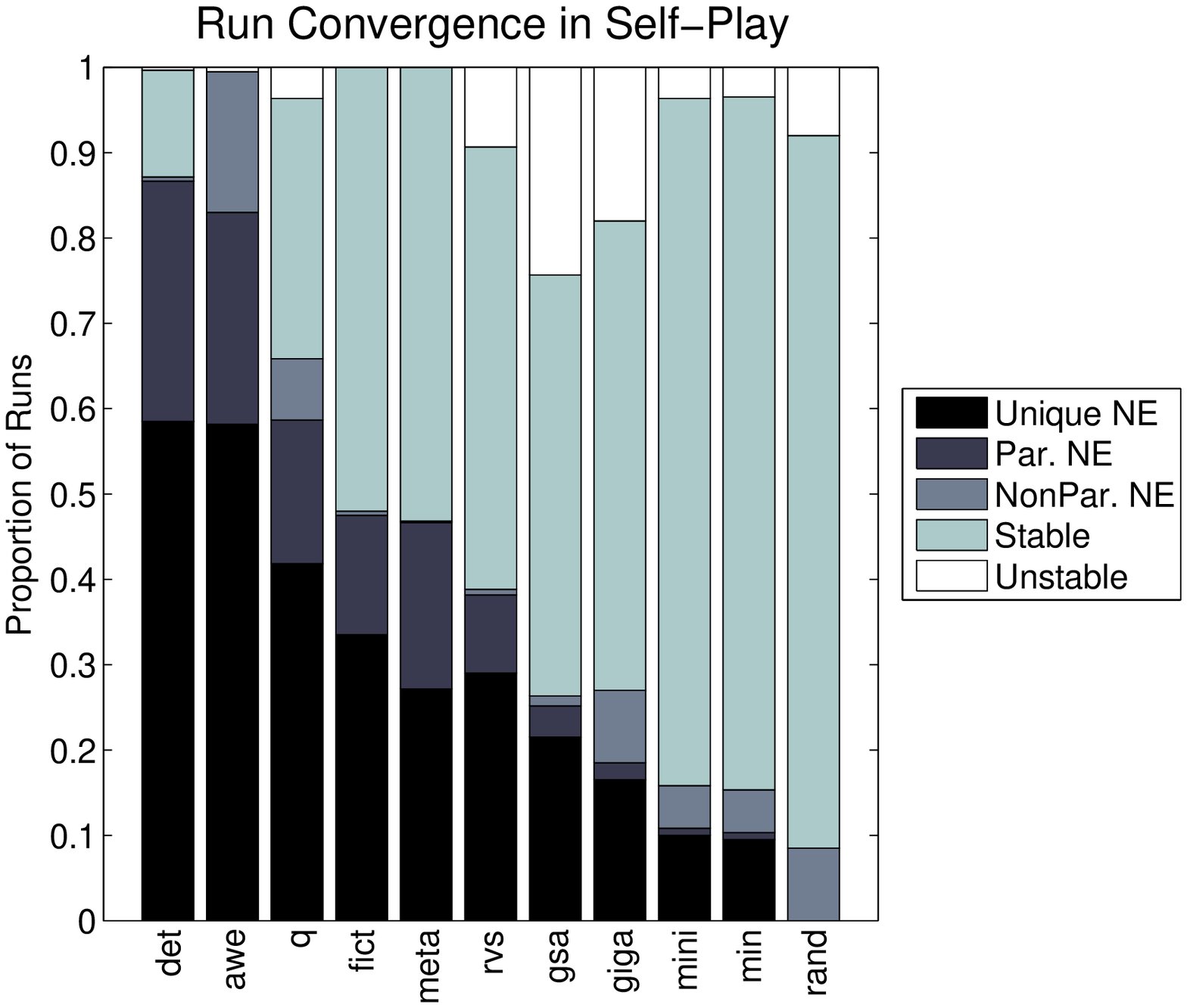}
	\caption[Self-play convergence proportions]{The proportion of self-play runs that were stationary, converged to a non-Pareto-optimal Nash equilibrium, or converged to a Pareto-optimal NE.}
\note{How can determined \emph{ever} converge to a non-Pareto-optimal NE in self play?}
\note{EPZ: how much detail do we want to spend on this? Can we just say ``the presence of non-Pareto-optimal NE in self play is caused by games where miscoordination just so happens to be close to a mixed strategy NE.''?}
		\label{fig:bar_sp_stability}
\end{minipage}
\end{zfigure}

Observe that \awesome\ nearly always converged. Recall that we previously found that \awesome\ received lower average reward in self-play than non-self-play runs (\S~\ref{sec:selfplay}). Now we can conclude that this failure to achieve high rewards was not due to a failure to reach equilibrium. An interesting modification of the \awesome\ algorithm would be to use its self-play machinery to converge to stable strategies that are not stage-game Nash equilibria, such as the socially-optimal outcome or the Stackelberg game equilibrium. The aim of this adjustment would be to improve self-play reward results while maintaining \awesome's resistance to exploitation by other algorithms.

\subsubsection{Links Between Nash Equilibrium Convergence and Reward}

Much work in the literature has aimed at MAL algorithms that converge to a stage-game Nash equilibrium. However, if the goal is high average reward, is such convergence desirable? More generally, is proximity to stage-game Nash equilibrium correlated with obtaining high reward?

\begin{ob}
Strategic proximity to stage-game Nash equilibrium was correlated with average reward for all algorithms and most algorithm--generator pairs.
\end{ob}

For all algorithms, reward was negatively correlated with $\ell_\infty$-distance to the closest Nash equilibrium (Spearman's rank correlation test; $\alpha = 0.05$). Furthermore, most algorithms were negatively correlated even on a per-generator basis (Figure~\ref{fig:ag_heatmap_spearman_ne_reward}). The most notable exceptions were D6, D12, and (especially) D10, where we saw \emph{positive} correlations between distance and reward.

\begin{zfigure}[t]
\begin{minipage}[t]{0.45\linewidth}
				\includegraphics[width=\textwidth]{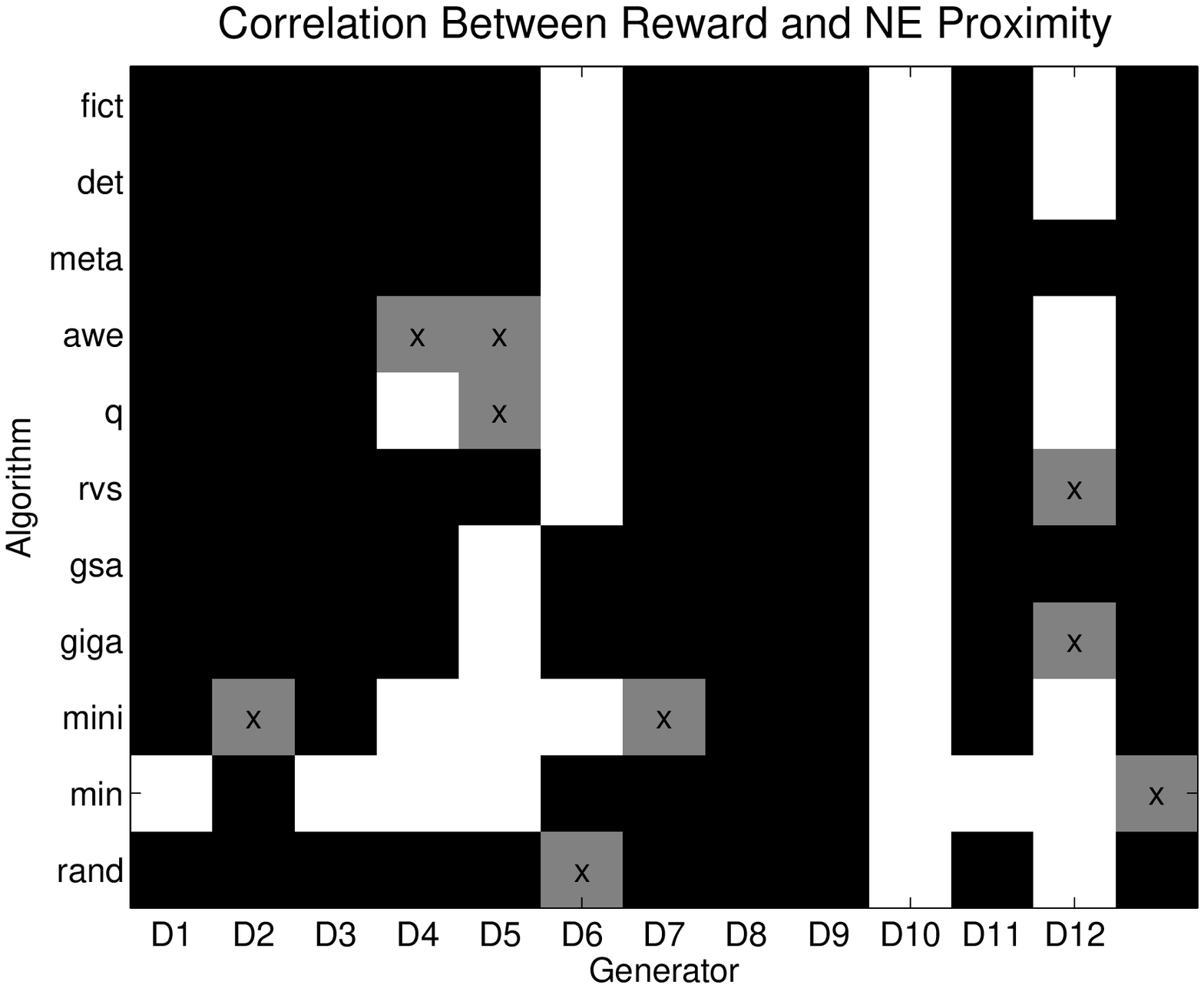}
	\caption[The sign of correlation between reward and $\ell_\infty$-distance to the closest Nash equilibrium]{The sign of correlation between reward and $\ell_\infty$-distance to the closest Nash equilibrium for each algorithm and each game generator. A white cell indicates positive correlation, a black cell indicates negative correlation, and a grey cell with an `x' indicates insignificant correlation.}
	\label{fig:ag_heatmap_spearman_ne_reward}
\end{minipage}
\hfill
\begin{minipage}[t]{0.45\linewidth}
			\includegraphics[width=1\textwidth]{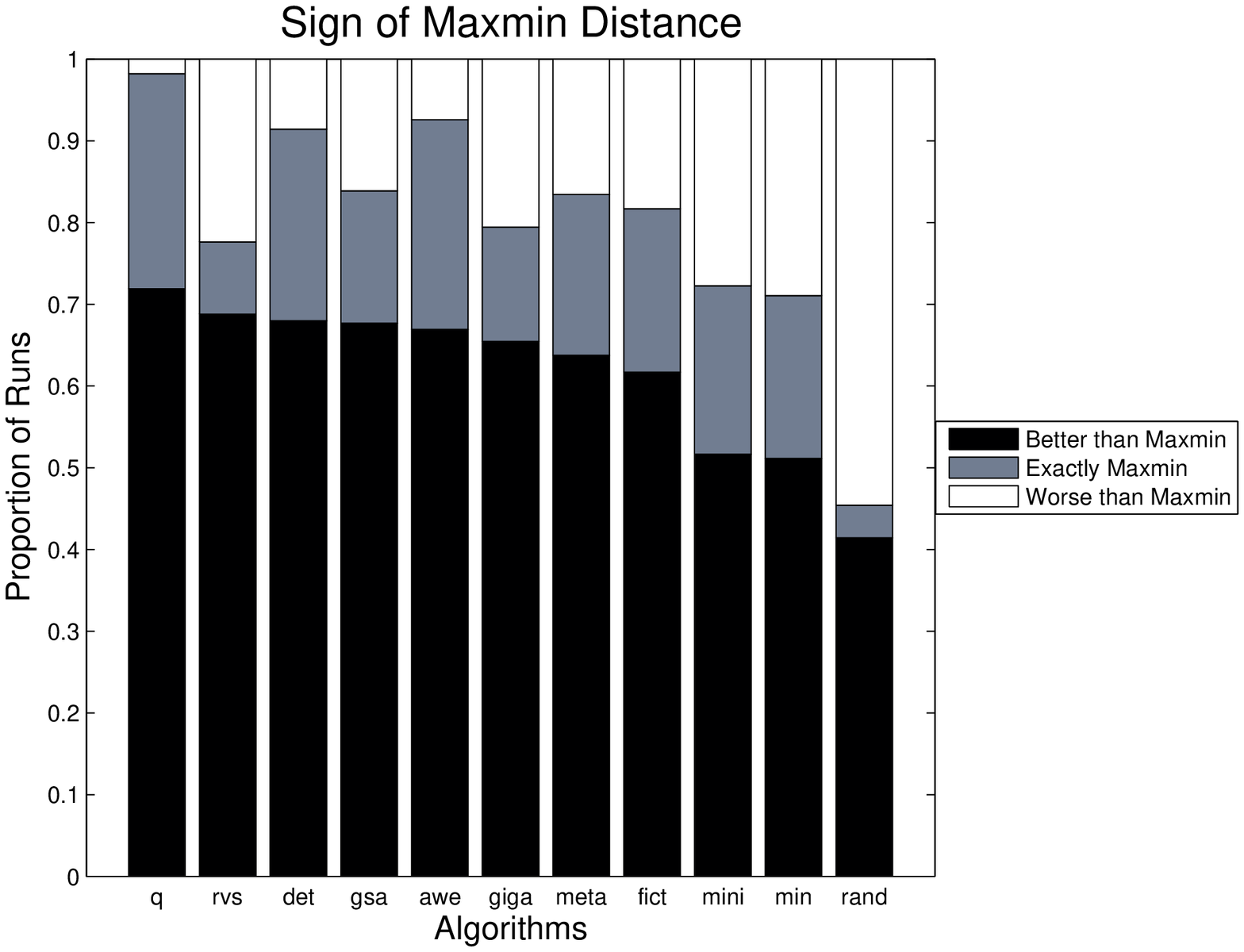}
	\caption{The sign of the maxmin distance of each run, by algorithm.}
		\label{fig:bar_minimax_frequency_grand}
\end{minipage}
\end{zfigure}

\subsection{Maxmin Distance}
\label{sec:maxmin}
An agent's maxmin value is the largest amount that it can guarantee itself regardless of its opponent's behavior. Thus, achieving average reward of at least this amount is widely seen as a necessary condition for sensible MAL behavior. Furthermore, the famous Folk Theorem of game theory demonstrates that enforceable payoffs (those with non-negative maxmin distances) are precisely those payoffs that can be achieved in equilibrium of an infinitely repeated game. We build on our results here to investigate this notion of convergence in \S\ref{sec:rne_convergence}.
In this section we consider the difference between average reward and the max\-min value of the underlying game instance:
\begin{equation}
MaxminDistance(\vec{r}_i) = \frac{\sum_{t=1}^{T}r_i^{(t)}}{T} - \max_{\sigma_i \in \prod{A_i}}\min_{\sigma_{-i} \in \prod{A_{-i}}}u(a_i,a_{-i}).
\label{eqn:maxmin_distance}
\end{equation} 
\note{You write this in terms of pure strategies. However, didn't you calculate it in terms of LPs? This would mean you should use strategies instead of actions, and maybe make a remark to this effect.}
We call this difference \textit{maxmin distance}, noting that it can be negative.

\begin{ob}
\Q\ attained an enforceable payoff more frequently than any other algorithm.
\end{ob}

\begin{zfigure}[t]
\begin{minipage}[t]{0.45\linewidth}
		\includegraphics[width=1\textwidth]{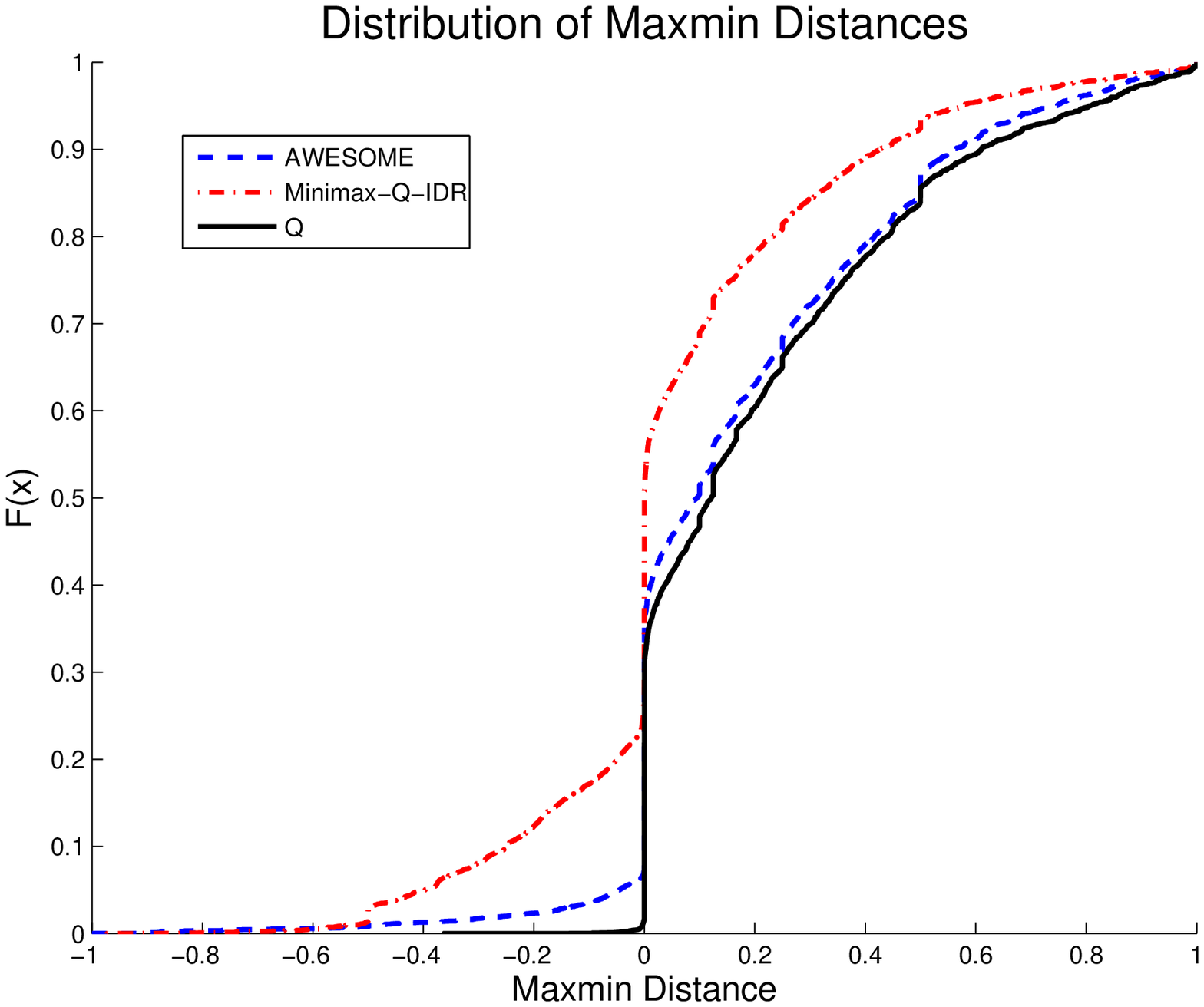}
	\caption{The distribution of maxmin distances for \awesome, \minimax\ and \q.}
		\label{fig:ecdf_minimax_grand}
\end{minipage}
\hfill
\begin{minipage}[t]{0.45\linewidth}
			\includegraphics[width=1\textwidth]{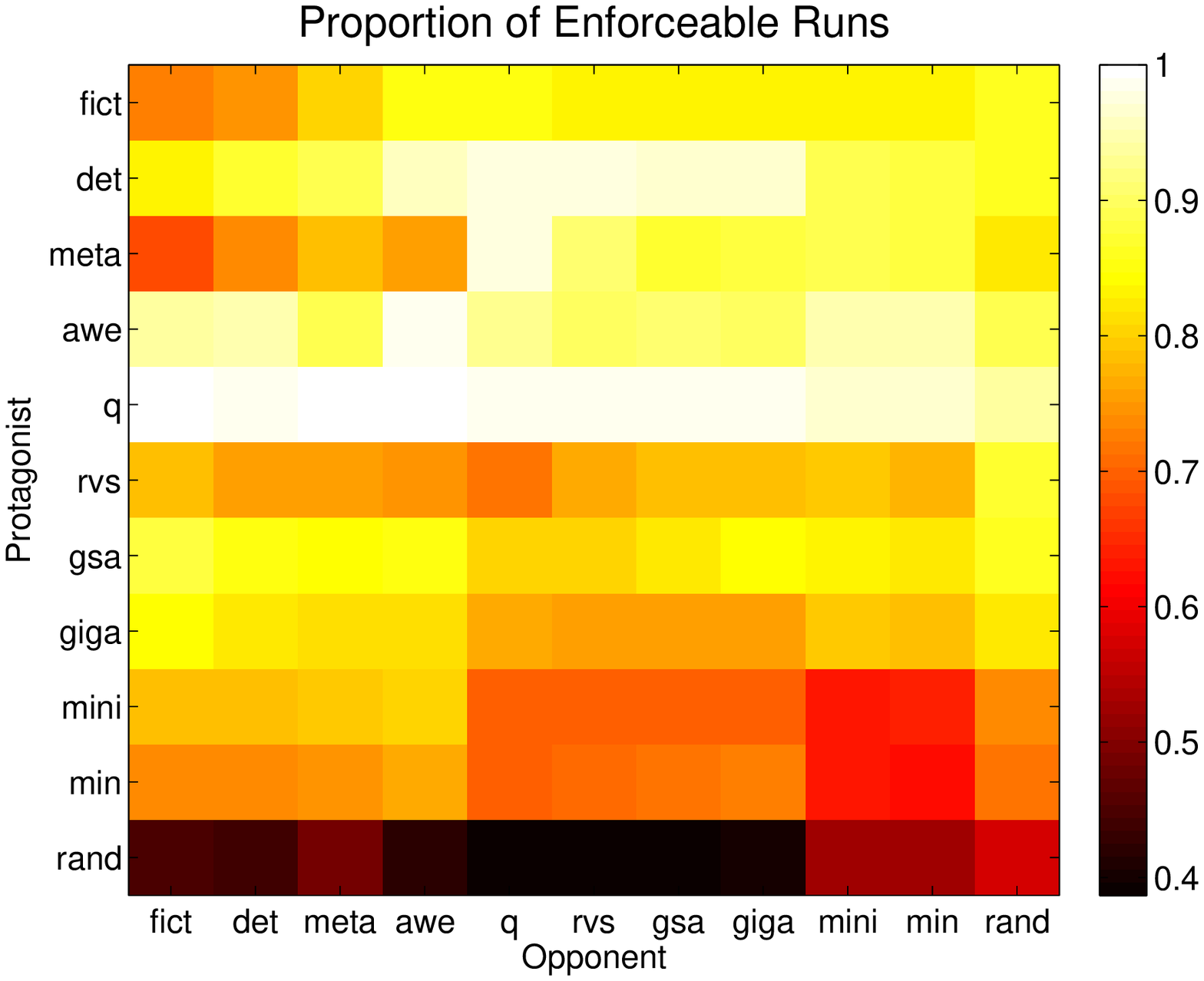}
	\caption{The proportion of enforceable runs, blocked by opponent.}
		\label{fig:aa_heatmap_minimax_enforce}
\end{minipage}
\end{zfigure}

\note{Asher: I don't think that figures 27 or 28, fig:ecdf minimax grand and fig:aa heatmap minimax enforce are referenced in the paper anywhere.}

\Q\ most frequently attained an enforceable payoff, with a negative maxmin distance in only $1.8\%$ of its runs (Figure~\ref{fig:bar_minimax_frequency_grand} and Figure~\ref{fig:ecdf_minimax_grand}). As can be seen in Figure~\ref{fig:aa_heatmap_minimax_enforce}, the runs on which \q\ failed to attain an enforceable payoff mostly came from either D4 (\zgame{Dispersion Game}; $37.6\%$ of \q's unenforceable runs) or D13 (\zgame{Two by Two Game}; $33.3\%$). They also occurred predominantly against \random\ ($29\%$ of the unenforceable runs), \minimax\ ($17.3\%$) and \minimaxidr\ ($16.0\%$). The next-best algorithm, \awesome, attained enforceable payoffs considerably less often, with a negative maxmin distance in $7.4\%$ of its runs.

After \random, \Minimax\ and \minimaxidr\ were the \emph{least} likely to attain enforceable payoffs, failing to do so in $28.9\%$ and $27.7\%$ of their runs respectively. This is interesting because these algorithms explicitly attempt to do well against adversarial opponents. One possible explanation is that they may have trouble learning accurate payoffs \note{we don't tell them the true payoffs? I thought we said in section 3 that we did for all algorithms...}, leading them to have difficulty obtaining their maxmin values. \note{I don't find this explanation very convincing. Other possibilities: (a) maybe they play ``pure-strategy maxmin'', but we're measuring mixed-strategy maxmin, leading them to be below; (b) maybe they really do play mixed strategy maxmin, but due to sampling error their average reward is a hair below their maxmin value.}

\Minimax\ and \minimaxidr\ were especially poor in self play, where conservative play can impair payoff learning. There is also a greater proportion of enforceable runs on $2 \times 2$ games ($75.2\%$) than on $10 \times 10$ games ($68.5\%$)---larger games have more payoffs to learn. Working on a more sophisticated exploration scheme looks like an especially promising place to improve our implementation of \minimax\ and its variant.

While \q\ was successful against a broad range of opponents, some other algorithms were less consistent. For example, \meta\ was quite good against all opponents except for \fp, \determined, \awesome\ and itself. \Meta\ was especially bad against \fp; in this pairing only $68.0\%$ of \meta's runs were enforceable. Compare this to \meta's excellent performance against \q, where it attained enforceable payoffs in $97.7\%$ of it runs. \FP\ also had trouble playing against \meta, \determined\ and itself. On the other hand, neither \awesome\ nor \determined\ shared this problem.

\begin{zfigure}[t]
\begin{minipage}[t]{0.45\linewidth}
	\includegraphics[width=1\textwidth]{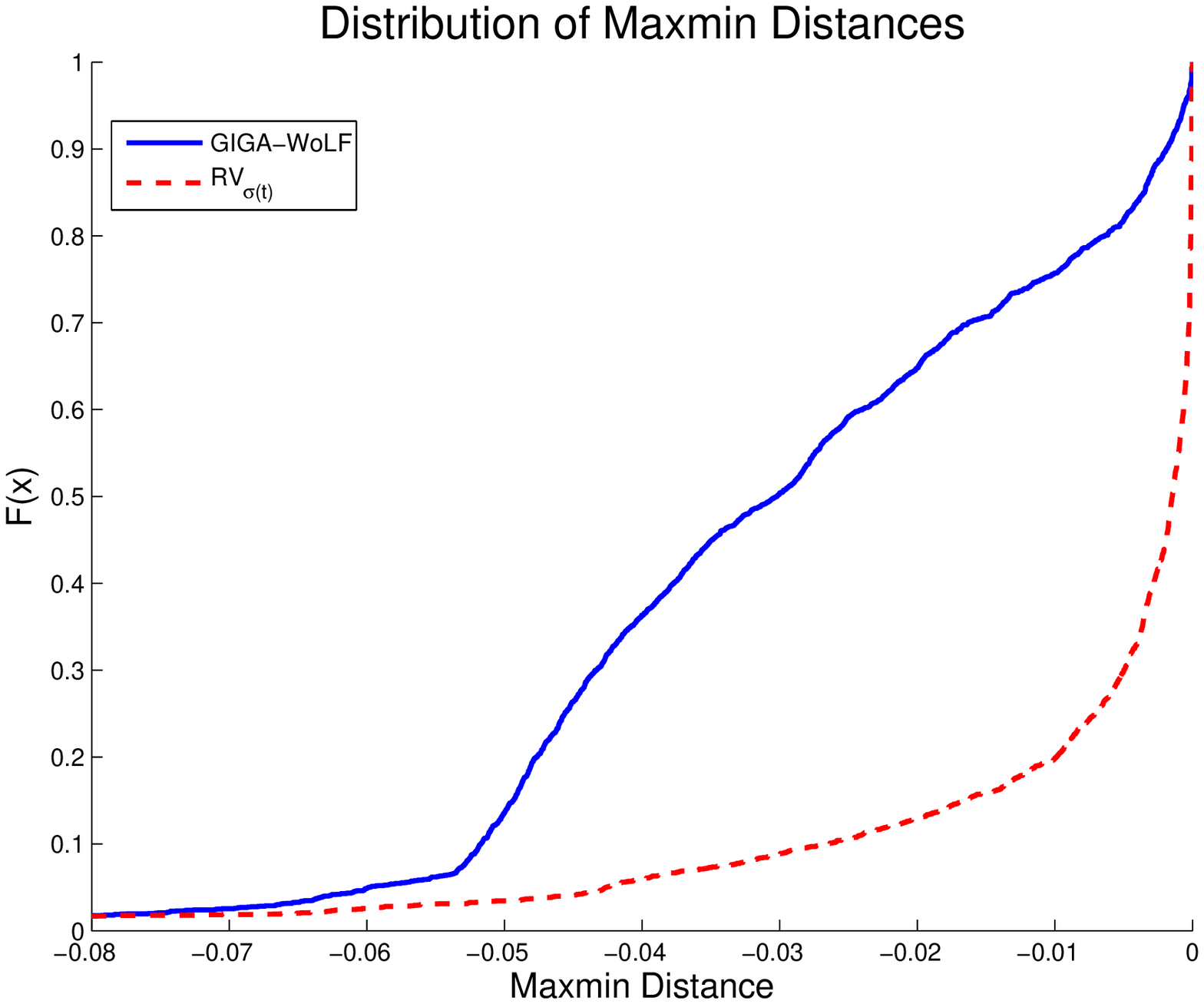}
	\caption{The distribution of negative maxmin distances for \giga\ and \rvs.}
		\label{fig:ecdf_minimax_rvs_normalize}
\end{minipage}
\hfill
%
\begin{minipage}[t]{0.45\linewidth}
					\includegraphics[width=1\textwidth]{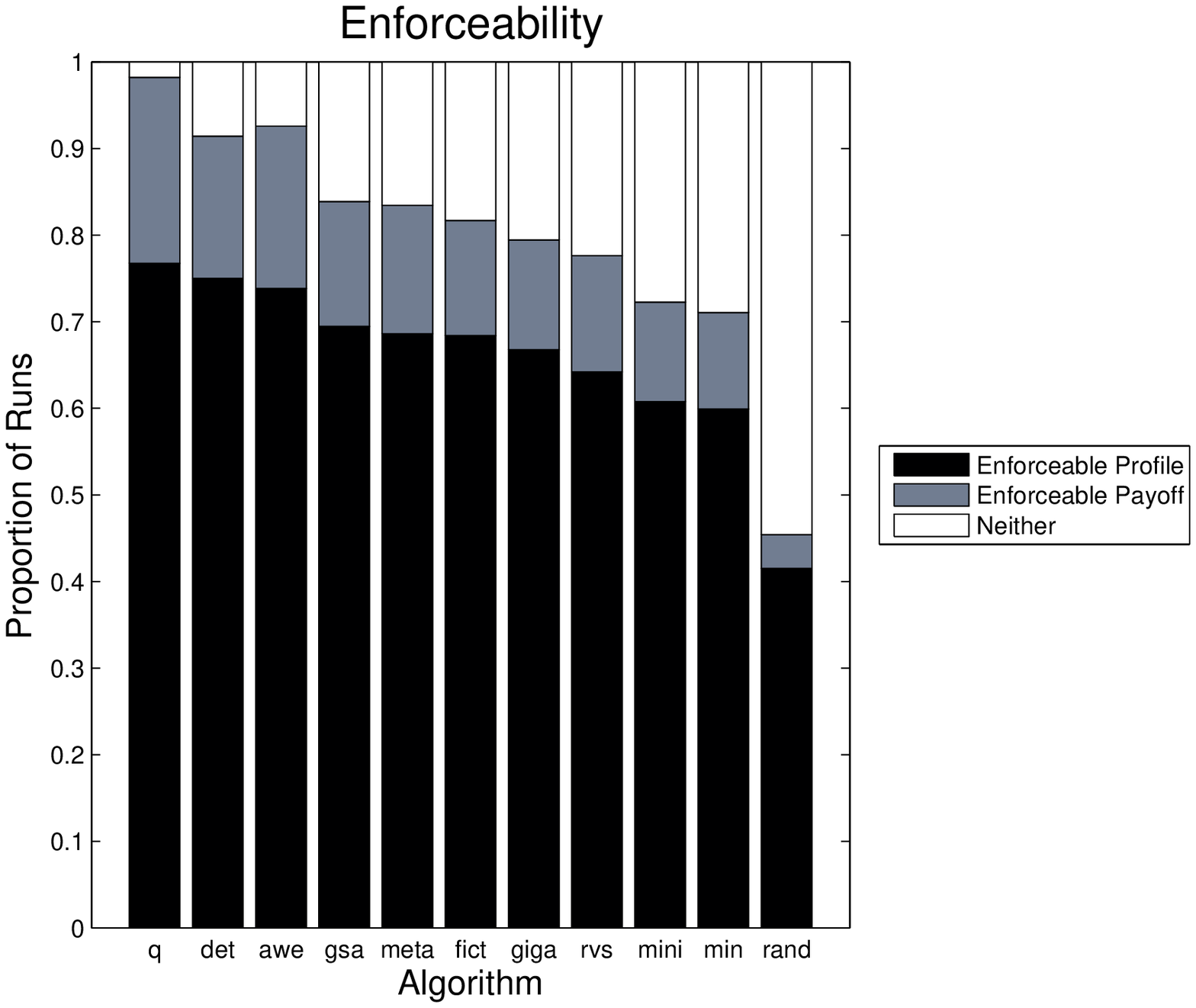}
	\caption{Proportion of PSMs with enforceable payoffs and payoffs profiles achieved, by algorithm.}
		\label{fig:bar_repeated_game}
\end{minipage}
\end{zfigure}

\rvs\ had problems attaining enforceable runs too, and although it received payoffs well above the maxmin value frequently (it had the second highest proportion of runs with strictly positive distances at $68.8\%$) there were also a large number of instances where \rvs's maxmin distance was close to, but below zero. This contrasts with \giga, which had fewer non-enforceable runs with greater negative minimax distance (see Figure~\ref{fig:ecdf_minimax_rvs_normalize}). We conjecture that this phenomenon occurred because \rvs\ maintains a small amount of probability mass on all of its actions, causing it to `tremble'. More specifically, \rvs, like all gradient algorithms, updates its mixed strategy by moving along the reward gradient. When the updated vector does not sum to one, it must be mapped back to the probability simplex. \rvs\ does this by normalizing the updated vector, while \gsa\ and \giga\ use a retraction operator that tends to drop actions from the mixed strategy's support (see \S~\ref{gradient-algs}). We conjecture that modifying \rvs\ to use \giga's retraction operator would improve \rvs's ability to achieve enforceable payoffs.

\subsubsection{Links Between Maxmin Distance and Reward}

Is there a connection between enforceable runs and high average rewards? It would be strange if some such relationship did not exist, since enforceability implies reward higher than the maxmin value. Indeed, we did observe that for all algorithms, maxmin distance was positively correlated with average reward. (Spearman's rank correlation test (\S\ref{sec:spearman}); $\alpha = 0.05$ significance level). On a per-generator basis, we again largely observed significant positive correlations.  
There were two deviations from this pattern. First, we found no significant correlation for half of the algorithms on D11, and for \minimax\ on D3. Second, there was a significant \emph{negative} correlation for \minimax\ on D11, though \minimaxidr\ still exhibited significant positive correlation.

\subsection{Convergence to Repeated-Game Nash Equilibrium}
\label{sec:rne_convergence}

In \S\ref{sec:ne_convergence} we considered algorithms' tendencies to converge to equilibria of the stage game. The algorithms actually played a repeated game, however. We now turn to analyzing this repeated game's properties. The payoff profiles achievable in Nash equilibrium of a repeated game are precisely the enforceable profiles (see, e.g., \inlinecite{OsbRub}). In order to determine whether a given strategy profile is an equilibrium of a repeated game, it is also necessary to consider how these strategies behave off the equilibrium path (e.g., how they punish deviations by the other agent). While the algorithms that we studied lack punishment mechanisms, it is still meaningful to assess how frequently they converged to payoff profiles consistent with repeated game Nash equilibria. We therefore build on the results from \S~\ref{sec:maxmin}, asking how often \emph{both} algorithms achieved enforceable payoffs.

\begin{ob}
\Q\ was involved in matches whose payoff profiles were consistent with a repeated game Nash equilibrium more often than any other algorithm.
\end{ob}

Of the algorithms that we examined, \q\ most frequently had runs that were consistent with a repeated game Nash equilibrium (Figure~\ref{fig:bar_repeated_game}). It was consistent with a repeated game equilibrium in $76.8\%$ of its runs. \Determined\ and \awesome\ were the next most frequently consistent ($75.0\%$ and $73.8\%$ of their runs respectively). Overall, consistency with a repeated game Nash equilibrium was common, but not universal. 
It is worth emphasizing that an enforceable payoff profile depends on both agents' actions, and so the behavior of weak agents like \random\ lowered the scores for stronger opponents.

\section{Discussion and Conclusion}
In this article we described MALT, a standardized testbed for multiagent experimentation. This testbed allows researchers to focus on experimental design and analysis instead of implementation. We also presented an in-depth analysis of a large experiment we conducted ourselves using MALT.

 
The most striking conclusion from our experiment was that \q\ achieved consistently excellent results, in many senses outperforming algorithms based on deeper insights about the multiagent setting (e.g., \giga, \awesome, and \meta). We were surprised by this finding, since we had taken for granted the idea that modern, multiagent algorithms would do better in a repeated-game environment than a classical, single-agent algorithm. The evidence we have shown to the contrary suggests that it should be possible to considerably improve the empirical performance of MAL algorithms. We suggest four areas in which efforts could be worthwhile.

First, a more experimentally-driven focus seems crucial. Our experiment was large, but there are many empirical questions that it does not answer. Some promising future directions include:
\begin{packed_list}
	\item More examination of the relationship between performance and game properties like size;
	\item More detailed investigation of algorithm be\-hav\-ior on instances from single generators;
    \item Investigation of additional algorithms like Hyper-Q \cite{tesauro04} and Nash-Q \cite{hu98};
	\item Study of $N$-player repeated games and stochastic games (along the lines of \inlinecite{vu06}).
\end{packed_list}

Second, the more sophisticated algorithms have many tunable parameters. Finding optimal settings for them was beyond the scope of our paper, and we instead relied on published parameter settings. Nevertheless, it is possible that some algorithms would have performed considerably better if they had been configured differently. Indeed, \q\ had only three parameters and all were easy to set, which might partly explain its strong performance. Tuning the other algorithms would require considerable experimental effort; hopefully MALT will be of assistance. There are some interesting questions to ask:
\begin{packed_list}
	\item Is one parameter setting good for many problems, or is it the case that some parameter settings are effective on some matches and poor on others?
	\item Which of an algorithm's (e.g., \meta's) parameters are the most important?
	\item Does \awesome's performance change radically when it selects the socially optimal Nash equilibrium as its special equilibrium? How about the `Stackelberg' equilibrium?
	\item For gradient algorithms, is it better to perform retraction or normalization?
	\item Do parameter settings that yield high reward also yield low regret?
\end{packed_list}

Third, we presented two different tweaks to existing algorithms: \minimaxidr\ and \gsa. These algorithms offered several improvements over their ``parent'' algorithms, and in many cases probabilistically dominated them. It would be interesting to explore similar modifications of other existing algorithms. 

Finally, managing a portfolio of existing algorithms seems like a promising approach for designing algorithms with good empirical properties. \awesome\ and \meta\ can both be seen as portfolio algorithms: they switch between different components based on the opponent's behavior. Much remains to be learned about the best framework for building portfolio algorithms, especially if we insist on frameworks that do not require hand-construction of a portfolio. Again, this direction of research invites a host of empirical questions. What features of a game and of game play should a portfolio track? In what situations does adding an algorithm to a portfolio improve performance?



\section*{Acknowledgements}

Thanks to Nando de Freitas for his involvement in the early stages of this project, and for helping us to develop the argument in Appendix A. Thanks to Holger H. Hoos and Yevgeniy Vorobeychik for feedback on drafts, and to anonymous reviewers at MLJ for helpful suggestions, and to David Ludgate for help with coding. Finally, thanks to Vincent Conitzer for providing us with code for \awesome{} and helping us to use it.
\note{Other acks? Edit as you see fit...}

\section*{Appendix A.~~~Independent vs. Stratified Sampling}
\label{app:stratify}
For all of the experiments described in this article, we were concerned with the expected performance of a match, denoted by $f(\mu,\zeta)$. Here, $f$ is some metric function, $\mu \sim M$ is a match, and $\zeta \sim Z$ is a random seed that completely determined any non-deterministic behavior in both algorithms. The game instance/seed pairing uniquely define a run.
When designing our experiment, we needed to choose whether to stratify runs based on the match. For instance, if we had enough time to run $100$ simulations, we could either have sampled a single run on $100$ matches, or $10$ runs on $10$ matches. Stratification clearly yields more detailed data about the role that randomization plays in each match. However, for estimating common summary statistics---means and quantiles---stratification should be avoided.

Formally, consider two schemes of sampling from $M$ and $Z$. Under \emph{independent sampling}, $M$ and $Z$ are sampled separately each time, yielding a set of samples $\left\{(M_1,Z_1),\ldots,(M_n,Z_n)\right\}$. Under \emph{stratified sampling}, $k$ samples are taken from $M$ and for each sample of $M$, $Z$ is sampled $s_i$ times, yielding a set of samples $\left\{(M_1,Z_{1,1}),\ldots,(M_1,Z_{1,s_1}), \ldots,(M_k,Z_{k,s_k})\right\}$.
In both schemes, the sample mean is used as an estimate for the population mean. Since $M$ and $Z$ are sampled independently, both schemes yield unbiased estimators. However, the following result shows that the schemes differ in terms of variance.

\begin{lem}
Independent sampling yields a lower-variance estimate of $\E{f_{(M,Z)}}$ than stratified sampling.
\label{lem:single_run}
\end{lem}

\begin{proof}
First, independent random variables have no covariance.
\begin{equation}
	\cov{f(M_i,Z_i), f(M_j,Z_j)} = \cov{f(M_k,Z_{k,l}), f(M_m,Z_{m,n})}\label{eq:same_cov}	
\end{equation}
On the other hand, if two samples share the same stratum (the same sample $\mu \sim M$) then they have weakly higher covariance.
\begin{equation}
	\cov{f(M_k,Z_{k,l}), f(M_k,Z_{k,m})} \geq \cov{f(M_i,Z_i), f(M_j,Z_j)} \label{eq:higher_cov}
\end{equation}
Using Equations~\eqref{eq:same_cov} and \eqref{eq:higher_cov} we can write
\begin{align*}
	\var{\sum_i f(M_i,Z_i)} &= \sum_{i,j} \cov{f(M_i,Z_i), f(M_j,Z_j)}&\\
	&\leq \sum_{k,l,m,n} \cov{f(M_k,Z_{k,l}), f(M_m,Z_{m,n})}&\\
	&= \var{\sum_{k,l} f(M_k,Z_{k,l})}.
&\Box
\end{align*}
\end{proof}

We also claimed that stratifying increases the variance of quantile point estimation. This result can be found (albeit without proof) in \inlinecite{heidelberger84}. 

\bibliographystyle{klunamed}
\bibliography{thesis}

\end{document}